\documentclass{article}
\usepackage{hyperref, graphicx}
\usepackage{amsmath}
\usepackage{tocloft}
\usepackage{amssymb}
\usepackage{slashed} 
\usepackage{yfonts}
\usepackage[T1]{fontenc}
\usepackage[utf8]{inputenc}
\usepackage[english]{babel}
\usepackage{amsthm}
\usepackage{tikz-cd}
\usepackage{quiver}
\usepackage{mathabx}
\usepackage{stmaryrd}
\usepackage{braket}
\usepackage{physics}
\usepackage{geometry}
\usepackage[
backend=biber,
style=numeric-comp,
sorting=none,
]{biblatex}
\newtheorem{theorem}{Theorem}[section]
\newtheorem{lemma}{Lemma}[section]
\newtheorem{corollary}{Corollary}[section]
\newtheorem{proposition}{Proposition}[section]

\usepackage{authblk}
    \allowdisplaybreaks

\hypersetup{
pdfstartview = {FitH},
}
\hypersetup{
	colorlinks=true,         
	linkcolor=blue,          
	citecolor=blue,        
	urlcolor=blue            
}

\DeclareMathAlphabet{\mathpzc}{OT1}{pzc}{m}{it}

\theoremstyle{remark}
\newtheorem{remark}{Remark}[section]

\theoremstyle{definition}
\newtheorem{definition}{Definition}[section]

\usepackage{mathtools}

\usepackage{mathrsfs}

\numberwithin{equation}{section}

\newcommand{\beq}{\begin{eqnarray}}
\newcommand{\eeq}{\end{eqnarray}}

\def\g{\mathfrak g}

\def\h{\mathfrak h}
\def\l{\mathfrak l}
\def\L{\mathscr{L}}

\def\G{\mathfrak{G}}
\def\D{\mathfrak{D}}
\def\t{\mathfrak{t}}

\def\R{\mathbb{R}}
\def\bbC{\mathbb{C}}

\def\rd{\mathrm{d}}

\newcommand{\cA}{{\cal A}}

\newcommand{\cF}{{\cal F}}

\newcommand{\cG}{{\cal G}}

\newcommand{\sfx}{{\mathsf{x}}}
\newcommand{\sfy}{{\mathsf{y}}}

\makeatletter
\newsavebox{\@brx}
\newcommand{\llangle}[1][]{\savebox{\@brx}{\(\m@th{#1\langle}\)}%
  \mathopen{\copy\@brx\kern-0.5\wd\@brx\usebox{\@brx}}}
\newcommand{\rrangle}[1][]{\savebox{\@brx}{\(\m@th{#1\rangle}\)}%
  \mathclose{\copy\@brx\kern-0.5\wd\@brx\usebox{\@brx}}}
\makeatother

\newcommand{\id}{\operatorname{id}}

\begin{document}

\title{Higher Gauge Theory and Integrability}
\author[1]{{ \sf Hank Chen}\thanks{hank.chen@uwaterloo.ca}}

\author[2]{{\sf Joaquin Liniado}\thanks{jliniado@iflp.unlp.edu.ar}}
\affil[1]{\small Department of Applied Mathematics, University of Waterloo, 200 University Avenue West, Waterloo, Ontario, Canada, N2L 3G1}
\affil[2]{\small Instituto de Física La Plata (IFLP; UNLP-CONICET), Diagonal 113, Casco Urbano, B1900, La Plata, Provincia de Buenos Aires, Argentina}

\maketitle

\begin{abstract}
  In recent years, significant progress has been made in the study of integrable systems from a gauge theoretic perspective. This development originated with the introduction of $4$d Chern-Simons theory with defects, which provided a systematic framework for constructing two-dimensional integrable systems. In this article, we propose a novel approach to studying higher-dimensional integrable models employing techniques from higher category theory. Starting with higher Chern-Simons theory on the $4$-manifold $\mathbb{R}\times Y$, we complexify and compactify the real line to $\mathbb{C}P^1$ and introduce the disorder defect $\omega=z^{-1}\rd z $. This procedure defines a holomorphic five-dimensional variant of higher Chern-Simons theory, which, when endowed with suitable boundary conditions, allows for the localisation to a three-dimensional theory on $Y$. The equations of motion of the resulting model are equivalent to the flatness of a $2$-connection $(L,H)$, that we then use to construct the corresponding higher holonomies. We prove that these are invariants of homotopies relative boundary, which enables the construction of conserved quantities. The latter are labelled by both the categorical characters of a Lie crossed-module and the infinite number of homotopy classes of surfaces relative boundary in $Y$. Moreover, we also demonstrate that the $3$d theory has left and right acting symmetries whose current algebra is given by an infinite dimensional centrally extended affine Lie 2-algebra. Both of these conditions are direct higher homotopy analogues of the properties satisfied by the 2d Wess-Zumino-Witten CFT, which we therefore interpret as facets of integrable structures. 

\end{abstract}

\newpage 

\tableofcontents

\newpage

\section{Introduction}

Integrable models represent a unique domain of exploration where complex systems exhibit a remarkable level of order. This is typically expressed through an infinite number of symmetries, which in turn, allow for the construction of an infinite number of independently conserved charges. The existence of such a large number of conserved quantities imposes stringent constraints on the system, which may render it soluble to some extent; it is precisely this solvable character what makes integrable systems so special. 

However, the strength of integrable theories is also their greatest weakness: identifying the infinite number of conserved charges is at the same time, the most complicated challenge. Significant progress was made in this direction with the introduction of the Lax formalism \cite{Lax:1968fm} for two-dimensional models. Indeed, this framework provides a systematic way to identify and construct the conserved charges. Specifically, in the case of two dimensional field theories, one looks for a $\mathfrak{g}^{\mathbb{C}}$-valued \emph{Lax connection} which is on-shell flat and depends meromorphically on a $\mathbb{C}P^1$-valued parameter known as the \emph{spectral parameter}. If such an object is found, then its holonomy along Cauchy slices can be used as a generating functional for the conserved charges, by expanding it in powers of the spectral parameter $z\in \mathbb{C}P^1$.

Similarly, in the context of four-dimensional systems, integrability has also been expressed in terms of the \emph{on-shell} flatness of some connection; the prototypical example being the anti-self-dual Yang-Mills (ASDYM) equations. In fact, it has long been established that two-dimensional integrable models arise as symmetry reductions of the ASDYM equations. Of course, anti-self-duality is not quite a flatness condition on the curvature, but instead, the requirement $F=-\star F$. However, the Penrose-Ward correspondence \cite{Ward:1977ta} identifies solutions to the ASDYM equations on four-dimensional space-time, with holomorphic vector bundles over twistor space. In particular, for a vector bundle to be holomorphic, the corresponding connection must be flat.

Despite the inherent elegance of these geometric characterizations of integrability, both suffer from the same issue: although they provide a systematic procedure for constructing the conserved charges once the on-shell flat connection is found, they offer no instruction whatsoever on how to find the connection in the first place.  

In 2013, Costello conceives a beautifully unifying framework to address this problem, originally, for the case of two-dimensional models\footnote{To be precise, the first series of papers dealt with discrete integrable systems in two-dimensions, whereas only the last article of the series discusses two-dimensional field theories. In particular, the latter is the one we will be interested in.} \cite{Costello:2013zra, Costello:2013sla}, which was further refined in a series of seminal papers written in collaboration with Witten and Yamazaki \cite{Costello:2017dso, Costello:2018gyb, Costello:2019tri}. Loosely speaking, one starts with three-dimensional Chern-Simons (CS$_3$) theory, whose fundamental field $A$ is a gauge connection which is \emph{on-shell} flat: at the very least, a tempting candidate to feature in this construction. There are nonetheless, two immediate concerns. First, we are looking for a Lax connection which is defined over a $2$-fold, whereas the gauge field $A$ lives in a three dimensional manifold. This is not \emph{really} a problem, as we can always gauge fix one of the components of the connection to zero. The real issue, is that we want the fields of the two-dimensional theory to be sections of a $G$-bundle over spacetime, not the connection itself! 

To resolve these issues, Costello complexifies and compactifies one of the real directions of the $3$-fold where CS$_3$ is defined, to obtain a $4$-fold $X=\mathbb{C}P^1 \times \Sigma$. Furthermore, he introduces a disorder operator $\omega$ which is the key object that defines $4$d Chern-Simons theory (CS$_4$), whose action is given by
\begin{equation}
    S_{\mathrm{CS}_4}=\int_X \omega \wedge \langle A, \mathrm{d}A+\tfrac{2}{3}A\wedge A\rangle \,.
\end{equation}
The disorder operator $\omega$ is a meromorphic $(1,0)$-form with poles (and eventually zeros) on $\mathbb{C}P^1$. These singularities, which are punctures on the Riemmann sphere, act as a boundary of the form $\{\text{poles of }\omega\}\times \Sigma$ for the $4$-fold $X$. Conceptually, by introducing a boundary on the theory we are partially breaking gauge symmetry. Now gauge transformations identify field configurations which are physically indistinguishable; said differently, they kill \emph{would-be} degrees of freedom. Thus, breaking the gauge symmetry of the theory will \emph{resurrect} \cite{Tong:2016kpv} these would-be degrees of freedom exactly where the symmetry was broken, namely, \emph{on the boundary} $\partial X = \{\text{poles of }\omega\}\times \Sigma$. These boundary degrees of freedom which we will refer to as \emph{edge-modes} are the fields of the two-dimensional theory. Crucially, the bulk of the theory is unaltered, so the flatness of the gauge field $A$, which will then become the Lax connection of the boundary theory, will be implied by construction, and with this, the integrability of the theory. 

Four-dimensional Chern-Simons theory has been an incredibly successful program since its conception. Most of the previously known two-dimensional integrable field theories (IFTs) have been constructed from this perspective, together with new models as well \cite{Delduc:2019whp, Costello:2020lpi,Fukushima:2020kta,Fukushima:2020tqv,Fukushima:2021ako, Ashwinkumar:2020gxt, Caudrelier:2020xtn, Berkovits:2024reg}. See also \cite{Lacroix:2021iit} for a review of its original formulation. Moreover, alternative approaches have been constructed in order to include more general choices of the meromorphic $(1,0)$-form $\omega$ using techniques from homotopical methods \cite{Benini_2019,Lacroix:2021iit, Liniado:2023uoo}. In addition, the framework has also been extended to the case in which $\mathbb{C}P^1$ is replaced with higher genus Riemmann surfaces \cite{Lacroix:2023qlz, Lacroix:2024wrd}. 

In an interesting turn of events, inspired by an idea suggested in a seminar by Costello, Bittleston and Skinner show that CS$_4$ can be described from an even more general perspective, starting with six-dimensional holomorphic Chern-Simons theory (hCS$_6$) on twistor space \cite{Bittleston:2020hfv}. Indeed, holomorphic Chern-Simons theory had first been considered in \cite{Witten:1992fb} to describe the open string sector of the type B topological string. In this context, the target space must be Calabi-Yau, which ensures the existence of a globally defined holomorphic $(3,0)$-form. Twistor space however, is not Calabi-Yau and it therefore does not admit a trivial canonical bundle. Costello's suggestion to overcome this problem was to consider a meromorphic $(3,0)$-form rather than a holomorphic one; in other words, introducing disorder defects. 

This led to a remarkable generalization of CS$_4$ theory because not only it provided a systematic way to construct four-dimensional integrable field theories (in the sense of ASDYM) but it also managed to include the symmetry reduction of ASDYM to 2d integrable models in a unique coherent scheme. More precisely, Bittleston and Skinner introduce a diamond of correspondence of theories as the one shown in figure \ref{fig:diamond}.
\begin{figure}[ht]
\label{fig:diamond}
\centering
\begin{tikzpicture}
\node at (0,2) {hCS$_6$ on $\mathbb{C}P^1\times \mathbb{R}^4$};
\node at (-2,0) {CS$_4$ on $\mathbb{C}P^1\times \mathbb{R}^2$};
\node at (2,0) {IFT$_4$ on $\mathbb{R}^4$};
\node at (0,-2) {IFT$_2$ on $\mathbb{R}^2$};
\draw[->,thick,decorate, decoration={snake, segment length=12pt, amplitude=2pt}] (-0.4,1.6)--(-1.6,0.4);
\draw[->,thick] (0.4,1.6)--(1.6,0.4);
\draw[->,thick,decorate, decoration={snake, segment length=12pt, amplitude=2pt}] (1.6,-0.4)--(0.4,-1.6);
\draw[->, thick] (-1.6,-0.4)--(-0.4,-1.6);
\end{tikzpicture}
\caption{Diamond of correspondence of theories}
\end{figure}
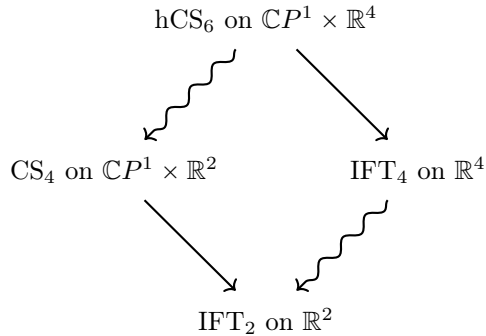

The straight lines represent integration along $\mathbb{C}P^1$. That is, one can morally do the same procedure to go from hCS$_6$ on (Euclidean) twistor space $\mathbb{PT}\cong \mathbb{C}P^1\times \mathbb{R}^4$ to an IFT on $\mathbb{R}^4$ than the one introduced in \cite{Costello:2019tri} to go from CS$_4$ on $\mathbb{C}P^1\times \mathbb{R}^2$ to an IFT on $\mathbb{R}^2$. The wiggly lines, in turn, represent symmetry reduction, which schematically consists of quotienting out a copy of $\mathbb{R}^2$ from the copy of $\mathbb{R}^4$. This was formally implemented in \cite{Bittleston:2020hfv, Penna:2020uky} for various choices of meromorphic $(3,0)$-forms, and it was shown that following both sides of the diamond leads to the same 2d IFT. This work was further generalized in \cite{Cole:2023umd, Cole:2024ess, Cole:2024noh}.

It is clear that some of the fundamental features of integrable structures appearing in field theories can be understood from a gauge theoretic perspective. So far, the strategy to study integrable models in higher dimensions from this perspective, has been to identify the structural similarities between hCS$_6$ on $\mathbb{PT}$ and CS$_4$ on $\mathbb{C}P^1\times \mathbb{R}^2$ and exploit them.

\medskip

In this article we propose a different alternative, based on the "categorical ladder = dimensional ladder" proposal \cite{Crane:1994ty, Baez:1995xq}. It states that higher-dimensional physics can be described by higher categorical structures, and that one can "climb" the dimensions by categorification. In brief, a category can be thought of as a collection of objects together with morphisms on them, such that certain "coherence conditions" are satisfied. Category theory itself can thus be understood as the study of \textit{structure} and the relations between them; for a comprehensive overview of category theory, one can consult \cite{maclane:71}. \textit{Categorification} is then, abstractly, a way to impart relations between structures in a coherent manner. Applying this idea to categories themselves gives rise to the notion of higher categories, which consist of objects, relations between these objects, relations between these relations, and so on. Indeed, as one climbs this categorical ladder, the structures that appear are suited, in each step, for describing higher dimensional data.

As expected, the theory of higher categories can become extremely complicated and abstract very rapidly (see eg. \cite{Lurie:2009}). Nevertheless, the idea of applying such intricate higher categorical tools to study higher-dimensional physics has been very popular in recent years \cite{Gaiotto:2014kfa,Kapustin:2013uxa,Wen:2019,Zhu:2019,Kong:2020wmn,Dubinkin:2020kxo,chen:2022,Delcamp:2023kew,Chen2z:2023}, and has led to many fruitful classification results. Examples include, but are not limited to, the study of higher-dimensional gapped topological phases in condensed matter theory \cite{Johnson-Freyd:2020,KongTianZhou:2020}, as well as the study of anomalies in quantum field theory (QFT) \cite{Freed:2014,Cordova:2018cvg,Benini_2019}. The striking power of functorial topological QFTs (TQFTs) since Atiyah-Segal \cite{Atiyah:1988}, in particular, to produce novel topological invariants \cite{lurie2008classification,Tillmann_2004,Freed_2021,Reutter:2020bav,Douglas:2018} from abstract categorical data cannot be overstated.

Higher category theory has also been used to study gauge theories in higher dimensions, leading to the notion of \emph{higher gauge theories}. These comprise a very rich and intricate system of gauge principles governed by categorical structures \cite{Baez:2002highergauge,Baez:2004in,Wockel2008Principal2A,Kim:2019owc}, which give rise to observables that are sensitive to the topology and geometry of surfaces, in analogy to the Wilson lines in the usual 3d Chern-Simons setup \cite{Reshetikhin:1991tc,Turaev:1992hq,Witten:1988hc,Guadagnini2008DeligneBeilinsonCA}. The existence of such theories, invites the following question:

\medskip

\begin{quote}
\centering
{\em What is the interplay between higher-gauge theories and higher-dimensional integrable systems?} 
\end{quote}

Our paper is dedicated to answering this question in a deep manner. Specifically, we will consider higher Chern-Simons theory based on a Lie 2-group/Lie group crossed module and its corresponding Lie 2-algebra \cite{Song_2023, Zucchini:2021bnn}. Schematically speaking, every object appearing in regular gauge theory is replaced by its corresponding higher \emph{counterpart}: the Lie group $G$ is replaced by a $2$-group $\mathsf{H}\rightarrow G$, the Lie algebra $\mathfrak{g}$ by a Lie 2-algebra $\h\rightarrow \g$, the connection $A$ by a $2$-connection $(A,B)$, and so on and so forth.  Crucially, introducing these higher categorical objects \emph{requires} increasing the dimension of the manifold where the theory is defined. Indeed, 2-Chern-Simons theory is defined over a four dimensional manifold $X$ and it is constructed so that its equations of motion correspond to the (higher) flatness of the 2-connection $(A,B)$ (see \S \ref{sec:higherchernsimons} for details). Our idea, is thus to explore in what sense this categorified flatness condition is related to integrability. 

We thus proceed \emph{a lá} Costello \cite{Costello:2019tri}. We write $X=\mathbb{R}\times Y$ for a $3$-fold $Y$, we complexify and compactify the copy of $\mathbb{R}$, and introduce a disorder operator $\omega$. This procedure defines a five-dimensional theory, which can be localised to a three-dimensional boundary theory on $Y$, with equations of motion equivalent to the higher flatness of the 2-connection. In this article we will focus on a specific choice of meromorphic $(1,0)$-form given by $\omega=z^{-1}\rd z$, inspired by the resulting 2d theory corresponding to this choice of $\omega$ in the context of CS$_4$: the 2d Wess-Zumino-Witten (WZW) model.

\subsection{Summary of results}
We explicitly derive the localized 3d theory on $Y$, whose fields are given by a smooth function $g \in C^{\infty}(Y)\otimes G$ and a $1$-form $\Theta \in \Omega^1(Y)\times \h$ parameterizing the {broken} gauge symmetries at the poles of $\omega$. We prove in {\bf Proposition \ref{extended2flat}} that the equations of motion of the theory  give rise to a flat $2$-connection $(L,H)$. {\bf Theorem \ref{descending}} states that the remaining, unbroken symmetries of the 5d bulk 2-Chern-Simons theory on $X=Y\times\mathbb{C}P^1$ descend to residual global symmetries of the 3d theory on $Y$, in complete analogy with the CS$_4$ -- IFT$_2$ story \cite{Costello:2019tri}. We also show in \S \ref{thf} that we recover {\it exactly} the Chern-Simons/matter coupling studied in the context of 3-brane/5-brane couplings in the A-model \cite{Aganagic:2017tvx}.

By making use of higher groupoids, we prove in {\bf Theorem \ref{relbdy}} that the higher holonomies $(V,W)$ arising in our 3d theory are invariants of {\it homotopies relative boundary}, and that they are consistent with the Eckmann-Hilton argument \cite{Gaiotto:2014kfa,Dubinkin:2020kxo}. This can be understood as an improved version of the "surface independence" notion discovered in \cite{Alvarez:1997ma}. Conserved quantities can then be obtained from $(V,W)$, labelled by \textit{categorical characters} \cite{Ganter:2006,Ganter:2014,davydov2017lagrangian,Sean:private,Huang:2024} of the Lie 2-group $\mathbb{G}$ and homotopy classes of surfaces relative boundary in $Y$. If certain technical conditions are met, the above results are fortified in {\bf Theorem \ref{bordinv}} to show that the higher holonomies in fact define bordism invariants. These can be understood as a "global" version of the statement of conservation for the currents $J =(L,H)$; a more detailed explanation (with images) is given in \S \ref{holonomies}.

Finally, by analyzing the conserved Noether charges associated to the (infinitesimal) symmetries of the 3d action, we prove in {\bf Theorem \ref{gradedbrackets}} that the currents $(L,H)$ form an infinite dimensinoal Lie 2-algebra. Moreover, by leveraging a transverse holomorphic foliation (THF) on $Y$, {\bf Theorem \ref{affine2alg}} further illustrates that they form a {\it Lie 2-algebra extension} \cite{Angulo:2018}; we call these current algebras the \textbf{affine Lie 2-algebras of planar currents}. The existence of such an infinite dimensional symmetry algebra of our three-dimensional model is the origin of its integrable structure. 

It is worth emphasizing that both the existence of a spectral-parameter independent on-shell flat $2$-connection, and every result we have proven in \S \ref{dgaffinecurrents}, are direct higher homotopy analogues of the properties of the 2d WZW model: the anti/chiral anti/holomorphic currents satisfy a centrally extended current algebra, which is the underlying affine Lie algebra of the 2d WZW model \cite{KNIZHNIK198483}. 


\subsection{Organization of the paper}

The paper is organized as follows. In section \S \ref{sec:higherchernsimons} we introduce the concepts of a Lie group crossed-module and it's corresponding Lie 2-algebra. This will allow us to define a homotopy Maurer-Cartan theory, of which 2-Chern-Simons theory is a particular example. We proceed with the introduction of the 2-Chern-Simons action and describe some of its properties. In \S \ref{sec:h2csth} we define \emph{holomorphic} 2-Chern-Simons theory by introducing the disorder operator $\omega$. We present in full detail the procedure of localisation to a three-dimensional boundary theory. In \S \ref{3dift} we study the family of three-dimensional actions obtained in \S \ref{sec:h2csth} and in particular, we show that the equations of motion are equivalent to the higher flatness of the $2$-connection $(L,H)$. In \S \ref{residuals} we analyze the symmetries of these theories in depth. In \S \ref{sec:holonomies} we discuss the integrable properties associated to the flatness of the $2$-connection $(L,H)$. Specifically , we show that their $2$-holonomies are conserved and homotopy invariant. Finally, in section \S \ref{dgaffinecurrents} we describe the current algebra associated to the symmetries of the $3$d theory studied in \S \ref{residuals}. This is a graded affine Lie $2$-algebra with a central extension, which we interpret as the higher version of the affine Lie algebra that appears in 2d WZW. We conclude the article with an outline of future interesting directions to explore.

\medskip

\subsubsection*{Acknowledgements} The authors would like to thank Charles Young, Kevin Costello, Roland Bittleston, Lewis Cole and Benoit Vicedo for very insightful discussions. JL would like to thank Horacio Falomir for his endless support as a PhD advisor. HC would like to thank Florian Girelli, Nicolas Cresto and Christopher Andrew James Pollack for valuable discussions. The work of JL is supported by CONICET. We would also like to thank the anonymous reviewer for the very insightful comments and suggestions.

\subsubsection*{Note of the Authors} A few weeks ago, a pre-print \cite{Schenkel:2024dcd} was published in the arXiv in which 2-Chern-Simons theory with disorder defects is studied. While it is fortunate that the scopes of our articles are quite different, it is important to clarify that our work is not a follow-up to their publication. Rather, we independently conceived the same idea, performed our works independently, and they happened to publish their findings first. Our research and conclusions were developed without knowledge of their article, underscoring the originality and integrity of our work.

\section{Higher Chern-Simons Theory}
\label{sec:higherchernsimons}


We will define higher Chern-Simons theory from the perspective of \emph{homotopy Maurer-Cartan} theories. As we will see, these theories are naturally described using higher structures, yet they all share a common ingredient: their equations of motion are flatness of some (eventually higher) connection. In this section we introduce the necessary mathematical background and definitions, together with the 2-Chern-Simons action and some of its properties.

The notion of a higher principal bundle with connection has been developed recently in the literature \cite{Baez:2004in,Martins:2006hx,chen:2022,Porst2008Strict2A}. The kinematical data attached to such higher-gauge theories is captured by the structure of \textit{homotopy $L_\infty$-algebras}, which are graded vector spaces equipped with higher $n$-nary Lie brackets \cite{Kim:2019owc,Baez:2004} satisfying Koszul identities \cite{Bai_2013}. In this paper, we will focus our attention on $L_2$-algebras/strict Lie 2-algebras and their corresponding Lie 2-groups, which we now define. 

\begin{definition}\label{2grpdef}
A {\bf (Lie) group crossed-module} $\mathbb{G}=(\mathsf{H}\xrightarrow{\bar\mu_1}G,\rhd)$ consists of two (Lie) groups $\mathsf{H},G$, a (Lie) group homomorphism $\bar\mu_1:\mathsf{H}\rightarrow G$ and an (smooth) action $\rhd:G\rightarrow \operatorname{Aut}\mathsf{H}$ such that the following conditions
\begin{equation}
    \bar\mu_1(x\rhd y) = x\bar\mu_1(y)x^{-1}\,,\qquad (\bar\mu_1(y))\rhd y'=yy'y^{-1}\label{pfeif2}
\end{equation}
are satisfied for each $x\in G$ and $y,y'\in \mathsf{H}$.
\end{definition}
The infinitesimal approximation of a Lie 2-group gives rise to a Lie 2-algebra, which is a strict 2-term $L_\infty$-algebra \cite{Baez:2003fs}. More precisely, we have

\begin{definition}\label{lie2alg}
Let $\mathbb{K}$ denote a field of characteristic zero (such as $\R$ or $\bbC$). A {\bf $L_2$-algebra} $\mathfrak{G}=\h\rightarrow \g$ over $\mathbb{K}$ consist of two Lie algebras $\big(\h,[-,-]_\h\big)$ and $\big(\g,[-,-]\big)$ over $\mathbb{K}$ and the tuple of maps,
\begin{equation*}
    \mu_1: \h\rightarrow\g,\qquad \mu_2: \g\wedge\h \rightarrow \h\, ,\footnotemark
\end{equation*}
\footnotetext{Here $\wedge$ denotes the skew-symmetric tensor product and $\odot$ denotes the symmetric tensor product.} subject to the following conditions for each $\sfx, \sfx',\sfx'' \in\g$ and $\sfy,\sfy'\in\h$:
\begin{enumerate}
    \item The $\g$-equivariance of $\mu_1$ and the Peiffer identity,
    \begin{equation}
        \mu_1(\mu_2(\sfx,\sfy))=[\sfx,\mu_1(\sfy)]\,,\qquad \mu_2(\mu_1(\sfy),\sfy') =  [\sfy,\sfy']_\h =- \mu_2(\mu_1(\sfy'),\sfy)\,.\label{pfeif1}
\end{equation}
\noindent Note $\operatorname{ker}\mu_1\subset \h$ is an Abelian ideal due to the Peiffer identity \eqref{pfeif1}.

\item Graded Jacobi identities,
\begin{align}
\label{ec:gradjacobid}
        0&=[\sfx,[\sfx',\sfx'']]+[\sfx',[\sfx'',\sfx]]+[\sfx'',[\sfx,\sfx']], \\
        0&= \mu_2(\sfx,\mu_2(\sfx',\sfy)) - \mu_2(\sfx',\mu_2(\sfx,\sfy)) - \mu_2([\sfx,\sfx'],\sfy)\,.
    \end{align}
\end{enumerate}
Moreover, we call $\G$ {\bf balanced} \cite{Zucchini:2021bnn} if it is equipped with a graded symmetric non-degenerate bilinear form $\langle -,-\rangle: \g\odot\h \rightarrow \mathbb{K}$ which is invariant
\begin{equation}
    \langle \sfx,\mu_2(\sfx',\sfy)\rangle = \langle [\sfx,\sfx'],\sfy\rangle\,,\qquad \langle \mu_1(\sfy),\sfy'\rangle = \langle \mu_1(\sfy'),\sfy\rangle\label{inv}
\end{equation}
for each $\sfx,\sfx'\in\g$ and $\sfy,\sfy'\in \h$.
\end{definition}
\noindent The unary and binary brackets $\mu_1,\mu_2$ are integrated respectively to a Lie group map $\bar\mu_1:\mathsf{H}\rightarrow G$ and an action $G\rightarrow\operatorname{Der}\h$ of $G$ on the degree-(-1) Lie algebra $\h$. As an abuse of notation, we shall denote this action also by $\rhd$.\footnote{This action $\rhd:G\rightarrow\operatorname{Der}\h$ is part of the 2-adjoint action of $\mathbb{G}$ on its own Lie 2-algebra $\G=\operatorname{Lie}\mathbb{G}$ \cite{Bai_2013,Angulo:2018,chen:2022}.}

Based on the structure of a balanced Lie 2-algebra $\big(\G,\mu_1,\mu_2,\langle -,-\rangle\big)$, we construct 2-Chern-Simons theory as a homotopy Maurer-Cartan theory from the Batalin–Vilkovisky (BV) formalism \cite{Jurco:2018sby}. We will also describe the relation with the derived superfield formalism considered in \cite{Zucchini:2021bnn}. 

\subsection{Homotopy Maurer-Cartan Theory}\label{hMC}

Three-dimensional Chern-Simons theory can be described as the simplest example of what is known as a \emph{homotopy Maurer-Cartan theory}. These were first introduced in \cite{Zwiebach:1992ie} in the context of string field theory and are structured upon $L_\infty$-algebras. Indeed, CS$_3$ is a homotopy Maurer-Cartan theory for an $L_1$-algebra, that is, a plain Lie algebra. We will therefore define 2-Chern-Simons theory by considering the simplest graded structure within the $L_\infty$-algebra, namely, an $L_2$-algebra. 

We thus begin by turning the Lie $2$-algebra $\G=\h\rightarrow\g$ into a differential graded (dg) algebra by tensoring $\G$ with the de Rham complex $\Omega^\ast(X)$ over a space $X$. This gives rise to a Lie 2-algebra with the graded components
\begin{equation*}
    \L_n = (\Omega^\bullet(X)\otimes\G)_n = \bigoplus_{i+j=n} \Omega^i(X)\otimes \g_j,\qquad \g_{-1}=\h,~\g_0=\g,
\end{equation*}
together with the differential $\ell_1 = \mathrm{d}-\mu_1$ and $\ell_n = \mu_n\otimes\wedge^n$ for all $n\geq 2$.

\begin{definition}

An element $\mathcal{A}\in\L_1$ is a {\bf Maurer-Cartan element} if its curvature vanishes, 
\begin{equation}
\label{ec:curvature}
   \mathscr{F}(\cA) = \sum_{n=1}^2 \frac{1}{n!}\ell_n(\mathcal{A},\dots,\mathcal{A})=0 \,.
\end{equation}
We denote the space of Maurer-Cartan elements by $\mathsf{MC}\subset\Omega^\bullet(X)\otimes \G$.

\end{definition}

Let us see what this means by writing out these objects explicitly. An element $\cA \in \L_1$, by definition, is of the form $\mathcal{A}=(A,B)$ where $A\in \Omega^1(X)\otimes\g$ and $B \in \Omega^2(X)\otimes\h$. The curvature \eqref{ec:curvature} is thus
\begin{align*}
   \mathscr{F}(\mathcal{A})=  \ell_1(\mathcal{A}) + \frac{1}{2}\ell_2(\mathcal{A},\mathcal{A})&= \rd A - \mu_1(B) + \frac{1}{2}[A, A]  + \rd B +\mu_2(A, B)\,.
\end{align*}
Organizing this quantity by degree, we see that we obtain two equations
\begin{equation}
\label{ec:2flatnessss}
    \mathcal{F}=\rd A + \tfrac{1}{2}[A,A] - \mu_1(B) =0\,,\quad \mathcal{G}=\rd B + \mu_2(A,B) =0\,,
\end{equation}
where $\cF$ and $\cG$ are known as the {\bf fake-curvature} and {\bf 2-curvature} respectively. Equation \eqref{ec:2flatnessss} then corresponds to the {\bf fake-/2-flatness conditions} \cite{Kim:2019owc,Baez:2004in,Martins:2006hx,Chen:2022hct}.

\medskip

\begin{remark}
    In the case of CS$_3$, a degree $1$ element $\mathcal{A} \in \L_1$ is simply a $\g$-valued $1$-form $A$. In particular, it will be a Maurer-Cartan element if it's curvature $F(A) = \mathrm{d}A+\tfrac{1}{2}[A,A]$ vanishes. 
\end{remark}

\medskip

The goal is to define an action whose variational principle is associated to the zero-curvature condition; in other words, the minimal locus of the action consists only of Maurer-Cartan elements. This brings us to the following definition

\begin{definition}

We define the \textbf{2-Maurer-Cartan action} as the action functional

\begin{equation}
\label{ec:2mcaction}
    S_\text{2MC}[\mathcal{A}]\coloneqq \sum_{m=1}^2\frac{1}{(m+1)!}\int_X \langle \mathcal{A},\ell_m(\mathcal{A},\dots,\mathcal{A})\rangle \,.
\end{equation}
In particular, for an arbitrary variation $\delta \cA$ we find that $\delta S_{\mathrm{2MC}}[\cA]=0$ if and only if $\mathscr{F}(\mathcal{A}) =0$. 
    
\end{definition}

\begin{remark}
    Note that $\cA=(A,B) \in \L_1$  consists of a $\mathfrak{g}$-valued $1$-form $A$ and an $\mathfrak{h}$-valued $2$-form $B$. Similarly $\mathscr{F}=(\mathcal{F},\mathcal{G}) \in \L_2$ consists of a $\g$-valued $2$-form $\cF$ and an $\h$-valued $3$-form $\cG$. Since the pairing on the balanced Lie algebra $\mathfrak{G}$ pairs elements of $\g$ with those of $\h$, we conclude that the integrand in \eqref{ec:2mcaction} is a $4$-form, and thus $X$ must be a $4$-manifold. This explains why dg Lie algebras (dgla's), at the very least, are {\it required} to define a 4d analogue of Chern-Simons theory.
\end{remark}
Going back to \eqref{ec:2mcaction}, we make use of the invariance property \eqref{inv} to bring the 2-Maurer-Cartan action into the form 
\begin{align}
    S_\text{2CS}[\mathcal{A}] = \int_X \langle \rd A + \tfrac{1}{2}[A, A]- \mu_1(B), B\rangle \,,\label{2chernsimons}
\end{align}
which is what we shall refer to in the rest of the paper as \textbf{2-Chern-Simons theory} (2CS) \cite{Zucchini:2021bnn}. This theory is also called the "4d BF-BB theory" in some literature \cite{Baez:1995ph,Geiller:2022fyg,Girelli:2021zmt}. By construction, the equations of motion of \eqref{2chernsimons} will be precisely $\mathscr{F}(\mathcal{A})=0$. Explicitly, it is simple to check that 
\begin{align*}
    \delta_B S_\text{2CS}[A,B]=0 \,: &\quad \mathcal{F}=\rd A + \tfrac{1}{2}[A,A] - \mu_1(B)=0\,,\\
    \delta_A S_\text{2CS}[A,B]=0\,:&\quad \mathcal{G}=\rd B + \mu_2(A, B) = 0\,.
\end{align*}

\subsection{Gauge Symmetries of 2CS Theory}
\label{sec:gsof2CS}

As in any gauge theory, we will be interested in gauge symmetries of the action. These are given by a higher analogue of usual gauge transformations and are characterized by a smooth map $h \in C^{\infty}(X,G)$ and a $1$-form $\Gamma \in \Omega^{1}(X)\otimes \mathfrak{h}$ which act on the gauge fields as \cite{Baez:2004in}
\begin{equation}
\begin{aligned}
\label{ec:2gaugetr1}
    &A \mapsto A^{(h,\Gamma)}=\operatorname{Ad}_h^{-1}A + h^{-1}\mathrm{d}h + \mu_1(\Gamma) \\
    &B \mapsto B^{(h,\Gamma)}= h^{-1}\rhd B + \mathrm{d}\Gamma + \mu_2(A^{(h,\Gamma)}, \Gamma) - \tfrac{1}{2}[\Gamma,\Gamma] \,.    
\end{aligned}
\end{equation}
We call the simultaneous transformations of the gauge fields  a {\bf $2$-gauge transformation}. Note that the above expressions correspond to finite gauge transformations of the gauge fields, which in this article will play a crucial role. In the literature people often consider only infinitesimal gauge transformations, which can be obtained by a formal expansion $h \sim 1+\alpha$ and neglecting terms quadratic in $\alpha$ and $\Gamma$, see for instance \textbf{Prop. 2.9} of \cite{Baez:2004in}. With a simple computation, it can be shown that under these gauge transformations, the fake-curvature and the $2$-curvature transform covariantly, 
\begin{align*}
    \mathcal{F}\big(A^{(h,\Gamma)},B^{(h,\Gamma)}\big) = \operatorname{Ad}_h^{-1}\mathcal{F}(A,B)\,,\qquad \mathcal{G}\big(A^{(h,\Gamma)},B^{(h,\Gamma)}\big) = h^{-1}\rhd \mathcal{G}(A,B) + \mu_2(\mathcal{F}(A,B), \Gamma) \,,
\end{align*}
whence the transformations \eqref{ec:2gaugetr1} leave the subspace $\mathsf{MC}\subset\Omega^\bullet(X)\otimes\G$ of Maurer-Cartan elements invariant.

\begin{remark}\label{homtrans}
Notice \eqref{ec:2gaugetr1} includes a class of transformations that translates the 1-form field. As such, the value of $A$ is only defined modulo the image of $\mu_1$ up to 2-gauge redundancy. However, to properly treat this seemingly large gauge freedom in the Batalin–Vilkovisky (BV) formalism, one would have to make use of {\it homotopy transfer}. Given the BV(-BRST) $L_\infty$-algebra $\L$ encoding all of the fields in the theory, it was argued in \cite{Arvanitakis:2020rrk} that integrating out degrees of freedom is equivalent to projecting onto its cohomology $H(\L)$. If $H(\L)\neq0$ is not trivial, then the homotopy transfer theorem states that higher homotopy brackets --- in particular a trilinear tertiary "homotopy map" $\tilde\ell_3$ \cite{Chen:2012gz} --- will be induced on $H(\L)$ through this projection \cite{stasheff2018,Arvanitakis:2020rrk}.

We shall explain the mathematics behind homotopy transfer in \S \ref{sec:homtransf}. But the upshot is that integrating out/gauging this shift symmetry $A\rightarrow A+\mu_1\Gamma$ away will in general lead to a non-trivial deformation of the theory, and the 2-Chern-Simons theory will end up differing perturbatively from the 4d BF theory. Though if $\mu_1$ \textit{were} invertible, then $A$ can indeed be completely gauged away. In this case, the cohomology of the BV-BRST complex $\L$ would be trivial, hence performing a homotopy transfer to project away the field $A$ leads to a trivial theory, at least perturbatively. In this case, the 2CS theory becomes the so-called 4d BF-BB theory, which is conjectured to coincide with the Crane-Yetter-Broda TQFT \cite{Baez:1995ph} that hosts only non-perturbative degrees-of-freedom. It is sensitive only to the ($G$-)bordism class of the underlying 4-manifold. 
\end{remark}

We are interested in the variation of 2CS theory under a gauge transformation, which we describe in the following proposition.

\begin{proposition}
\label{prop:gvar}
Under the 2-gauge transformations \eqref{ec:2gaugetr1} the 2CS action transforms as
\begin{equation}
\label{ec:actionundergt}
\begin{aligned}
S_{\mathrm{2CS}}\big[A^{(h,\Gamma)},B^{(h,\Gamma)}\big] &= S_{\mathrm{2CS}}[A,B] \\
&\hspace{.7cm}+\int_X\mathrm{d}\big[\langle \mathrm{Ad}_h^{-1}F(A) ,\Gamma\rangle + \tfrac{1}{2} \langle\mathrm{Ad}_h^{-1}A+h^{-1}\mathrm{d}h,[\Gamma,\Gamma]\rangle + \tfrac{1}{2}L_{\mathrm{CS}}(\Gamma)\big]    
\end{aligned}
\end{equation}
where
\begin{equation}
\label{ec:basicCS}
  L_{\mathrm{CS}}(\Gamma) = \langle\mu_1(\Gamma),\rd\Gamma+\tfrac{1}{3}[\Gamma,\Gamma]\rangle\,. 
\end{equation}
  
\end{proposition}
\begin{proof}
Let us begin by simplifying the expressions that appear in the Lagrangian $L_{\mathrm{2CS}}\big[A^{(h,\Gamma)},B^{(h,\Gamma)}\big]$ conveniently. First, we note that
\begin{equation}
    F\big(A^{(h,\Gamma)}\big) = F(A^h)+F(\Gamma)+[A^h,\mu_1(\Gamma)]
\end{equation}
where $A^h = h^{-1}Ah+h^{-1}\rd h$. On the other hand, we can write
\begin{equation}
    B^{(h,\Gamma)}=h^{-1}\rhd B + F(\Gamma)+\mu_2(A^h,\Gamma)\,.
\end{equation}
Thus, we find that
\begin{equation}
    F\big(A^{(h,\Gamma)}\big)-\tfrac{1}{2}\mu_1\big(B^{(h,\Gamma)}\big) = \mathrm{Ad}_h^{-1}\big[F(A)-\tfrac{1}{2}\mu_1(B)\big] + \tfrac{1}{2}\big(\mu_1(F(\Gamma))+[A^h,\mu_1\Gamma]\big)\,,
\end{equation}
where we have used the covariance $F(A^h)=\mathrm{Ad}_h^{-1}F(A)$ of the usual curvature. Hence we have
\begin{equation}
\begin{aligned}    L_{\mathrm{2CS}}\big[A^{(h,\Gamma)},B^{(h,\Gamma)}\big]= L_{\mathrm{2CS}}[A,B] &+ \langle \mathrm{Ad}_h^{-1}\big[F(A)-\tfrac{1}{2}\mu_1(B)\big], F(\Gamma)+\mu_2(A^h,\Gamma)\rangle \\
& \hspace{1cm}+\tfrac{1}{2}\langle \mathrm{Ad}_h^{-1}\mu_1(B)+ F(\Gamma)+\mu_2(A^h,\Gamma),F(\Gamma)+\mu_2(A^h,\Gamma)\rangle \,.
\end{aligned}
\end{equation}
The terms containing $B$ cancel each other. Expanding the last term and doing some algebraic manipulations we get
\begin{equation}
\label{ec:formula1}
\begin{aligned}
    \tfrac{1}{2}\langle F(\Gamma)+\mu_2(A^h,\Gamma),F(\Gamma) &+\mu_2(A^h,\Gamma)\rangle = \\
    &\tfrac{1}{2}\rd L_{\mathrm{CS}}(\Gamma)+\tfrac{1}{2}\langle [A^h,\mu_1(\Gamma)],\mu_2(A^h,\Gamma)\rangle + \langle \rd \Gamma,\mu_2(A^h,\Gamma)\rangle\,,
\end{aligned}
\end{equation}
where we have used the Jacobi identity to eliminate terms cubic in $\Gamma$ and $L_{\mathrm{CS}}(\Gamma)$ is given in \eqref{ec:basicCS}. The remaining term is after some manipulations,
\begin{equation}
\label{ec:formula2}
\begin{aligned}
    \langle \mathrm{Ad}_h^{-1}F(A),F(\Gamma)+\mu_2(A,\Gamma)\rangle 
    &= \rd\langle A^h,\tfrac{1}{2}[\Gamma,\Gamma]\rangle +  \rd\langle F(A^{h}),\Gamma \rangle -\langle \rd \Gamma,\mu_2(A^h,\Gamma)\rangle   \\
    & \hspace{5.5cm} +\tfrac{1}{4}\langle [A^h,A^h],[\Gamma,\Gamma]\rangle
\end{aligned}
\end{equation}
where we have used the Bianchi identity $\rd F(A^h)+[A^h,F(A^h)]=0$. Putting together \eqref{ec:formula1} and \eqref{ec:formula2} we see that the non-exact terms cancel using the graded Jacobi identity \eqref{ec:gradjacobid}, and we arrive to the desired result.

\end{proof}

    Note that the gauge variation is a total boundary term. In other words, violation to gauge non-invariance of 2-Chern-Simons theory is completely hologrpahic, in contrast with CS$_3$ whose gauge variation contains the well-known Wess-Zumino-Witten term, which must be defined in the $3$d bulk.

\subsubsection{Secondary gauge transformations}
\label{secondarygaugetransf}
A special feature of 2-gauge theory in general is that there are redundancies in the \textit{2-gauge transformations} $(h,\Gamma)$ \eqref{ec:2gaugetr1} itself \cite{chen:2022,Baez:2004,Soncini:2014ara,Kapustin:2013uxa}. This can be attributed to the fact that the gauge symmetries in \eqref{2chernsimons} form in actuality a 2-group, whose objects are the gauge parameters $(h,\Gamma)$.\footnote{This is an example of a categorical {\it gauge} symmetry, which is distinct from the categorical global symmetries that have been discussed recently in various places, eg. \cite{Gaiotto:2014kfa,Delcamp:2023kew}.} We call these \textbf{secondary gauge transformations}, which can be formalized in a categorical context as in \cite{Baez:2004,Baez:2004in,Soncini:2014ara}.

More explicitly, these secondary gauge transformations are in essence deformations/"morphisms" $(h,\Gamma)\rightarrow (h',\Gamma')$ such that $(h',\Gamma')$ remains a 2-gauge transformation as in \eqref{ec:2gaugetr1}, 
\begin{equation}\label{ec:2ndgt}
    (A^{(h,\Gamma)},B^{(h,\Gamma)}) \rightarrow (A^{(h',\Gamma')},B^{(h',\Gamma')}) 
\end{equation}
In full generality, these morphisms are labelled by a pair of fields (the so-called "ghosts-of-ghosts") $\varphi\in C^\infty(X)\otimes\mathsf{H}, E\in\Omega^1(X)\otimes\h.$ Essentially, the 0-form $\varphi$ induces a translation $h\rightarrow \bar\mu_1(\varphi)h$ along the base while the 1-form $E$ induces a translation $\Gamma\rightarrow\Gamma+ E$ along the faces. We leave the interested reader to find the detailed differential constraints these fields must satisfy in the relevant literature \cite{Soncini:2014ara,Baez:2004in}. 

As we are in the context of a {\it strict} 2-gauge theory, it is safe to not make crucial use of these data. For {\it weak/semistrict} 2-gauge theories, on the other hand, it has been noted that they in fact play a key role \cite{Baez:2004,Soncini:2014ara}, even in the case where the 2-group is finite \cite{Kapustin:2013uxa}. Though this requires one to always be on-shell of the fake-flatness condition $F - \mu_1B=0$ \cite{Kim:2019owc}, which is undesirable in our considerations.    

\subsection{The derived superfield formulation}
We pause here momentarily to comment that the above higher-gauge theory computations can be reformulated using the general operational framework developed in \cite{Zucchini:2019pbv,Zucchini:2019rpp,Zucchini:2021bnn}. This formalism makes use of \emph{derived} structures, and provides a direct link with ordinary gauge theory. We will content ourselves with presenting the key expressions, the details can be found in \cite{Zucchini:2021bnn}.

Given the Lie $2$-group $\mathbb{G}$ we define the derived Lie $2$-group $\mathsf{D}\mathbb{G}$ as the set of maps $\R[1]\rightarrow \mathbb{G}$ labelled by the semidirect product $G \ltimes \h$,
\begin{equation}
    \alpha \mapsto (x,e^{\alpha\cdot \sfy})\,, \qquad \alpha\in \R\,,\qquad x\in G,~\sfy\in\h \,.
\end{equation}
In particular, $\mathsf{D}\mathbb{G}$ is a graded Lie group; it shall make appearance again in \S \ref{dgaffinecurrents}. Similarly, one can also define the corresponding derived Lie 2-algebra $\mathsf{D}\G$ as the space of maps $\R[1]\rightarrow \g\ltimes \h$.
We may now describe our higher connection and higher gauge transformations as elements valued in $\mathsf{D}\mathfrak G$ and $\mathsf{D}\mathbb{G}$ respectively. Explicitly, for the $2$-connection we have a degree-$1$ polyform
\begin{equation}
    \cA (\alpha)=(A,\alpha \cdot B) \in (\Omega^\bullet(X)\otimes \mathsf{D}\mathfrak G)_1 \,,
\end{equation}
whereas for the $2$-gauge transformation, a degree-$0$ polyform
\begin{equation}
    U(\alpha)=(h, e^{\alpha\cdot\Gamma}) \in (\Omega^\bullet(X)\otimes \mathsf{D}\mathbb{G})_0 \,.
\end{equation}
Note $\Gamma\in\Omega^1(X)\otimes\h$, and $e^{\alpha\cdot\Gamma} = 1 + \alpha\cdot\Gamma + \frac{1}{2}[\alpha\cdot\Gamma,\alpha\cdot\Gamma]$ terminates due to the Jacobi identity \cite{Baez:2004in}. The tuples $(A,B),(h,\Gamma)$ that we have used previously can therefore be understood respectively as the labels for these derived polyforms $\mathcal{A},U$. 

With these expressions at hand, we may write a general $2$-gauge transformation \eqref{ec:2gaugetr1} in the compact form
\begin{equation}
\label{ec:superfieldgauge}
    \mathcal{A} \to \mathcal{A}^U =~ _2\operatorname{Ad}_U^{-1}\mathcal{A} + U^{-1}\tilde{\mathrm{d}} U\,
\end{equation}
as a function of the derived parameter $\alpha$, where the expressions for the $2$-adjoint action $_2\operatorname{Ad}$  and the Maurer-Cartan form $U^{-1}\tilde{\mathrm{d}}U$ are given, respectively, in eqs. (3.2.8) and (3.2.14) of \cite{Zucchini:2021bnn}, and $\tilde{\mathrm{d}}= \mathrm{d}-\mu_1\tfrac{\partial}{\partial \alpha}$. Finally, we may induce a pairing $(-,-)$ on $\mathsf{D}\mathfrak G$ from the pairing $\langle -,- \rangle$ on $\frak G$ (see eq. (3.2.15) in \cite{Zucchini:2021bnn}) to write the four-dimensional $2$CS action \eqref{2chernsimons} as
\begin{equation}
    S_{2\mathrm{CS}}[\cA] = \int_X (\cA,\tilde{\mathrm{d}}\cA + \tfrac{1}{3}[\cA,\cA]) \,.
\end{equation}
As such, the derived superfield formalism allows us to rewrite expressions in higher gauge theory as analogues of those in ordinary gauge theory.

\medskip

In addition to the compactness of this derived operational framework, it also introduces an \textit{internal} grading attached to the graded field algebra: in contrast to ordinary gauge theory, derived gauge theory has two gradings \cite{Zucchini:2019pbv, Zucchini:2019rpp}.
This additional grading endows both connections and gauge transformations with \emph{ghostlike} partners, in a way similar to the AKSZ formulation of BV theory \cite{Alexandrov:1995kv,Ikeda:2012pv,Calaque:2021sgp}. The derived formalism therefore gives a useful and unifying description of higher gauge theory, also known as "non-Abelian bundle gerbes with connection" \cite{Baez:2004in,Nikolaus2011FOUREV,Schreiber:2013pra}.

\section{Holomorphic 2-Chern-Simons Theory}
\label{sec:h2csth}

Having introduced 2-Chern-Simons theory together with its equations of motion and its behaviour under gauge transformations, we are in conditions to define its holomorphic variant. Following \cite{Costello:2019tri}, we take our $4$-manifold $X=\mathbb{R} \times Y$ and we complexify and compactify the copy of $\R$ to $\mathbb{C}P^1$. Taking coordinates $z,\bar z$ on $\mathbb{C}P^1$ and $x^i$ with $i=1,2,3$ on $Y$, we define \textbf{holomorphic 2-Chern-Simons theory} (h2CS) by the action functional 
\begin{equation}
\label{ec:h2cs}
    S_\text{h2CS}[A,B] = \frac{1}{2 \pi i} \int_X \omega \wedge \langle F(A) - \tfrac{1}{2}\mu_1(B), B\rangle
\end{equation}
where $\omega \in \Omega^{(1,0)}(X)$ is a meromorphic $(1,0)$-form on $\mathbb{C}P^1$. In this article we will make a particular choice of $\omega$ given by 
\begin{equation}
    \omega = \frac{\mathrm{d}z}{z} \,.
\end{equation}
Note that $\omega$ is nowhere vanishing and has simple poles at $z=0$ and $z=\infty$. For future use, we write down the gauge fields in components as 
\begin{align}
    & A= A_{\bar z} \mathrm{d}\bar z + A_i \mathrm{d}x^i \\
    & B= B_{i\bar z}\,\mathrm{d}x^i \wedge \mathrm{d}\bar z + B_{ij}\,\mathrm{d}x^i \wedge \mathrm{d}x^j \,,
\end{align}
where we have ignored the $z$-components of the gauge fields provided $\omega \in \Omega^{(1,0)}(X)$.  

Without any additional constraints, the action \eqref{ec:h2cs} is not well defined: although the singularities in $\omega$ do not break the local integrability of the action (see \textbf{Lemma 2.1} in \cite{Benini:2020skc}), boundary terms arise under an arbitrary field variation. Indeed, we have
\begin{equation}
\label{ec:fullvar}
    \delta S_{\mathrm{h2CS}} = \frac{1}{2\pi i}\int_X  \omega\wedge \langle \delta \mathcal{F},B \rangle - \omega\wedge \langle \delta A, \mathcal{G} \rangle + \mathrm{d}\omega \wedge \langle \delta A,B \rangle.
\end{equation}
The first two terms are bulk terms which give rise to the equations of motion. The third term, however, is a distribution localized at the poles of $\omega$. Indeed, we can integrate along $\mathbb{C}P^1$, to find (see \textbf{Lemma 2.2} in \cite{Delduc:2019whp})
\begin{equation}
\label{ec:bdaryvar}
    \frac{1}{2\pi i}\int_{\mathbb{C}P^1\times Y} \mathrm{d}\omega \wedge \langle \delta A,B\rangle = \int_Y \iota_Y^*\langle \delta A,B\rangle|_{z=0}-\iota_Y^*\langle \delta A,B\rangle_{z=\infty}\,,
\end{equation}
where $\iota_Y: Y\hookrightarrow X$ is an embedding of the 3-manifold $Y$ into $X=\mathbb{C}P^1 \times Y$.
Thus, to have a well defined action principle, namely $\delta S_{\text{h2CS}}=0$, we must impose boundary conditions on the gauge fields $A$ and $B$ at $z=0,\infty$. Throughout the bulk of the article we will make different choices of boundary conditions depending on the additional structure we introduce on $Y$. However, the main example to which we will consistently come back to simplify discussions is given by the choice
\begin{equation}
\begin{aligned}
    \label{ec:bconditions1}
    & A_1|_{z=0}=A_2|_{z=0}=B_{12}|_{z=0}=0 \,,\\
    & A_3|_{z=\infty}=B_{13}|_{z=\infty}=B_{23}|_{z=\infty}=0 \,,
\end{aligned} 
\end{equation}
and we restrict to variations satisfying the same boundary conditions. It is immediate to verify that this choice of boundary conditions makes \eqref{ec:bdaryvar} vanish.

\begin{remark}
    The quantity \eqref{ec:bdaryvar} defines the {\it symplectic form} $\varpi_\text{blk}$ on the on-shell fields (i.e. Maurer-Cartan elements, c.f. \S \ref{hMC}) $\mathsf{MC}_\text{blk}\subset \iota_Y^* \Omega(X)\otimes \G$ of the bulk $5$d theory. In particular, a choice of boundary conditions such that $\delta S_{\mathrm{h2CS}}=0$ is then equivalent to a choice of a Lagrangian subspace $\mathcal{L}\subset \mathsf{MC}_\text{blk}$ such that
\begin{equation*}
    \varpi_\text{blk}\vert_\mathcal{L} =0.
\end{equation*}
Indeed, fields $(\iota_Y^*A,\iota_Y^*B)\in\mathcal{L}$ by definition will satisfy $\delta S_\text{2CS}[A,B] = 0$. The variational quantities such as $\delta A,\delta B$, are naturally interpreted as covectors on {\it field space}. This is made precise in the covariant phase space formalism \cite{Julia:2002df,Geiller:2022fyg}.
\end{remark}


In the above language, the boundary conditions \eqref{ec:bconditions1} correspond to the choice of Lagrangian subspace
\begin{equation*}
    \mathcal{L} = \operatorname{Span}\{A_3,B_{23},B_{13}\}\vert_{z=0} \cup \operatorname{Span}\{A_1,A_2,B_{12}\}\vert_{z=\infty}\,,
\end{equation*}
which is simply given by the components of the fields that survive at the respective poles $z=0,z=\infty$. We shall see in \S \ref{3dift} that this example is in fact part of a covariant family of such "chiral" Lagrangian subspaces $\mathcal{L}_\ell\subset\mathsf{MC}_\text{blk}$, which are parameterized by a unit vector $\ell\in S^2$ called the \textbf{chirality}. Moreover, in \S \ref{covbcs} we shall see that, if additional structure such as a \emph{transverse holormophic foliation} is imposed on $Y$, the different choices of the chirality vector $\ell$ can lead to very different three dimensional actions. 

\medskip

Going back to our h2CS action, the bulk equations of motion are readily obtained from \eqref{ec:fullvar}. Indeed, they imply the Maurer-Cartan condition --- namely the fake- and 2-flatness, see \S \ref{hMC} --- of the 2-connection $\mathcal{A}=(A,B)$. In components, these are 
\begin{align}
    \label{ec:eom1}
   & \partial_{\Bar{z}}A_i - \partial_{i}A_{\bar z}+[A_{\bar z},A_i]-\mu_1(B_{\Bar{z}i})=0 \\
   \label{ec:eom2}
    & \partial_j A_i - \partial_i A_j +[A_j,A_i]-\mu_1(B_{ij})=0  \\
    \label{ec:eom3}
    & \partial_{i}B_{j\Bar{z}}+ \partial_{\bar z}B_{ij}+\mu_2(A_i,B_{j \Bar{z}})+\mu_2(A_{\bar z},B_{ij}) =0 \\
    \label{ec:eom4}
    & \partial_{i}B_{jk}+\mu_2(A_i,B_{jk}) =0
\end{align}
where $i,j,k\in \{1,2,3\}$ and where we have explicitly separated the $\bar z$ component equations for later convenience. Preempting what will happen after localisation to a three-dimensional boundary theory, equations \eqref{ec:eom1} and \eqref{ec:eom3} will imply that $A$ and $B$ are holomorphic on $\mathbb{C}P^1$ and therefore constant along $\mathbb{C}P^1$. On the other hand, equations \eqref{ec:eom2} and \eqref{ec:eom4} will result on fake flatness and 2-flatness of the three-dimensional localised fields.

\subsection{Gauge Symmetries of \texorpdfstring{$\mathrm{h2CS}$}{V}}
\label{sec:gaugesym5d}

As discussed in \S \ref{sec:gsof2CS}, 2-Chern-Simons theory is invariant (modulo boundary terms) under an arbitrary 2-gauge transformation. In the holomorphic case, the introduction of the meromorphic $1$-form $\omega$ breaks this gauge symmetry, in the sense that only a subset of gauge transformations leave \eqref{ec:h2cs} invariant. In particular, it is precisely the breaking of this gauge symmetry which gives rise to the edge-modes on the boundary, which will be the fields of our three-dimensional theory. Moreover, the transformations which remain symmetries of the 5d action will descend to symmetries of the 3d boundary theory.

This is entirely analogous to what happens in the case of 4d Chern-Simons theory, as well as in holomorphic Chern-Simons theory on Twistor space. Notably, in these cases, the gauge transformations which remain symmetries of the action after the introduction of $\omega$ are those which preserve the boundary conditions imposed on the gauge fields. We will show that this is not true any more in h2CS theory. 

Indeed, in proposition \ref{prop:gvar} we have shown that under the 2-gauge transformation \eqref{ec:2gaugetr1} the variation is a total boundary term 
\begin{equation}
\label{ec:omega}
    \Omega= \langle \mathrm{Ad}_h^{-1}F(A) ,\Gamma\rangle + \tfrac{1}{2} \langle\mathrm{Ad}_h^{-1}A+h^{-1}\mathrm{d}h,[\Gamma,\Gamma]\rangle + \tfrac{1}{2}L_{\mathrm{CS}}(\Gamma)\,.
\end{equation}
With the introduction of the meromorphic $1$-form we have that under a finite gauge transformation
\begin{equation}
\begin{split}
        S_{\text{h2CS}}\big[A^{(h,\Gamma)},B^{(h,\Gamma)}\big]
        &=S_{\text{h2CS}}[A,B]+\frac{1}{2\pi i}\int_X \omega \wedge \mathrm{d}\Omega[A,B,h,\Gamma] \\
        &= S_{\text{h2CS}}[A,B] + \int_{Y}\iota_Y^*\Omega[A,B,h,\Gamma]|_{z=0}-\iota_Y^*\Omega[A,B,h,\Gamma]|_{z=\infty}\,,
\end{split}
\end{equation}
where in the second line we have integrated by parts to get the distribution $\mathrm{d}\omega$ and then we have integrated along $\mathbb{C}P^1$ to localize at the poles. Let us consider the boundary conditions \eqref{ec:bconditions1} and restrict to gauge transformations which preserve these boundary conditions. In other words, we constraint our gauge parameters $(h,\Gamma)$ by requiring that $A^{(h,\Gamma)}$ and $B^{(h,\Gamma)}$ satisfy the same boundary conditions than $A$ and $B$. From \eqref{ec:2gaugetr1} we find the set of constraints
\begin{align}
   z=0: \quad & h^{-1}\partial_1h + \mu_1(\Gamma_1) = h^{-1}\partial_2 h + \mu_1(\Gamma_2) = \partial_1 \Gamma_2 - \partial_2 \Gamma_1 -[\Gamma_1,\Gamma_2] =0  \label{ec:gtbczero} \\    z=\infty:\quad & h^{-1}\partial_3h  + \mu_1(\Gamma_3)=\partial_i \Gamma_3 - \partial_3 \Gamma_i + \mu_2(A^{(h,\Gamma)}_i,\Gamma_3) - [\Gamma_i,\Gamma_3] =0 \,,\quad i=1,2\,. \label{ec:gtbcinfty}
\end{align}
Writing down $\iota_Y^*\Omega$ explicitly, using the boundary conditions \eqref{ec:bconditions1} and the constraints \eqref{ec:gtbczero} and \eqref{ec:gtbcinfty} we find
\begin{align}
     \int_{Y}\iota_Y^*\Omega[A,B,h,\Gamma]|_{z=0}&-\iota_Y^*\Omega[A,B,h,\Gamma]|_{z=\infty} =\int_Y \mathrm{vol}_3 \big[ \langle \mu_1 (\partial_3 \Gamma_1),\Gamma_2 \rangle |_{z=0} \nonumber\\
     &\qquad -\big(\langle A_2,h\rhd[\Gamma_1,\Gamma_3]\rangle-\langle [A_1,A_2],h\rhd \Gamma_3\rangle -\langle \mu(\Gamma_1),\partial_3 \Gamma_2\rangle\big)|_{z=\infty}\big]\,.\label{ec:deltaSgamma1}
\end{align}
In other words, we see that preserving the boundary conditions is not enough to attain the symmetry. This is in fact a surprising result, because of equation \eqref{ec:bdaryvar}. Indeed there we have imposed the same boundary conditions on the gauge fields and its variations and this led to the vanishing of $\delta S$ on-shell. In particular, if we take the arbitrary variations $\delta A,\delta B$ to be infinitesimal gauge transformations then this would imply that \eqref{ec:deltaSgamma1} should vanish on-shell, which it does if we neglect terms quadratic in $\Gamma$.\footnote{The attentive reader might worry about the term $\langle [A_1,A_2],h\rhd \Gamma_3\rangle$ not being quadratic in $\Gamma$. However, the claim is that $\delta S =0$ \emph{on-shell}. In particular, using the equations of motion and the boundary conditions, one can show that this term indeed vanishes for infinitesimal gauge transformations.} The failure of the vanishing of the boundary variation in going from infinitesimal to finite gauge transformations can be attributed to the presence of the 3d Chern-Simons $L_{\mathrm{CS}}(\Gamma)$ term appearing in \eqref{ec:omega}, which we know is sensitive to the global structure of the $3$-manifold $Y$.

This implies that if we want our transformations to be symmetries of h2CS theory, we must impose further constraints on the gauge parameters $(h,\Gamma)$. In order to do this in a consistent manner, we introduce some structure. Recall that in \S \ref{sec:h2csth} we considered Lagrangian subspaces to define boundary conditions. In particular, the condition that the gauge transformations should preserve the boundary conditions can be expressed as follows. Given a choice of Lagrangian subspace $\mathcal{L} \subset\mathsf{MC}_\text{blk}$ of the on-shell Maurer-Cartan fields $(A,B)$, we require 
\begin{equation}\label{def2alg}
    \big(\iota_Y^*A^{(h,\Gamma)},\iota_Y^*B^{(h,\Gamma)}\big)\in \mathcal{L}\,,\qquad \forall~ (\iota_Y^*A,\iota_Y^*B)\in\mathcal{L}\,.
\end{equation}
This imposes differential and algebraic constraints on the derived 2-gauge parameters $(h,\Gamma)\in \mathsf{D}\L_0 =(\Omega^\bullet(X)\otimes \mathsf{D}\mathbb{G})_0$, and thus it defines a {\it defect subcomplex} $\mathsf{D}\L_\text{def}\subset \mathsf{D}\L_0$. This is in fact a derived Lie 2-subgroup of $\mathsf{D}\L_0$, since the composition law in the derived 2-group $\mathsf{D}\L_0$ satisfies\footnote{This follows from the closure of the 2-gauge algebra \cite{Martins:2010ry,Mikovic:2016xmo,Baez:2004in,chen:2022}. For strict Lie 2-algebras $\G=\operatorname{Lie}\G$, this is true in general. However, if $\G$ were weak, then it is known that the 2-gauge algebra closes only on-shell of the fake-flatness condition \cite{Kim:2019owc}.}
\begin{equation*}
    \big(\iota_Y^*A^{(h,\Gamma)\cdot (h',\Gamma')},\iota_Y^*B^{(h,\Gamma)\cdot (h',\Gamma')}\big) = \bigg(\iota_Y^*\big(A^{(h,\Gamma)}\big)^{ (h',\Gamma')},\iota_Y^*\big(B^{(h,\Gamma)}\big)^{(h',\Gamma')}\bigg).
\end{equation*}
Therefore, if $(h,\Gamma),(h',\Gamma')\in\mathsf{D}\L_\text{def}$ lie in the defect subcomplex, then \eqref{def2alg} state that so does $(h,\Gamma)\cdot (h',\Gamma')$. In the example considered above, we have that $\mathsf{D}\L_{\mathrm{def}}$ is given by the set of 2-gauge transformations which satisfy the boundary constraints \eqref{ec:gtbczero} and \eqref{ec:gtbcinfty}. 

As these constraints are localized at the poles, are independent, and only involve $Y$-dependencies, they are constraints on 2-gauge transformations of the \textbf{archipelago type} \cite{Delduc:2019whp}. Formally, they are finite 2-gauge transformations that are non-trivial only at a neighborhood around each puncture $z=0,z=\infty$. The precise definition is the following.
\begin{definition}
    We say a 2-gauge parameter $(h,\Gamma)\in\mathsf{D}\L_0$ is of \textbf{archipelago type} if and only if there exists open discs $D_z\subset\mathbb{C}P^1$ of radius $R_z$ containing each puncture $z=0,\infty$ such that:
    \begin{enumerate}
        \item $(h,\Gamma)=(1,0)$ is trivial away from $D_0\coprod D_\infty$,
        \item the restriction $(h,\Gamma)\vert_{Y\times D_z}$ depends on $Y$ and the radial coordinate $0<r<R_z$ in $D_z$, and
        \item there exists an open disc $V_z\subset D_z$ such that $(h,\Gamma)\vert_{Y\times V_z}$ depends only on $Y$.
    \end{enumerate}
\end{definition}
It suffices to consider 2-gauge transformations of archipelago type in $\mathsf{D}\L_\text{def}$; the restriction map $-\vert_{Y\times V_0} \times -\vert_{Y\times V_\infty}$ to the two open discs gives a surjection 
\begin{equation*}
    \mathsf{D}\L_\text{def} \rightarrow \mathsf{D}\L_\text{def}^0\times \mathsf{D}\L_\text{def}^\infty \subset (\Omega^\bullet(Y)\otimes\mathsf{D}\mathbb{G})_0^{2\times}
\end{equation*}
onto those $(h,\Gamma)\in\mathsf{D}\L_\text{def}^0$, depending only on $Y$, satisfying \eqref{ec:gtbczero} and those $(h',\Gamma')\in\mathsf{D}\L_\text{def}^\infty$ satisfying \eqref{ec:gtbcinfty}. Conversely, the elements in $\mathsf{D}\L_\text{def}^0\times \mathsf{D}\L_\text{def}^\infty$ can always be patched along $\mathbb{C}P^1$ into a 2-gauge transformation of archipelago type, by using bump functions along the radial coordinate on $D_z\supset V_z$.

Let us use this structure to impose the additional constraints on the gauge parameters so they become symmetries of h2CS theory. We therefore introduce a subcomplex $\mathsf{D}\L_\text{sym}\subset\mathsf{D}\L_\text{def}$ consisting of those derived gauge parameters which are symmetries of the action. In other words, we define $\mathsf{D}\L_\text{sym}$ such that
\begin{equation}
    (h,\Gamma)\in\mathsf{D}\L_\text{sym} \quad \text{if and only if}\quad  \Omega[A,B,h,\Gamma]\vert_{z=0} - \Omega[A,B,h,\Gamma]\vert_{z=\infty}=0 \,.
\end{equation}
The additional constraints we impose on $(h,\Gamma)$ will be dictated by the requirement that $\mathsf{D}\L_\text{sym}\subset \mathsf{D}\L_\text{def}$ forms a derived Lie 2-subgroup. More precisely, we have
\begin{proposition}
The derived Lie 2-group of defect symmetries $\mathsf{D}\L_\text{sym}$ associated to h2CS and the Lagrangian subspace defined by \eqref{ec:bconditions1} is  
\begin{equation}
    \mathsf{D}\L_\text{sym} \cong \mathsf{D}\L_\text{sym}^0\times \mathsf{D}\L_\text{sym}^\infty\,,\qquad \begin{cases} \mathsf{D}\L_\text{sym}^0 = \{(h,\Gamma)\in\mathsf{D}\L_\text{def}^0\mid \Gamma_1,\Gamma_2 \in \Omega^1(Y)\otimes \t\} \\ \mathsf{D}\L_\text{sym}^\infty = \{(h',\Gamma')\in \mathsf{D}\L_\text{def}^\infty\mid  \Gamma_3'=0\}\end{cases}\label{defectgaugeconstraints}
\end{equation}
where $\t\subset\h$ is the maximal Abelian subalgebra of $\h$.  
\end{proposition}
\begin{proof}
     Since these constraints concern only the 1-form gauge parameters $\Gamma$, it is easy to see that $\mathsf{D}\L_\text{sym}\subset\mathsf{D}\L_\text{def}$ is a derived 2-subgroup and the constraints are closed under 2-group composition law. To see that \eqref{defectgaugeconstraints} indeed eliminates \eqref{ec:deltaSgamma1}, we note that $\Gamma'_3=0$ eliminates the first two terms evaluated at $z=\infty$. Moreover, for the last term at $z=\infty$, we first notice that if $\Gamma_3'=0$ then \eqref{ec:gtbcinfty} implies
\begin{equation*}
    \partial_{i} \Gamma'_3 - \partial_3 \Gamma'_{i}=- \partial_3 \Gamma'_{i}=0\,,\quad i=1,2\,,
\end{equation*}
hence the last term automatically drops. Finally, for the term evaluated at $z=0$ we have
\begin{equation}
   2\langle \mu_1(\partial_3\Gamma_1),\Gamma_2\rangle=\langle \mu_1(\partial_3\Gamma_1),\Gamma_2\rangle-\langle \mu_1(\partial_3\Gamma_2),\Gamma_1\rangle= \langle h^{-1}\partial_3h, [\Gamma_1,\Gamma_2]\rangle \label{defectcomputation}
\end{equation}
where we have repeatedly integrated by parts and used the boundary constraints \eqref{ec:gtbczero}. In particular, this term drops iff $\Gamma_1,\Gamma_2$ are Abelian.
\end{proof}


We will prove in \S \ref{residuals} that $\mathsf{D}\L_\text{sym}$ indeed descends to global symmetries of the 3d field theory upon localization. In fact, we will also demonstrate in \S \ref{dgaffinecurrents} that these global symmetries exhibit conserved Noether charges that satisfy a differential graded analogue of the affine Virasoro Lie algebra that lives in the 2d Wess-Zumino-Witten model \cite{KNIZHNIK198483}.

\subsection{Field Reparametrization}
\label{sec:fieldreparam}

We now begin the localization procedure for 5d h2CS theory \eqref{ec:h2cs}, to obtain a 3-dimensional boundary theory. This will proceed, in spirit, entirely analogous to that in the context of CS$_4$ theory \cite{Costello:2019tri} and hCS$_6$ theory on twistor space \cite{Bittleston:2020hfv}. However, this has not been done explicitly before in the context of 2-Chern-Simons theory, and we will therefore describe it in detail.

Towards this, we first introduce new field variables $A' \in \Omega^1(X)\otimes \mathfrak{g}$, $\hat g \in C^{\infty}(X,G)$, $B'\in \Omega^{2}(X)\otimes \mathfrak{h}$ and $\hat \Theta \in \Omega^1(X)\otimes \mathfrak{h}$ that reparameterize our original fields $(A,B)$ in \eqref{ec:h2cs}, with
\begin{align}
\label{ec:reparA}
    &A= \hat g^{-1}A'\hat g  +\hat g^{-1}\mathrm{d} \hat g + \mu_1(\hat \Theta)\,,\\
\label{ec:reparB}
    &B = \hat g^{-1}\rhd B' + \mathrm{d}\hat \Theta + \mu_2(A,\hat \Theta) -\tfrac{1}{2}[\hat \Theta, \hat \Theta] \,.
\end{align}
This expression implies that, if $(\hat g,\hat \Theta) \in \mathsf{D}\L_{\mathrm{sym}}$ then $S_{\text{2CS}}[A,B]=S_{\text{2CS}}[A',B']$, in which case such reparametrizations are just a gauge redundancy. Crucially, however, we will intentionally take $(\hat g,\hat \Theta) \notin \mathsf{D}\L_{\mathrm{sym}}$. 

On the other hand, the reparametrisations \eqref{ec:reparA} and \eqref{ec:reparB} are in fact, redundant. More precisely, we can perform transformations on $A',B',\hat g$ and $\hat \Theta$ which leave $A$ and $B$ invariant. These are characterized by the following proposition
\begin{proposition}
\label{prop:internalsym}
    Given $u \in C^{\infty}(X,G)$ and $\Lambda \in \Omega^1(X) \otimes \mathfrak{h}$, the transformations
    \begin{align}
    \label{ec:actiononu}
    & \hat g\mapsto u^{-1}\hat g,\hspace{3.05cm} A'\mapsto u^{-1}A'u + u^{-1}\mathrm{d}u+ \mu_1(\Lambda)\\
    \label{ec:actionontheta}
    &\hat \Theta\mapsto \hat \Theta -(\hat g^{-1}u)\rhd \Lambda \qquad \qquad B'\mapsto u^{-1}\rhd B' +d\Lambda + \mu_2(A'^{(u,\Lambda)},\Lambda)-\tfrac{1}{2}[\Lambda,\Lambda].
\end{align}
leave the gauge fields $A$ and $B$ invariant. 
\end{proposition}

\begin{proof}
    The proof follows by explicit computation. Indeed, from equation \eqref{ec:reparA} we find
    \begin{align*}
    \hat g^{-1}A' \hat g &\mapsto  \hat g^{-1}A'\hat g + \mathrm{Ad}_{\hat g}^{-1} \mathrm{d}uu^{-1} + \mathrm{Ad}_{\hat{g}^{-1}u} \mu_1(\Lambda)\\
     \hat g^{-1}\mathrm{d}\hat g &\mapsto  -\mathrm{Ad}_{\hat g}^{-1} \mathrm{d}uu^{-1} + \hat g^{-1}\mathrm{d}\hat g \\
     \mu_1(\Theta) &\mapsto \mu_1(\Theta) - \mathrm{Ad}_{\hat{g}^{-1}u} \mu_1(\Lambda)\,,
\end{align*}
from where invariance of $A$ follows. Similarly, from equation \eqref{ec:reparB} we find
\begin{align}
\label{ec:Binv1}
    \hat g^{-1}\rhd B' &\mapsto \hat{g}^{-1}\rhd \big[ B' + u\rhd (\mathrm{d}\Lambda + \tfrac{1}{2}[\Lambda,\Lambda]) +\rhd\mu_2(A',u\rhd \Lambda)\big] + \mu_2(\mathrm{d}uu^{-1},u\rhd \Lambda) \\
    \label{ec:Binv2}
     \mathrm{d}\hat \Theta &\mapsto  \mathrm{d}\hat \Theta + \mu_2(\hat g^{-1}\mathrm{d}\hat g,\hat g u^{-1}\rhd \Lambda)- \mu_2(\mathrm{d}uu^{-1},u\rhd \Lambda)-\hat g^{-1}u\rhd \mathrm{d}\Lambda\\
     \label{ec:Binv3}
    \mu_2(A,\hat \Theta) &\mapsto \mu_2(A,\hat \Theta)-\mu_2(A,\hat g^{-1}u\rhd \Lambda)\\
    \label{ec:Binv4}
     -\tfrac{1}{2}[\hat\Theta,\hat\Theta] &\mapsto -\tfrac{1}{2}[\hat\Theta,\hat\Theta] + \mu_2(\mu_1(\hat \Theta),\hat g^{-1}u\rhd \Lambda)-\hat g^{-1}u \rhd \tfrac{1}{2}[\Lambda,\Lambda]\,.
\end{align}
Adding all of these together we are left after some cancellations with 
\begin{equation}
    B \mapsto B - \mu_2(A,\hat g^{-1}u\rhd \Lambda) + \mu_2(\hat g^{-1}A'\hat g  +\hat g^{-1}\mathrm{d} \hat g + \mu_1(\hat \Theta),\hat g^{-1}u\rhd \Lambda) =B \,,
\end{equation}
where the last equality follows from the definition of $A$ in \eqref{ec:reparA}. 
\end{proof}

\begin{definition}
    We call the transformations $(u,\Lambda)$ of proposition \ref{prop:internalsym} {\bf internal}  2-gauge transformations, and the invariance of $A$ and $B$ under these transformations the {\bf reparametrization symmetry}. 
\end{definition}

The idea is to use this reparametrization symmetry to fix the $\bar z$ components of $A'$ and $B'$ to zero. This will allow us to interpret the (gauge fixed) gauge fields as a flat 2-connection on $Y$ after localisation. Thus, we take an internal 2-gauge transformation $(u_1,\Lambda_1)$ such that 
\begin{equation}
    L' = A'^{(u_1,\Lambda_1)}\,,\quad H'=B'^{(u_1,\Lambda_1)}
\end{equation}
satisfy $L'_{\bar z}=0$ and $H'_{i\bar z}=0$. This amounts to a set of constraints on $u_1$ and $\Lambda_1$ given by 
\begin{align}
    & L'_{\bar z}= u_1^{-1}A'_{\bar z}u_1 + u_1^{-1}\partial_{\bar z}u_1 + \mu_1(\Lambda_{1,\bar z})=0 \label{1gauconst1} \\
    &H'_{i\bar z}= u_1^{-1}\rhd B'_{i\bar z} + \partial_{i}\Lambda_{1,\bar z}-\partial_{\bar z}\Lambda_{1,i} + \mu_2(L'_i,\Lambda_{1,\bar z})-\mu_2(L'_{\bar z},\Lambda_{1,i})- [\Lambda_{1,i},\Lambda_{1,\bar z}]=0.\label{2gauconst1}
\end{align}

\begin{lemma}
There exist an internal 2-gauge transformation satisfying \eqref{1gauconst1} and \eqref{2gauconst1}.  
\end{lemma}

\begin{proof}
    Consider first the constraint \eqref{1gauconst1}, which can be written as
\begin{equation}
    \partial_{\bar z}u_1 u_1^{-1} = A'_{\bar z} + \mu_1(u_1\rhd \Lambda_{1,\bar z})\,.
\end{equation}
Due to the graded nature of the problem, the most natural constraint is the following:
\begin{equation}
    u_1\rhd \Lambda_{1,\bar z} \in \operatorname{ker}\mu_1\,,\qquad \partial_{\bar z}u_1 u_1^{-1} = A'_{\bar z}\,.
\end{equation}
Since $\operatorname{ker}\mu_1$ is an ideal, this implies $\Lambda_{1,\bar z}\in\operatorname{ker}\mu_1$. The quantity $u_1$ is then only subject to the differential equation $\partial_{\bar z}u_1u_1^{-1}=A'_{\bar z}$, which fixes the $\bar z$-dependence of $u_1$. The above solution also implies
\begin{equation*}
    [\Lambda_{1,i},\Lambda_{1,\bar z}] = -\tfrac{1}{2}\mu_2(\mu_1(\Lambda_{1,\bar z}),\Lambda_{1,i})=0 \,,
\end{equation*}
by the Peiffer identity, hence the quadratic term in \eqref{2gauconst1} drops. Assuming further that each component of $\Lambda_1$ are independent, this secondary differential constraint can be written as a decoupled set of PDEs 
\begin{align}
    &\partial_i \Lambda_{1,\bar z} + \mu_2(L_i',\Lambda_{1,\bar z}) = -u_1^{-1}\rhd B'_{i\bar z}\,,\\
    &\partial_{\bar z} \Lambda_{1,i} + \mu_2(L_{\bar z}',\Lambda_{1,i}) = -u_1^{-1}\rhd B'_{i \bar z}\,.
\end{align}
If we introduce the covariant derivatives on $X=\mathbb{C}P^1\times Y$
\begin{equation}
    Y: \nabla^{L'}_i = \partial_i +\mu_2(L'_i,-) \,,\qquad \mathbb{C}P^1: \bar\partial^{L'} = \partial_{\bar z} + \mu_2(L'_{\bar z},-)\,,
\end{equation}
then these equations take the form of inhomogeneous covariant curvature equations
\begin{equation}
    \nabla^{L'}_i \Lambda_{1,\bar z} = -u^{-1}_1\rhd B'_{i\bar z}\,,\qquad \bar\partial^{L'}\Lambda_{1,i} = -u_1^{-1}\rhd B'_{i \bar z}\,.
\end{equation}
By Fredholm alternative, \textbf{provided the differential operators $\nabla^{L'},\bar\partial^{L'}$ have trivial kernel}, then we can find {\it unique} solutions $(\Lambda_{1,\bar z},\Lambda_{1,i})$ to these inhomogeneous PDEs for each value of the forcing $(u_1\rhd B'_{i\bar z},u_1^{-1}\rhd B'_{i \bar z})$.
\end{proof}

Having used part of the reparametrization symmetry to fix the $\bar z$ components of the fields to zero, we look for any remaining internal gauge transformation which preserves this condition. In other words, we look for $(u_2,\Lambda_2)$ such that 
\begin{equation}
    L \coloneqq L'^{(u_2,\Lambda_2)}\,,\quad H \coloneqq H'^{(u_2,\Lambda_2)} \,
\end{equation}
satisfy $L_{\bar z}=0$ and $H_{i\bar z}=0$. This amounts to
\begin{align}
    & L_{\bar z}=  u_2^{-1}\partial_{\bar z}u_2 + \mu_1(\Lambda_{2,\bar z}) =0 \\
    &H_{i\bar z} = \partial_{i}\Lambda_{2,\bar z}-\partial_{\bar z}\Lambda_{2,i} + \mu_2(L_i,\Lambda_{2,\bar z})-\mu_2(L_{\bar z},\Lambda_{2,i})- [\Lambda_{2,i},\Lambda_{2,\bar z}]=0\,,
\end{align}
which can be achieved by taking $(u_2, \Lambda_2)$ with 
\begin{equation}
    \partial_{\bar z}u_2=0\,,\quad \Lambda_{2,\bar z}=0 \,,\quad \partial_{\bar z}\Lambda_{2,i} =0\,.
\end{equation}
Note that these conditions on $(u_2,\Lambda_2)$ give us complete freedom on $u_2|_Y$ and $\Lambda_{2,i}|_Y$, since they only constrain the $\mathbb{C}P^1$-dependence of the latter. 

Hence, we can use this internal gauge transformation to fix the values of $\hat g$ and $\hat \Theta$ at one of the two punctures, either $z=0$ or $z=\infty$. Taking $(u_2,\Lambda_2)$ such that $u_2|_Y=\hat g|_{z=\infty}$ and $\Lambda_{2,i}|_{Y}=\hat{\Theta}_i|_{z=\infty}$, our fields $\hat g$ and $\hat \Theta$ at the poles become
\begin{equation}
\label{ec:gaugefixing1}
    \hat g|_{z=0}\coloneqq g\,,\quad \hat g|_{z=\infty}=1\,, \quad \hat\Theta|_{z=0}\coloneqq \Theta\,,\quad \hat \Theta_i|_{z=\infty}=0\,.
\end{equation}
To summarize, after fixing the reparametrisation symmetry we have arrived to the expressions
\begin{align}
\label{ec:AintermsofL}
    & A= \hat g^{-1} L \hat g + \hat g^{-1}\rd \hat g + \mu_1(\hat \Theta) \\
    & B = \hat g^{-1}\rhd H + \mathrm{d}\hat \Theta +\mu_2(A,\hat \Theta)-\tfrac{1}{2}[\hat \Theta,\hat\Theta]
\end{align}
where the $\bar z$ components of $L$ and $H$ vanish, and the values of $\hat g$ and $\hat \Theta$ at the poles of $\omega$ are given by \eqref{ec:gaugefixing1}.

\section{Three-Dimensional Field Theories}\label{3dift}

We are ready to construct the main player of this paper: a family of three-dimensional field theories associated to different choices of boundary conditions. To begin with, we write our 5d action \eqref{ec:h2cs} in terms of the reparametrisation fields $L,H,\hat g$ and $\hat \Theta$. Given that the expression of $A$ and $B$ in terms of the latter is formally the same than that of a gauge transformation, we can use \eqref{ec:actionundergt} to write
\begin{equation}
\label{ec:reparaction}
    S_{\mathrm{h2CS}}[A,B] = \frac{1}{2\pi i}\int_X \langle F(L)-\tfrac{1}{2}\mu_1(H),H\rangle + \frac{1}{2\pi i}\int_X\omega \wedge  \mathrm{d}\Omega\big(L,H,\hat g,\hat \Theta\big)
\end{equation}
where we recall the expression for $\Omega$ once again
\begin{equation}
\label{ec:omega12345}
    \Omega\big(L,H,\hat g,\hat \Theta\big) = \langle \hat g^{-1}F(L)\hat g,\hat \Theta\rangle + \langle \hat g^{-1}L \hat g+\tfrac{1}{2}\hat g^{-1}\mathrm{d}\hat g,[\hat\Theta, \hat\Theta]\rangle + \tfrac{1}{2} L_{\mathrm{CS}}\big(\hat\Theta\big)\,.
\end{equation}
Let us note that since $H_{i\bar z}=0$ then $\langle \mu_1(H),H\rangle=0$. On the other hand, the second term can be integrated by parts to find
\begin{equation}
\label{ec:reparaction2}
    S_{\mathrm{2CS}}\big[L,H,\hat g,\hat \Theta\big] = \frac{1}{2\pi i}\int_X \omega \wedge \langle F(L),H\rangle + \frac{1}{2\pi i}\int_X \mathrm{d}\omega \wedge \Omega\big(L,H,\hat g,\hat \Theta\big)\,.
\end{equation}
While the second term in the above expression is effectively three-dimensional due to the presence of $\mathrm{d}\omega$, the first term is a genuine five-dimensional bulk term which we must eliminate to obtain a three-dimensional theory. We do this by going partially on-shell. Indeed, under an arbitrary variation $\delta H$ we find the bulk equation of motion $\partial_{\bar z}L=0$. The remaining term, can be integrated along $\mathbb{C}P^1$ to find (see \textbf{Lemma 2.2} in \cite{Delduc:2019whp})
\begin{equation}
\label{ec:3dift1}
\begin{aligned}
    S_{3d}[L,H,\hat g,\hat \Theta] &= 
    \int_{Y}\iota_Y^*\Omega\big(L,H,\hat g,\hat\Theta\big)|_{z=0}-\iota_Y^*\Omega\big(L, H,\hat g,\hat \Theta\big)|_{z=\infty}\\
    & =\int_Y  \big[\langle g^{-1}F(L) g,\Theta\rangle + \tfrac{1}{2} \langle g^{-1}L g+ g^{-1}\mathrm{d}g,[\Theta,\Theta]\rangle + \tfrac{1}{2} L_{\mathrm{CS}}(\Theta)\rangle\big]_{z=0}
\end{aligned}
\end{equation}
where we have used the expressions \eqref{ec:gaugefixing1} for the evaluations of $\hat g$ and $\hat \Theta$ at $z=0$. Note that term at infinity is identically vanishing due to the gauge fixing conditions \eqref{ec:gaugefixing1}. The action \eqref{ec:3dift1} is three-dimensional but it is still written in terms of $L$, whereas we want our 3d action to be written in terms of $g$ and $\Theta$ only. To do so, we will consider a family of boundary conditions which includes \eqref{ec:bconditions1}, in order to express $L$ in terms of $g$ and $\Theta$ and obtain explicit forms for three-dimensional actions.


\subsection{Covariant family of boundary conditions}\label{covbcs}
We begin with a particularly convenient family of boundary conditions/Lagrangian subspaces $\mathcal{L}\subset\mathsf{MC}_\text{blk}$, which can be understood as a generalization of the example \eqref{ec:bconditions1}. First, recall that the 2-Chern-Simons action is symmetric under global Poincar{\'e} transformations. For a product manifold $X = Y\times \mathbb{C}P^1$, the Poincar{\'e} group admits as subgroups the isometry group $\operatorname{Iso}(Y)$ of $Y$. It can be seen that this isometry group $\operatorname{Iso}(Y)$ is inherited as a global symmetry of the symplectic form $\varpi_\text{blk}$ \eqref{ec:bdaryvar}.

Here, we are going to introduce a family of boundary conditions that breaks only the linear part of the isometries $\operatorname{Iso}(Y)$ (namely, we neglect the translations). For concreteness and simplicity, let us for now think of $Y$ as a real Riemannian 3-manifold, then the linear isometries are given by $O(3)\subset \operatorname{Iso}(Y)$. The symmetry breaking patterns are described by maximal subgroups of $O(3)$. It is known that there is at most one unique maximal Lie (ie. infinite) subgroup $O(2)$ of $O(3)$, from which we obtain the following {\it symmetry breaking pattern}
\begin{equation}
    S^2\hookrightarrow O(3)\rightarrow O(2).
\end{equation}
The fibre $S^2\cong O(3)/O(2)$ is a 2-sphere, which paramterizes a collection $\{\mathcal{L}_\ell\mid \ell\in S^2\}$ of Lagrangian subspaces of $\mathsf{MC}_\text{blk}$. We call this data $\ell\in S^2$ the \textbf{chirality} vector.

\medskip

We now work to explicitly write down the boundary conditions associated to $\mathcal{L}_\ell$. To do so, we first define the 3-dimensional vectors $\vec{A} = (A_1,A_2,A_3)$ and $\vec{B} = (B_{23},B_{13},B_{12})$, which are built locally out of the components of our fields $(A,B)$ along $Y$. For a given triple $\vec{\ell},\vec{n},\vec{m}$ of unit vectors on $Y$, we consider the family of boundary conditions
\begin{align}
\label{ec:covbc1}
    & \vec n\cdot \vec{A}|_{z=0}= \vec m\cdot \vec{A}|_{z=0}= \vec \ell\cdot \vec{B}|_{z=0}=0 \\
    \label{ec:covbc2}
    & \vec \ell\cdot \vec{A}|_{z=\infty}= \vec m\cdot \vec{B}|_{z=\infty}=\vec n\cdot \vec{B}|_{z=\infty}=0 \,,
\end{align}
which is given by the following Lagrangian subspace
\begin{equation}
\label{ec:lagrangiansubspaces}
    \mathcal{L}_\ell = \operatorname{Span}\{\vec \ell\cdot \vec{A},\vec{B}_\perp\}\vert_{z=0} \cup\operatorname{Span}\{\vec{A}_\perp,\vec \ell\cdot \vec{B}\}\vert_{z=\infty},
\end{equation}
where we have used a shorthand to denote $\vec{A}_\perp = (\vec n\cdot \vec{A},\vec m\cdot \vec{A})$.

\begin{remark}
It can be seen by direct computation, that this choice of boundary conditions implies $\delta S_{\text{2CS}}=0$ if and only if the triple $(\vec\ell,\vec n,\vec m)$ defines a global orthonormal frame on $Y$ specified up to orientation by, say, $\vec\ell$. The simple example considered in \eqref{ec:bconditions1} can be recovered with $\vec\ell=(0,0,1)$, and all other Lagrangian subspaces $\mathcal{L}_\ell$ work in the same way. The characterization of the defect symmetry 2-group $\mathsf{D}\L_\text{def}$ \eqref{defectgaugeconstraints}, thus holds for any member $\mathcal{L}_\ell$ of this $S^2$-family of Lagrangian subspaces up to a global rotation of the framing on $Y$.
\end{remark}

\subsection{The 3d Actions}\label{3dlocalization}
Having chosen our boundary conditions defined through the Lagrangian subspace \eqref{ec:lagrangiansubspaces} we proceed with the evaluation of the three dimensional action \eqref{ec:3dift1}. To do so, we must use our boundary conditions \eqref{ec:covbc1} and \eqref{ec:covbc2} to solve for $L$ in terms of $g$ and $\Theta$. Recall that $A$ and $L$ are related by equation \eqref{ec:AintermsofL}, namely
\begin{equation}
    A = \hat g^{-1}L\hat g + \hat g^{-1}\mathrm{d}\hat g + \mu_1(\hat \Theta)\,.
\end{equation}
Hence, the boundary conditions imply
\begin{align}
    & 0=\vec n \cdot \vec A|_{z=0} = g^{-1} \big(\vec n \cdot \vec L|_{z=0}\big)g+ g^{-1}\big(\vec n \cdot \nabla\big) g + \mu_1\big(\vec n \cdot \vec \Theta\big) \label{chiralbc1}\\
    &0=\vec m \cdot \vec A|_{z=0} = g^{-1} \big(\vec m \cdot \vec L|_{z=0}\big)g+ g^{-1}\big(\vec m \cdot \nabla\big) g + \mu_1\big(\vec m \cdot \vec \Theta\big) \label{chiralbc2}\\
    &0 = \vec \ell \cdot \vec A|_{z=\infty} =\vec \ell \cdot \vec L|_{z=\infty}\label{chiralbc3}
\end{align}
where we have used the fact that $\hat g|_{z=\infty}=1$ and $\hat \Theta|_{z=\infty}=0$. The last constraint \eqref{chiralbc3} is of particular importance, since it will in fact eliminate the component $\operatorname{proj}_\ell L = (\vec \ell \cdot \vec L)\vec\ell=0$ along $\ell$ entirely. This is because $L$ is constrained to be holomorphic on $\bbC P^1$ according to the bulk equation of motion $\partial_{\bar z}L_i =0$ and hence constant. The boundary condition $\vec \ell \cdot \vec L|_{z=\infty}=0$ then fixes this constant to be zero.


Let us compute each of the terms in $\Omega$ separately, making use of the boundary conditions \eqref{chiralbc1}, \eqref{chiralbc2}, \eqref{chiralbc3}. It will be convenient to introduce the notation 
\begin{equation}
    \vec{a}_\ell = (\vec\ell\cdot \vec a)\vec\ell = \operatorname{proj}_\ell(a)\nonumber
\end{equation}
for the projection of vectors $\vec{a}$ (or vector-valued operator such as $\nabla$) along $\ell$, and for any Lie algebra-valued vectors $\vec{\alpha},\vec{\beta}$ we define
\begin{align*}
    \langle \vec\alpha\cdot\vec\beta\rangle &= \langle \alpha_1,\beta_1\rangle + \langle \alpha_2,\beta_2\rangle+\langle \alpha_3,\beta_3\rangle\\
    [\vec{\alpha}\times \vec{\beta}] &= ([\alpha_2,\beta_3] - [\alpha_2,\beta_3],[\alpha_1,\beta_3] - [\beta_1,\alpha_3],[\alpha_1,\beta_2] - [\beta_1,\alpha_2]).\footnotemark
\end{align*}
\footnotetext{Since we know $(\vec n,\vec m,\vec \ell)$ forms an orthonormal frame on $Y$, the triple product $(-\times -)\cdot -$ satisfies
\begin{equation*}
    (\vec a\times \vec b)\cdot \vec c_\ell = (\vec a_{n}\times \vec b_{m} - \vec a_{m}\times \vec b_{n})\cdot \vec c.
\end{equation*}} Let $\vec k=\nabla gg^{-1}$ denote the {\it right} Maurer-Cartan form. We can then compute
\begin{equation}
\begin{split}
    \iota_Y^* \langle g^{-1}F(L)g,\Theta\rangle 
    &=-\langle \mu_1(\nabla\times\vec\Theta)\cdot\vec\Theta_\ell\rangle  +\langle  \mu_1([\vec\Theta\times\vec\Theta])\cdot \Theta\rangle\\
    &\qquad -\langle \mathrm{Ad}_g^{-1}\nabla_\ell \times \vec k\cdot \vec\Theta\rangle-2\langle \operatorname{Ad}_g^{-1}\vec k_\ell\cdot [\vec\Theta\times \vec\Theta]\rangle \,.
\end{split}
\end{equation}
Second, we have
\begin{equation}
    \iota_Y^*\langle g^{-1}Lg+g^{-1}\mathrm{d}g,\Theta \wedge \Theta\rangle=-2\langle \mu_1(\vec\Theta)\cdot [\vec\Theta \times\vec\Theta]\rangle+\langle \operatorname{Ad}_g^{-1}\vec k_\ell\cdot [\vec\Theta\times\vec\Theta]\rangle\,,
\end{equation}
while the Chern-Simons term $L_{\mathrm{CS}}(\Theta)$ is simply
\begin{align}
    L_{\mathrm{CS}}(\Theta) &= 2\langle \mu_1(\vec\Theta)\cdot\nabla\times\vec\Theta\rangle+2\langle\mu_1(\vec\Theta),[\vec\Theta\times \vec\Theta]\rangle\,.
\end{align}
In combination, we find\footnote{Notice the term $-\langle \mu_1(\nabla\times\vec\Theta)\cdot\Theta_\ell\rangle$ in $\iota_Y^*\langle g^{-1}F(L)g,\Theta\rangle $ combines to change the sign of an identical term in $L_{\mathrm{CS}}(\Theta)$. This allows us to perform an integration by parts to eliminate the appearance of $\Theta_\ell$ altogether in the resulting Lagrangian \eqref{3dintegrand}.} from \eqref{ec:omega12345}
\begin{align}
\iota_Y^*\Omega &= -\langle \mathrm{Ad}_g^{-1}\nabla_\ell\times\vec k\cdot \vec\Theta\rangle+ \langle\mu_1(\vec\Theta)\cdot \nabla_\ell \times \vec\Theta\rangle-\langle\operatorname{Ad}_g^{-1}\vec k_\ell \cdot [\vec\Theta\times\vec\Theta] \rangle\,.\label{3dintegrand}
\end{align}
Finally, we note that we can perform one more simplification by using the invariance of the pairing,
\begin{equation}
    \langle \mu_1(g\rhd\vec\Theta)\cdot \nabla_\ell \times (g\rhd \vec\Theta)\rangle= \langle\mu_1(\vec\Theta)\cdot\nabla_\ell \times\vec\Theta\rangle - \langle \vec k_\ell\cdot g\rhd [\vec\Theta\times\vec\Theta]\rangle,
\end{equation}
from which we construct a family of 3d theories given by
\begin{equation}\label{covariant3daction}
    S_{3d}[g,\Theta] = -\int_Y \langle \nabla_{\ell}\times \vec{k}\cdot \vec{\Theta}^g\rangle - \langle 
    \mu_1(\vec\Theta^g)\cdot \nabla_\ell\times \vec\Theta^g\rangle \,,
\end{equation}
parameterized by ${\ell}\in S^2$, where $\vec\Theta^g = g\rhd\vec\Theta$. 


\medskip

We shall show in the following that this theory is {\it completely topological}, in the sense that its equations of motion describe configurations of 2-connections $(L,H)$ that are flat in each real (topological) direction.


\subsubsection{Equations of Motion}\label{3deoms}
The equations of motion of the theory are obtained by varying the fields $g,\Theta$. Indeed, under a variation $\delta\Theta$, we find
\begin{equation}\label{1flat?}
   \delta \Theta: \quad \nabla_\ell \times \vec{k} + \mu_1(\nabla_\ell\times (g\rhd\vec\Theta))=0\,.
\end{equation}
To vary $g$ we define $\delta g$ as the infinitesimal translation of $g$ under right-multiplication $g+\delta_g g = g\cdot(1+\epsilon)$, where $\epsilon\in\mathfrak{g}$ is an infinitesimal Lie algebra element. Using the following formulas
\begin{gather}
    \delta_g (\nabla gg^{-1}) = \operatorname{Ad}_g \nabla \epsilon\,,\qquad \delta_g (g\rhd \vec{\Theta}) = g\rhd \mu_2(\epsilon,\vec{\Theta})\,, \nonumber\\
\delta_g (g^{-1}\nabla g) =  \mathrm{d}\epsilon + [g^{-1}\nabla g,\epsilon]\,, \nonumber
\end{gather}
we compute the variation of the first term neglecting exact boundary terms
\begin{equation}
    \delta_g\langle \nabla_\ell\times \vec{k}\cdot g\rhd\vec{\Theta}\rangle =\langle \epsilon,g^{-1}\rhd \nabla_\ell\cdot (g\rhd (\nabla\times \vec{\Theta}))\rangle \label{variation1}
\end{equation}
where we have used the vector identities
\begin{gather*}
  \nabla\cdot\nabla\times\vec A=0\,,\quad \nabla\cdot(\vec A\times\vec B) = \vec{A}\cdot(\nabla\times\vec{B}) - \vec{B}\cdot(\nabla\times\vec A)\,, \quad
  \nabla\times(f\vec{A})= (\nabla f)\times\vec{A} + f\nabla\times\vec{A}\,.
\end{gather*}
The variation of the second term is
\begin{equation}
    \delta_g \langle \vec{k}_\ell\cdot g\rhd [\vec{\Theta}\times \vec{\Theta}]\rangle = -\langle \epsilon, g^{-1}\rhd (\nabla_\ell \cdot g\rhd [\vec{\Theta}\times\vec{\Theta}])\rangle\,,
\end{equation}
and hence combining with \eqref{variation1} we find the second equation of motion
\begin{equation}\label{2flat}
    \delta g: \quad \nabla_\ell\cdot g\rhd (\nabla\times \vec{\Theta} - [\Theta\times\Theta]) =0.
\end{equation}

\paragraph{Fake- and 2-flatness in 3-dimensions.}
We now work to show that the equations of motion \eqref{1flat?} and \eqref{2flat} of \eqref{covariant3daction} implies the flatness of the 2-connection. Without loss of generality (WLOG), it suffices to prove this statement for one choice of $\ell$, such as $\vec{\ell}=(0,0,1)$ as in \eqref{ec:bconditions1}. This is because $S^2$ is a homogeneous space under $O(3)$ and hence any two choices of the Lagrangian subspace $\mathcal{L}_\ell$ are related by an $O(3)$-action on the 3d theory \eqref{covariant3daction}. For this choice of $\ell$, the equations of motion are given by
\begin{align}
\label{ec:123eomth}
    &\delta \Theta: \quad \partial_3k_i  = -\mu_1(\partial_3(g\rhd\Theta_i))\,,\quad i=1,2\\
    &\delta g: \quad \partial_3\left[g\rhd(\partial_1\Theta_2-\partial_2\Theta_1 - [\Theta_1,\Theta_2])\right]=0\,.\label{2flatx3}
\end{align}

\begin{proposition}\label{extended2flat}
    The equations of motion for \eqref{covariant3daction} are equivalent to fake- and 2-flatness for $(L,H)$ on $Y$.
\end{proposition}
\begin{proof}
    The fake-flatness and 2-flatness conditions for $(\tilde L,\tilde H) \coloneq (\iota_Y^* L,\iota_Y^* H)\vert_{z=0}$ are given by
    \begin{equation}
    \label{ec:formflatness}
        F(\tilde L)-\mu_1(\tilde H)=0\,,\quad \rd \tilde H + \mu_2(\tilde L,\tilde H)=0\,.
    \end{equation}
    The bulk equations of motion \eqref{ec:eom1}, \eqref{ec:eom3} and the boundary conditions \eqref{ec:bconditions1} at $z=\infty$ make $L_3=0, H_{13}=0, H_{23}=0$ on $X=\mathbb{C}P^1\times Y$, so that fake flatness becomes in components
        \begin{equation}
        \partial_1\tilde L_2 - \partial_2\tilde L_1 + [\tilde L_1,\tilde L_2] = \mu_1(\tilde H_{12})\,,\quad \partial_3 \tilde L_1 =0 \,,\quad \partial_3 \tilde L_2 = 0\,, 
    \end{equation}
    and 2-flatness is simply
    \begin{equation}
        \partial_3 \tilde H_{12}=0\,.
    \end{equation}
    On the other hand, the boundary conditions at $z=0$ allow us to compute
    \begin{align}
        &\tilde L_i = -k_i - \mu_1(g\rhd \Theta_i)\,,\quad i=1,2\\
        &\tilde H_{12} = -g\rhd(\partial_1\Theta_2 - \partial_2\Theta_1 - [\Theta_1,\Theta_2])\,,\label{currents}
    \end{align}
    In particular, we find that
    \begin{equation*}
        \partial_3 \tilde H_{12} = -\partial_3 g\rhd(\partial_1\Theta_2-\partial_2\Theta_1 - [\Theta_1,\Theta_2]) 
    \end{equation*}
    vanishes due to \eqref{2flatx3}, and this is precisely the 2-flatness condition. On the other hand, we begin by evaluating 
    \begin{align*}
        \partial_1 \tilde L_2 - \partial_2\tilde L_1 + [\tilde L_1,\tilde L_2] &= -\partial_1k_2 + \partial_2k_1 - [k_1,k_2] + [\mu_1(g\rhd\Theta_1),\mu_1(g\rhd\Theta_2)]\\
        &\qquad -~ \partial_1\mu_1(g\rhd\Theta_2) + \partial_2\mu_1(g\rhd\Theta_1) + [k_1,\mu_1(g\rhd\Theta_2)] - [\mu_1(g\rhd\Theta_1),k_2]\\
        &= -\mu_1\mu_2(k_1,g\rhd\Theta_2) +[k_1,\mu_1(g\rhd\Theta_2)] + \mu_1\mu_2(k_2,g\rhd\Theta_1) - [k_2,\mu_1(g\rhd\Theta_1)] \\
        &\qquad - \mu_1(g\rhd (\partial_1\Theta_2 - \partial_2\Theta_1 - [\Theta_1,\Theta_2])) \\ 
        &= \mu_1(\tilde H_{3}),
    \end{align*}
    where we have used the Maurer-Cartan equations for $k_1,k_2$ as well as the equivariance and Peiffer identities. On the other hand, we have
    \begin{equation*}
        -\partial_3\tilde L_1 = \partial_3(k_1 + \mu_1(g\rhd \Theta_1)),\qquad -\partial_3\tilde L_2 = \partial_3(k_2 + \mu_1(g\rhd \Theta_2)),
    \end{equation*}
    which both vanish due to the equations of motion \eqref{ec:123eomth}. These are precisely the fake-flatness conditions.
\end{proof}
\noindent Note the proof involves local computations, and hence the above statement still holds regardless if $\ell$ aligns with the foliation, given we swap the indices of the (co)vectors involved consistently.

\subsection{Transverse Holomorphic Foliation}\label{thf}
In this section, we are going to relax the assumption that $Y$ is a real Riemannian 3-manifold. We do this by equipping $Y$ with a {\it transverse holomorhic foliation} (THF) $\Phi$, which we take to be along the $x_3$-direction. This means that we have a decomposition of the tangent bundle $TY$ such that each transverse leaf $\operatorname{ker}\mathrm{d}x_3$ has equipped an almost complex structure \cite{Scrdua2017OnTH}. Each local patch $U\subset Y$ can then be given coordinates $(x_3,w,\bar w)$, where $x_3 \in \R$ and $(w,\bar w)\in \mathbb{C}$, such that transitions $(x_3,w,\bar w) \mapsto (x_3',w',\bar w')$ look like \cite{Aganagic:2017tvx}
\begin{equation*}
    x_3'(x_3,w,\bar w),\qquad w'(w),\qquad \bar w'(\bar w)\,.
\end{equation*}
In the following, we will assume the leaves of $\Phi$ never intersect, and that they define an integrable distribution such that $\Phi$ foliates out a complex 2-submanifold $M\subset Y$. 

There are then two directionalities present in our theory: the chirality vector $\ell$ and the direction of the foliation $\Phi$. These directions can either be aligned or misaligned, which lead to distinct theories. We shall study these two cases in more detail in the following.

\subsubsection{Foliation aligned with chirality} 
First, let us suppose the chirality $\ell = \mathrm{d}x_3$ is aligned. The complex structure on the transverse leaves of $\Phi$ allows us to impose boundary conditions which is a mix of topological and holomorphic directions, such as
    \begin{align*}
    & A_w|_{z=0}= A_{\bar w}|_{z=0}= B_{w\bar w}|_{z=0}=0 \\
    & A_3|_{z=\infty}=  B_{3w}|_{z=\infty}= B_{3\bar w}|_{z=\infty}=0\,.
\end{align*}
These boundary conditions can be understood as a partially holomorphic version of \eqref{ec:bconditions1} and indeed lead to the following 3d \textbf{topological-holomorphic} field theory,
\begin{align}\label{holo3daction}
    S_{\mathrm{h}3d} &= \int_Y \operatorname{vol}_Y\left[-\langle \partial_3 k_{\bar w},\Theta^g_{w}\rangle +\langle \partial_3 k_w,\Theta^g_{\bar w}\rangle - \langle \mu_1(\Theta_w^g),\partial_3\Theta_{\bar w}^g\rangle\right].
\end{align}
Writing $w = x_1+ix_2\in M$ in terms of the real coordinates, the above action can be expressed in terms of the lightcone coordinates $x_\pm = x_1 \pm x_2$ on $M$,
\begin{align*}
    S_{\mathrm{h}3d} &= 2\int_Y \operatorname{vol}_Y\left[-\langle \partial_3 k_+,\Theta^g_-\rangle +\langle \partial_3 k_-,\Theta^g_+\rangle - \langle \mu_1(\Theta_+^g),\partial_3\Theta^g_-\rangle\right].
\end{align*}
Notice this action can be reproduced from the action \eqref{covariant3daction} with $\ell=\mathrm{d}x_3$ from a $SO(2)$-rotation
\begin{equation*}
    \hat x_1 \mapsto \frac{1}{\sqrt{2}}(\hat x_1-\hat x_2),\qquad \hat x_2\mapsto \frac{1}{\sqrt{2}}(\hat x_1 + \hat x_2),
\end{equation*}
which can be interpreted as a boost into the lightcone coordinates on $M$. Therefore, the action \eqref{holo3daction} can also be intuitively obtained by "endowing" \eqref{covariant3daction} with a complex structure. This is not a purely cosmetic choice, however; the existence of a complex structure on $M$ plays a significant role in determining the boundary conditions satisfied by the symmetries of our theory (see \S \ref{residuals}).

\subsubsection{Foliation misaligned with chirality}\label{misaligned}
Let us now consider the case where the direction $\mathrm{d}x_3$ of the transverse holomorphic foliation $\Phi$ does {\it not} align with the chirality $d\vec\ell$. We may pick, for instance, the chirality $\vec\ell$ along the holomorphic direction $w$, which leads to the boundary conditions 
\begin{align*}
    & A_3|_{z=0}= A_{\bar w}|_{z=0}= B_{3\bar w}|_{z=0}=0 \\
    & A_{w}|_{z=\infty}=  B_{w\bar w}|_{z=\infty}= B_{3 w}|_{z=\infty}=0 \,.
\end{align*}
The three-dimensional action is then 
\begin{align*}
    S_{\mathrm{h}3d} &= \int_Y \operatorname{vol}_Y\left[-\langle \partial_w k_3,\Theta_{\bar w}^g\rangle +\langle \partial_w k_{\bar w},\Theta_3^g\rangle - \langle\mu_1(\Theta_{\bar w}^g),\partial_w\Theta^g_3\rangle \right].
\end{align*}
Consider the case $\mu_1=0$, whence the final term drops. Using the identities 
\begin{equation}\label{MCidentity}
    \partial_i(\partial_j gg^{-1}) = \operatorname{Ad}_g(\partial_j(g^{-1}\partial_ig)),\qquad \partial_i(g^{-1}\partial_j g) = \operatorname{Ad}_g^{-1}(\partial_j(\partial_i gg^{-1})), 
\end{equation}
we can bring the above action to the form
\begin{equation}\label{CSmatter}
    S_{\mathrm{h}3d} =-\int_Y dw\wedge d\bar w \wedge \mathrm{d}x_3\left[\langle\Theta_{\bar w}, \partial_3m_w\rangle - \langle \Theta_3,\partial_{\bar w}m_w\rangle\right],
\end{equation}
where $m = g^{-1}\mathrm{d}g$ is the {\it left} Maurer-Cartan form. This is in fact nothing but the Chern-Simons/matter theory studied in \cite{Aganagic:2017tvx}, where we take $R= \g = \operatorname{Lie}G$ as the adjoint $G$-representation and $R^c=\h$ its conjugate.\footnote{With respect to the degree-1 paring $\langle-,-\rangle$ on the balanced Lie 2-algebra $\G$.} In other words, we are able to recover the interaction between Chern-Simons theory and matter from our localization procedure on 5d 2-Chern-Simons theory.

It was explained in \cite{Aganagic:2017tvx} that, in the context of the topological A-model, when the flux on a 5-brane vanishes on an embedded 3-brane, then the Chern-Simons theory living on this 3-brane acquires this matter coupling term due to those strings that extend in the ambient 5-brane directions. The fact that we can recover such matter coupling terms from our prescription suggests that the holomorphic 5d 2-Chern-Simons theory may play a part in how the ambient strings couple to the 3-brane in the topological A-model.

\begin{remark}
    We emphasize here that the distinction of aligned vs. misaligned matters only when $Y$ is equipped with a THF $\Phi$. Indeed, if no such data is imposed, then the chirality can always be made to align with $\Phi$ with an $SO(3)$-action on the family $\{\mathcal{L}_\ell\}_{\ell\in S^2}$ of Lagrangians. However, this cannot be done with a THF: indeed, if $\ell$ is tangent to $M$, then one cannot rotate it out of the leaves without destroying the complex structure on $M$. We will see in \S \ref{integrability} that this makes the two types (aligned/misaligned) of theories dynamically distinct. 
\end{remark}

\section{Symmetries of the 3d Theory}\label{residuals}
Let us now study the symmetries of our three dimensional theories. As discussed in \S \ref{sec:gaugesym5d}, the 2-gauge symmetries of h2CS together with the Lagrangian subspace $\mathcal{L}_\ell$ (i.e. the defect) form the derived 2-group of defect symmetries $\mathsf{D}\L_\text{sym}(\ell)$ characterized in \eqref{defectgaugeconstraints}. The main goal of this section is to examine the {\it global} symmetries $\mathsf{D}\L_{3d}$ of the localized 3d theory \eqref{covariant3daction}, and we shall prove that it in fact sits inside of $\mathsf{D}\L_\text{sym}$.

For simplicity, we shall focus on the analysis of the symmetries of the 3d theory in the case where the chirality $\ell$ aligns with the $\mathrm{d}x_3$-direction. If $Y$ has equipped a THF, then the misaligned case can be analogously dealt with, but the resulting symmetries will have different properties.

We begin by first writing down the 3d action \eqref{covariant3daction} with the chirality $\mathrm{d}x_3$,
\begin{equation}
    S_{3d}[g,\Theta] =  \int_Y \langle\partial_3 k_2,\Theta_1^g\rangle -\langle \partial_3 k_1,\Theta^g_2\rangle +\langle \mu_1(\Theta^g_1),\partial_3(\Theta^g_2)\rangle \,,\label{3daction}
\end{equation}
where $\Theta^g = g\rhd \Theta$. As we have already noted in \S \ref{3dlocalization}, the field $\operatorname{proj}_\ell\Theta = \Theta_3$ is completely absent from the theory, so that we will restrict to gauge transformation parameters $\Gamma$ with no $\rd \ell$-leg. We denote the space of 1-forms with no $\mathrm{d}\ell$-leg by $\Omega_\perp^1(Y)\subset\Omega^1(Y)$.

Let us first begin by characterizing the symmetries of the 3d theory \eqref{3daction}. We define the following derived 2-group algebra $(\mathsf{D}\L_\perp)_\bullet = C^\infty(Y)\otimes G \times \Omega^\bullet_\perp(Y)\otimes \h$, and consider the subspace $\mathsf{D}\L_{3d}^R \subset (\mathsf{D}\L_\perp)_0$ consisting of elements $(\tilde h,\tilde \Gamma)$ satisfying
   \begin{equation}
    \tilde h^{-1}\partial_1 \tilde h + \mu_1(\tilde \Gamma_1) =\tilde h^{-1}\partial_2 \tilde h + \mu_1(\tilde \Gamma_2) = 0 \quad \text{and} \quad 
    \partial_1 \tilde \Gamma_2 - \partial_2\tilde \Gamma_1 - [\tilde \Gamma_1,\tilde \Gamma_2] =0\,.\label{ec:exgtatzero}
\end{equation}
On the other hand, we let $\mathsf{D}\L_{3d}^L\subset(\mathsf{D}\L_\perp)_0$ denote the subspace consisting of elements $(h,\Gamma)$ satisfying
\begin{equation}\label{ec:extgtatinfty}
    h^{-1}\partial_3 h = 0 \quad \text{and}\quad \partial_3\Gamma_1=\partial_3\Gamma_2=0\,.
\end{equation}

\begin{lemma}
Let $\mathsf{D}\mathring{\L}_{3d}^R\subset \mathsf{D}\L_{3d}^R$ be the subspace for which the components $\tilde \Gamma\in\Omega^1_\perp(Y)\otimes\t$ are valued in the maximal Abelian subalgebra $\t\subset\h$. The 3d theory \eqref{3daction} is invariant under the following transformations:
    \begin{align}
\label{ec:rightaction}
     g \mapsto  g \tilde h, &\qquad \Theta \mapsto \tilde h^{-1}\rhd \Theta+\tilde \Gamma,\qquad (\tilde h,\tilde \Gamma) \in \mathsf{D}\mathring{\L}_{3d}^R \,\\
     \label{ec:leftaction}
    g\mapsto h^{-1}g, &\qquad \Theta \mapsto \Theta-(h^{-1}g)^{-1}\rhd \Gamma,\qquad (h,\Gamma)\in\mathsf{D}\L_{3d}^L\,.
\end{align}
Moreover, these transformations commute.
\end{lemma}
\noindent In other words, \eqref{3daction} is invariant under the action of the derived 2-group $\mathsf{D}\L_{3d} = \mathsf{D}\mathring{\L}_{3d}^R\times \mathsf{D}\L_{3d}^L$.
\begin{proof}
    A simple computation using \eqref{ec:extgtatinfty} shows that the action \eqref{3daction} is invariant under a left-action $(h,\Gamma)\in\mathsf{D}\L_{3d}^L$. For the right-action $(\tilde h,\tilde\Gamma)\in\mathsf{D}\mathring{\L}^R_{3d}$  \eqref{ec:rightaction}, one can show using \eqref{ec:exgtatzero} that the action transforms as
    \begin{equation*}
        S_{3d}[g,\Theta]\rightarrow S_{3d}[g,\Theta] + \int_Y \langle \mu_1 \partial_3 \tilde\Gamma_1,\tilde\Gamma_2 \rangle.
    \end{equation*}
    The same computation as \eqref{defectcomputation} then removes the remaining term. 

    Now consider a combination of right-acting $(\tilde h,\tilde \Gamma)$ and left-acting $(h,\Gamma)$ symmetry transformations. Of course, their actions clearly commute on $g$, thus we focus on $\Theta$. We have 
\begin{align*}
    \Theta&\mapsto \tilde h^{-1}\rhd \Theta + \tilde \Gamma \mapsto  \tilde h^{-1}\rhd (\Theta - (h^{-1}g)^{-1}\rhd \Gamma)) + \tilde \Gamma \\
    &=  \tilde h^{-1}\rhd \Theta - (h^{-1}g \tilde h)^{-1}\rhd \Gamma + \tilde \Gamma. 
\end{align*}
On the other hand, we have
\begin{equation*}
    \Theta\mapsto \Theta - (h^{-1}g)^{-1}\rhd \Gamma \mapsto  \tilde h^{-1}\rhd \Theta + \tilde \Gamma - (h^{-1}g\tilde h)^{-1}\rhd \Gamma,
\end{equation*}
where we noted that $g\mapsto g\tilde h$ under a right-action. These quantities indeed coincide.
\end{proof}

In section \S \ref{dgaffinecurrents} we will study the conserved Noether currents corresponding to these symmetries, and analyze their homotopy Lie algebra structure in detail. In particular, we will prove that these currents in fact form a centrally extended affine Lie 2-algebra, in complete analogy with the chiral currents in the 2d WZW model \cite{KNIZHNIK198483}.


\subsection{Origin of the Symmetries}

Having characterized the symmetries of the 3d action, we are now in the position to examine how these are actually inherited from the symmetries of the five-dimensional theory discussed in \ref{sec:gaugesym5d}. In other words, we will show that the defect symmetries $\mathsf{D}\L_\text{sym}$ defined in \eqref{defectgaugeconstraints} descend to the 3d symmetries $\mathsf{D}\L_{3d}$. 

Recall that in \S \ref{sec:gaugesym5d} we studied the transformation properties of the 5d gauge fields $(A,B)$, whereas in \S \ref{sec:fieldreparam} we studied those of the reparametrisation fields $(A',B';\hat g, \hat \Theta)$. We will call the former {\bf external} 2-gauge transformations to distinguish them from the latter, which we have dubbed {\bf internal}. The action of the external 2-gauge transformations on the reparametrization fields is characterized by the following
\begin{proposition}
\label{prop:externalsym}
    The action of an external 2-gauge transformation $(\hat h,\hat \Gamma)$\footnote{We have slightly modified the notation from $(h,\Gamma)$ to $(\hat h,\hat \Gamma)$ for convenience.} given by \eqref{ec:2gaugetr1} induces an action on the reparametrisation fields $\hat g$ and $\hat \Theta$ given by 
    \begin{equation}
     \label{ec:extgt123}
     \hat g \mapsto \hat g \hat h \,,\quad \hat \Theta \mapsto h^{-1}\rhd \hat \Theta +\hat \Gamma \,,
\end{equation}
    and leaves the fields $A'$ and $B'$ invariant. 
\end{proposition}

\begin{proof}
    The proof follows by simply replacing \eqref{ec:extgt123} in \eqref{ec:reparA} and \eqref{ec:reparB} while keeping $A'$ and $B'$ unchanged. 
\end{proof}

We summarize the action of both external and internal gauge transformations on the fields in table \ref{symmtab}. The "mixed" column is there to simply emphasize that, as the fields $(\hat g,\hat \Theta)$ transforms under both internal and external transformations (propositions \ref{prop:internalsym} and \ref{prop:externalsym}), one must in general perform both {\it simultaneously} in order to keep them fixed.

\begin{table}[h]
    \centering
    \begin{tabular}{c|c|c|c}
         & External & Internal & Mixed \\
         \hline
      $(A,B)$   & $\checkmark$ & $\times$  & $\checkmark$ \\ 
        \hline
        $(A',B')$ & $\times$ & $\checkmark$ & $\checkmark$ \\ 
        \hline
        $(\hat g,\hat \Theta)$ & $\checkmark$ & $\checkmark$ & $\times$
    \end{tabular}
    \caption{A table summarizing how different types of symmetries act on the reparameterized fields $(A,B) = (A',B';\hat g,\hat\Theta)$. A check mark $\checkmark$ indicates that the fields under consideration transforms non-trivially, while a cross $\times$ indicates that the fields do not transform. } 
    \label{symmtab}
\end{table}

In \S \ref{sec:fieldreparam} we have shown that for each field configuration $(A,B)$ in h2CS, one can use the internal symmetries to find a \textbf{Lax presentation} $(A,B) = (L,H;\hat g,\hat \Theta)$ such that $(L,H)$ has no $\bar z$-components, and that $(\hat g,\hat \Theta)\vert_{z=\infty}=(1,0)$. We then used this Lax presentation to construct the $S_{3d}$ in \S \ref{3dlocalization}. The goal now is to show that this Lax presentation is stable under an external 2-gauge transformation \eqref{ec:2gaugetr1}. For simplicity, we shall work with the Lagrangian subspace $\mathcal{L}_{\mathrm{d}x_3}$ \eqref{ec:bconditions1}, but our computations can be generalized.

Recall that in \S \ref{3dlocalization}, using the equations of motion and the boundary conditions, we have argued that $L_3,H_{31},H_{32}=0$ are trivial globally. Hence, we look for an external gauge transformation $(h,\Gamma)$ which preserves these conditions, together with $(\hat g,\hat \Theta)\vert_{z=\infty}=(1,0)$. Following table \ref{symmtab} there are two different cases: 

\begin{enumerate}
    \item {\bf Case 1:} As the external symmetries will transform $(\hat g,\hat\Theta)$, in order to preserve the constraint at $z=\infty$ we may simply ask  $(\hat h,\hat \Gamma)\vert_{z=\infty}=(1,0)$ to be trivial there. This is equivalent to simply doing a 2-gauge transformation in a neighborhood of $z=0$ which will not affect the fields at $z=\infty$.
    
    \item {\bf Case 2:} Now suppose we take $(\hat h,\hat \Gamma)$ to be generic. In order to preserve $(\hat g,\hat \Theta)\vert_{z=\infty}=(1,0)$, we must then simultaneously perform an internal symmetry $(u,\Lambda)$ to compensate. However, since $(u,\Lambda)$ transforms the reparameterization fields, they must be holomorphic in order to preserve the conditions $L_{\bar z}=0$ and $H_{i\bar z}=0$. As such, they, too, are constants on $\mathbb{C}P^1$ with values equal to
    \begin{equation*}
        u=\hat h\vert_{z=\infty},\qquad \Lambda = \hat\Gamma\vert_{z=\infty}\,.
    \end{equation*}
    Moreover, if we want these transformations to preserve the global condition $L_3,H_{31},H_{32}=0$ of the Lax presentation $(L,H)$ we need
    \begin{align*}
        0&=u^{-1}\partial_3 u + \mu_1\Lambda_3 \\
        0&=\partial_{i}\Lambda_3 - \partial_3\Lambda_{i} + \mu_2(A_{i},\Lambda_3) - \mu_1(A_3,\Lambda_{i}) - [\Lambda_{i},\Lambda_3] \,, \quad i=1,2\,.
    \end{align*}
    
    Consider the second equation. The only object that depends on $\bbC P^1$ is the field $A$. As such, due to the boundary conditions $A_3\vert_{z=\infty}=0$ and $A_{1,2}\vert_{z=0}=0$, we must in fact have
    \begin{equation*}
        \partial_{i}\Lambda_3 - \partial_3\Lambda_{i} - [\Lambda_{i},\Lambda_3] =- F_{i3}(-\Lambda) = 0 \,,\quad i=1,2
    \end{equation*}    
    everywhere, where $F(\Lambda) = \rd\Lambda + \tfrac{1}{2}[\Lambda
    ,\Lambda]$. On the other hand, the boundary conditions satisfied by $(\hat h,\hat \Gamma)\vert_{z=\infty}$ read
    \begin{equation}
    0=\partial_{i}\hat\Gamma_3 - \partial_3\hat\Gamma_{i} + \mu_2(A_{i},\hat\Gamma_3) - [\hat\Gamma_{i},\hat\Gamma_3]\vert_{z=\infty} = -F_{i3}(-\hat\Gamma) + \mu_2(A_{i},\hat\Gamma_3)\vert_{z=\infty} \,,\quad i=1,2.
    \end{equation}
    But if $(\hat h,\hat\Gamma)\vert_{z=\infty} = (u,\Lambda)$ in order to preserve the condition $(\hat g,\hat \Theta)\vert_{z=\infty} = (1,0)$, then the flatness condition $F(-\Lambda)_{i3}=0$ satisfied by $(u,\Lambda)$ implies
\begin{equation}\label{crowningachievement}
        \mu_2(A_{i},\hat\Gamma_3)\vert_{z=\infty} =0 \,,\quad i=1,2.
    \end{equation}
\end{enumerate}
Clearly, this condition \eqref{crowningachievement} is supplanted by the constraint $\hat\Gamma_3\vert_{z=\infty}=0$ in the defect symmetries $\mathsf{D}\L_\text{sym}$ \eqref{defectgaugeconstraints}. It is therefore a {\it 2-subgroup} of the external 2-gauge transformations that preserve the Lax presentation!

\subsection{3d Symmetries from 5d Symmetries}

Having analyzed the action of both the external and internal 2-gauge transformations on the Lax presentation, we are ready to describe how these symmetries combine to give the symmetries of the three-dimensional action.

Given that we have obtained $S_{3d}$ by localizing h2CS at $z=0$, one may expect that only the component $\mathsf{D}\L_\text{def}^0$ of the derived defect 2-group $\mathsf{D}\L_\text{def} = \mathsf{D}\L_\text{def}^0 \times \mathsf{D}\L_\text{def}^\infty$ contributes to the symmetries of $S_{3d}$. Indeed, \eqref{ec:exgtatzero} coincides exactly with the boundary conditions \eqref{ec:gtbczero} that defines $\mathsf{D}\L_\text{def}^0.$ However, we shall see that the issue is more subtle than that.

Let $\big(\hat h,\hat \Gamma\big)$ be a 2-gauge transformation localized around $z=0$, such that $\big(\hat h,\hat \Gamma\big) = (1,0)$ in a small open neighborhood around $z=\infty$ (cf. "Case 1" analyzed above). Provided such a finite 2-gauge transformation satisfies the boundary condition \eqref{ec:gtbczero}, we can take a limit $\big(\hat h,\hat\Gamma\big)\xrightarrow[z\rightarrow 0]{} \big(\tilde h,\tilde \Gamma\big)$ to arrive at a symmetry transformation which is precisely \eqref{ec:rightaction}. The characterization \eqref{defectgaugeconstraints} then tells us that $\tilde \Gamma_1,\tilde \Gamma_2$ are valued in the maximal Abelian subaglebra $\t\subset\h$ in order to preserve the bulk h2CS action. 

\medskip

Now where does the left-acting symmetries \eqref{ec:leftaction} come from? To answer this question, we turn our attention to those 2-gauge transformations $(\hat h',\hat \Gamma')$ that are localized around $z=\infty$ and satisfying the boundary conditions \eqref{ec:gtbcinfty}. As we have already mentioned in the "Case 2" analysis above, the subtlety here is that such transformations must preserve the values of the bulk fields $(\hat g,\hat \Theta)$ at $z=\infty$. If $(\hat h',\hat \Gamma')=(1,0)$ is trivial in a small neighborhood around $z=0$, this forces us to simultaneously perform an {\it internal} gauge transformation in order to compensate.

We thus perform an external and internal transformation simultaneously 
\begin{align}
    \hat g &\mapsto u^{-1}\hat g \hat h'\\
    \hat \Theta &\mapsto  \hat h'^{-1}\rhd \hat\Theta - (u^{-1}\hat g\hat h')^{-1}\rhd\Lambda + \hat \Gamma'\,. 
\end{align}
At the puncture $z=\infty$, to preserve  $\hat g|_{z=\infty}=1$ we must fix $u=\hat h'\vert_{z=\infty}$. This choice of $u$ is $\mathbb{C}P^1$-independent and compatible with preserving the condition $L_{\bar z}=0$. Now, if we evaluate at the puncture $z=0$, recalling $\hat h'|_{z=0}=1$,  we find
\begin{equation}
    \hat g \mapsto u^{-1} \hat g|_{z=0} = \hat h'^{-1}|_{z=\infty} g
\end{equation}
where we have used the fact that $u=\hat h'|_{z=\infty}$. Similarly, at the puncture $z=\infty$ we have for $\hat \Theta$
\begin{align}
    \hat \Theta\vert_{z=\infty}&\to \hat h'^{-1}|_{z=\infty}\rhd \hat \Theta\vert_{z=\infty} - (u^{-1}\hat g\vert_{z=\infty}\hat h'\vert_{z=\infty})^{-1}\rhd\Lambda  +\hat \Gamma'\vert_{z=\infty} \nonumber\\
    & = - \Lambda + \hat \Gamma'\vert_{z=\infty}. 
\end{align}
Thus, we conclude that we must take $\Lambda =\hat \Gamma'|_{z=\infty}$ which, again, is a $\mathbb{C}P^1$-independent choice, compatible with preserving the condition $H_{i\bar z}=0$. Now, if we evaluate at $z=0$ we find, recalling again $(\hat h',\hat \Gamma')\vert_{z=0} = (1,0)$ by hypothesis,
\begin{align}
    \hat \Theta\vert_{z=0}&\to \hat h'^{-1}|_{z=0}\rhd \hat \Theta\vert_{z=0} - (u^{-1}|_{z=0}\hat g\vert_{z=0}\hat h'\vert_{z=0})^{-1}\rhd\Lambda + \hat \Gamma'\vert_{z=0} \nonumber\\
    & =\Theta - (u^{-1} g)^{-1}\rhd\Lambda = \Theta - (\hat h'\vert_{z=\infty}^{-1}g)^{-1}\rhd \hat \Gamma'\vert_{z=\infty},\nonumber
\end{align}
which indeed coincides with the left-action \eqref{ec:leftaction} such that $(\hat h',\hat\Gamma')\xrightarrow[z\rightarrow\infty]{} (h,\Gamma)$. To summarize, we have the following.

\begin{theorem}\label{descending}
    Let $-_\perp: \Omega^1(Y)\rightarrow\Omega^1_\perp(Y)$ denote the map that projects out the $\mathrm{d}x_3$-components of the 1-forms on $Y$. The 5d defect symmetries $\mathsf{D}\L_\text{sym}$ descends to symmetries $\mathsf{D}\L_{3d}$ of the 3d theory \eqref{3daction} through the diagram
\[\mathsf{D}\L_\text{sym} \cong\begin{tikzcd}
	{\mathsf{D}\L_\text{sym}^0} & {\mathsf{D}(\L_\text{sym}^0)_\perp} & {\mathsf{D}\mathring{\L}^R_{3d}} \\
	{\mathsf{D}\L_\text{sym}^\infty} & {\mathsf{D}(\L_\text{sym}^\infty)_\perp} & {\mathsf{D}{\L}^L_{3d}}
	\arrow["{-_\perp}", from=1-1, to=1-2]
	\arrow["\times"{marking, allow upside down}, draw=none, from=1-1, to=2-1]
	\arrow["{z\rightarrow 0}", from=1-2, to=1-3]
	\arrow["\times"{description}, draw=none, from=1-2, to=2-2]
	\arrow["\times"{description}, draw=none, from=1-3, to=2-3]
	\arrow["{-_\perp}", from=2-1, to=2-2]
	\arrow["z\rightarrow\infty", from=2-2, to=2-3]
\end{tikzcd} = \mathsf{D}\L_{3d} \]
\end{theorem}

We emphasize once again that our analysis above holds for any member $\mathcal{L}_\ell$ in the $S^2$-family of Lagrangian subspaces, hence there is a map
\begin{equation*}
    \mathsf{D}\L_\text{sym}(\ell)\rightarrow \mathsf{D}\L_{3d}(\ell)
\end{equation*}
that associates a 5d defect symmetry to a 3d global symmetry for each $\ell\in S^2$.

\paragraph{Misaligned chirality.} To conclude this section, we mention briefly some points of interest when the THF $\mathrm{d}x_3$ of $Y$ does {\it not} align with the chirality. WLOG suppose we take, say, $\ell = w$. In this case, the boundary conditions at $z=0$ would read $$\tilde h^{-1}\partial_{\bar w}\tilde h+\mu_1(\tilde\Gamma_{\bar w})\vert_{z=0}=0,\qquad \tilde h^{-1}\partial_3 \tilde h+\mu_1(\tilde \Gamma_3)\vert_{z=0}=0,\qquad \partial_{\bar w} \tilde \Gamma_3 - \partial_3\tilde \Gamma_{\bar w} - [\tilde\Gamma_{\bar w},\tilde \Gamma_3]\vert_{z=0}=0,$$ while those at $z=\infty$ would read $$h^{-1}\partial_wh+\mu_1\Gamma_w\vert_{z=\infty}=0,\qquad \partial_w\Gamma_{\bar w}\vert_{z=\infty} = \partial_w\Gamma_3\vert_{z=\infty} = 0.$$ 

Here, we see that if $\mu_1=0$, then the boundary conditions at $z=0$ forces $\iota_M^*h\vert_{z=0}$ to merely be {\it holomorphic}, rather than being constant. Similarly, the boundary conditions at $z=\infty$ states that $\iota_M^*\Gamma_w\vert_{z=\infty}$ is \textit{antiholomorphic}. The charges associated to these symmetries therefore acquire holomorphicity properties in the misaligned case. This echoes the statement of Hartog's theorem \cite{Alfonsi_2023}, which forces the negative modes/anti-holomorphic charges in higher-dimensional chiral currents to be split from the positive/holomorphic ones.

\section{Conservation and bordism invariance of the 2-holonomies} \label{sec:holonomies}

In \S \ref{3dift} we constructed a three-dimensional action, and we proved in \textbf{Proposition} \ref{extended2flat} that it's equations of motion are equivalent to the fake- and 2-flatness of $(L,H)$. The aim of this section is to show that this allows for the construction of conserved quantities. 

As the discussion will become quite technical, we begin with a brief informal reminder of how this works in the two-dimensional setting. Consider a theory defined on a cylinder $\Sigma = S^1\times \mathbb{R}$, and assume that its equations of motion can be recast in terms of a flatness equation for some connection $L=L_t \rd t + L_x \rd x$. Conserved quantities can be constructed as follows.

Given a loop $\gamma$ on $\Sigma$, the 1-holonomy of the connection $L$, which is the solution to the parallel transport along $\gamma$, is given by the path-ordered exponential 
\begin{equation}
\label{ec:1holonomy}
    W_\gamma = P \exp \left(-\int_\gamma L\right) \,.
\end{equation}
The flatness of $L$ implies that the 1-holonomy is independent of the path, as long as the endpoints are kept fixed. Formally, we say that it depends solely on the homotopy class relative base point. Intuitively, this can be understood, in the Abelian case, through Stokes theorem. Indeed, if $\gamma_0$ and $\gamma_1$ are two paths with the same endpoints and $\Gamma$ is the surface defined by the smooth homotopy from $\gamma_0$ to $\gamma_1$ we have
\begin{equation}
    \int_{\gamma_0} L - \int_{\gamma_1}L = \int_\Gamma \mathrm{d}L = 0 \quad \Longrightarrow \quad W_{\gamma_0} = W_{\gamma_1}\,.
\end{equation}
We can use this invariance to construct conserved quantities. Indeed, let us consider at a time $t_0$ a loop $\gamma_0$ starting and ending at $x_0$, and at a later time $t_1$ a loop $\gamma_1$ starting and ending at $x_0$ as well. Note that since we are at two different time slices, $\gamma_0$ and $\gamma_1$ are not homotopical relative base point and thus $W_{\gamma_0}\neq W_{\gamma_1}$. However, we can  consider a path $\gamma_{t}$ connecting $(x_0,t_0)$ with $(x_0,t_1)$, such that the path concatenation $\gamma_t\ast\gamma_0\ast \gamma_t^{-1}$ is a loop based at $(x_0,t_1)$, see fig. \ref{fig:pathconc}.  

\begin{figure}[h]
    \centering
    \begin{tikzpicture}
        \draw[thick] (0, 1) ellipse (2cm and 0.7cm);

        \draw[thick] (0, -1) ellipse (2cm and 0.7cm);

        \draw[thick, purple] (-2, -1) -- (-2, 1);
        \draw[thick, purple, ->] (-2, -1) -- (-2, 0);  

        \node at (2.5, 1) {$\gamma_1$};
        \node at (2.5, -1) {$\gamma_0$};
        \node[purple] at (-2.3, 0) {$\gamma_t$};

        \node at (-2.7, 1) {$(x_0,t_1)$};
        \node at (-2.7, -1) {$(x_0,t_0)$};

        \filldraw[purple] (-2, 1) circle (2pt);
        \filldraw[purple] (-2, -1) circle (2pt);
    \end{tikzpicture}
    \caption{Schematically, we may understand the path concatenation as follows. We start at $(x_0,t_1)$ we go back to $(x_0,t_0)$ with $\gamma_t^{-1}$, then we go around $(x_0,t_0)$ with $\gamma_0$ and then we go back to $(x_0,t_1)$ with $\gamma_t$. This defines a loop based at $(x_0,t_1)$.}
    \label{fig:pathconc}
\end{figure}
Given that $\gamma_t\ast\gamma_0\ast \gamma_t^{-1}$ and $\gamma_1$ are homotopical relative base point, flatness of the connection implies $W_{\gamma_t\ast\gamma_{0}\ast \gamma_t^{-1}} = W_{\gamma_1}$. Moreover, 1-holonomies are multiplicative under path concatenation, namely, 
\begin{equation}
W_{\gamma_t\ast\gamma_{0}\ast \gamma_t^{-1}}=  W_{\gamma_t}W_{\gamma_0}W_{\gamma_t}^{-1}= W_{\gamma_1}\,.
\end{equation}
Hence, for each invariant character $\chi_k$ of the group, we find that invariance of $\chi_k$ implies
\begin{equation}
    \chi_k (W_{\gamma_0}) = \chi_k (W_{\gamma_1})\,.
\end{equation}
Thus the quantity $M_k = \chi_k (W_{\gamma_0})$ is conserved.



We now wish to generalize this story to the three-dimensional setting. The mathematics involved in this procedure are subtle, but the essence is the same than in the $2$d case. We will construct $2$-holonomies, which are a higher analogue of \eqref{ec:1holonomy}, defined as a surface ordered exponential of an operator which depends on both $L$ and $H$. We will then show that fake and 2-flatness imply that the $2$-holonomy depends solely on the homotopy class relative boundary of the surface, and we will then use this fact to construct conserved quantities.

Thus, we start with our action $S_{3d}$ given in \eqref{covariant3daction}, whose equations of motion \eqref{1flat?} and \eqref{2flat} imply the fake and $2$-flatness of the $2$-connection $(L,H)$, as shown in \textbf{Proposition} \ref{extended2flat}. We will construct the aforementioned $2$-holonomies as parallel transport operators on the \textit{loop space fibration} $\Omega Y \rightarrow Y$ on $Y$, following \cite{Alvarez:1997ma}. However, our treatment differs from theirs in that our 2-dimensional surface holonomies can be non-Abelian, and are sensitive to the boundary data. 

Associated to any flat $\mathbb{G}$-connection are such 2-holonomies that we construct, and we will prove their conservation and homotopy invariance in a general context. Since we know, from {\bf Proposition \ref{extended2flat}}, that the currents $(L,H)$ we have obtained specifically from $S_{3d}$ are flat $\mathbb{G}$-connections, these are sufficient to point towards the integrability of $S_{3d}$. However, under certain circumstances, we can in fact prove a much stronger "invariance" property that these specific currents satisfy.


In the following, we shall first review some of the key homotopical aspects of the theory. Our surface holonomies/2-holonomies are constructed as a parallel transport operator defined on the space $\Omega Y$ of loops $S^1\rightarrow Y$ based at a point $y_0\in Y$ \cite{Alvarez:1997ma}. These will then be put together to prove bordism invariance in the aligned case.

\subsection{Surface holonomies from connections on loop spaces}\label{alvarez}
Given a 2-connection $(L,H)$, we can define $\mathbb{G}$-valued integrable charges $(V,W)$ satisfying the following parallel transport equations 
\begin{align}
    0 &= \frac{\rd W}{\rd t} + L_i \frac{\rd\gamma^i}{\rd t}W \implies  W = P\exp\left(-\int_\gamma L\right) \label{paratrans1}\\ 
    0&= \frac{\rd V}{\rd\tau} + \left(\int_0^1 \mathrm{d}t \,\mathcal{A}_i(t)  \frac{\rd\gamma'^i(t)}{\rd\tau}\right)V\implies V = P\exp \left(-\int_{\gamma'}\cA\right)\label{paratrans2},
\end{align}
where $\gamma: [0,1]\rightarrow Y$ is a path on $Y$ based at $y_0\in Y$ and $\gamma':[0,1]\rightarrow PY$ is a \textit{path on path space} based at the constant path $\gamma_0\in PY$ at $y_0$\footnote{This means that $\gamma_0(t)=y_0$ for all $y\in[0,1]$.}. These $(y_0,\gamma_0)$ shall be the initial data for our parallel transport equations \eqref{paratrans1}, \eqref{paratrans2}. 

Geometrically, the first equation describes the usual parallel transport on $Y$, while the second describes a parallel transport over a path $\gamma'$ in the {\it path space} $PY$ of $Y$. The quantities $\cA_i(t)$ are to be understood as the components of a connection on $PY$, defined locally (at a point $\gamma \in PY$) by the formula
\begin{equation}\label{surfacetransport}
   \cA = \int_0^1 \mathrm{d}t\, \mathcal{A}_i(t)\, \delta \gamma^i(t)= \int_0^1\mathrm{d}t\, W_\gamma^{-1}\rhd H_{ij}(t)\,\dot \gamma^j \,\delta \gamma^i(t)\,.
\end{equation}
Its holonomy 
\begin{equation}
    P\exp\left(-\int_{\gamma'}\cA\right) = S\exp\left(-\int_\Sigma (W^{-1}\rhd H)\right)
\end{equation}
serves as the definition of the surface-ordered exponential \cite{Yekuteli:2015}. The above parallel transport equations \eqref{paratrans1}, \eqref{paratrans2} are inspired from \cite{Alvarez:1997ma}, but notably, our approach does not enforce the 1-flatness of $L$. As we shall see, this is necessary for our 2-holonomies to keep track of its boundaries.


\subsubsection{Surface holonomies}

We now describe the construction of surface holonomies presented in \cite{Alvarez:1997ma}. Let $P\rightarrow Y$ be a principal $G$-bundle. We can induce a $G$-bundle $\mathcal{P}=\pi^*P\rightarrow PX$ on the path space simply by pulling-back $P$ along the path space fibration $\pi: PY\rightarrow Y$, sending a path $\gamma$ to its endpoint $\gamma(1)$. In \cite{Alvarez:1997ma}, it was shown that a $G$-connection $\mathcal{A}\in \Omega^1(PY)\otimes\g$ of $\mathcal{P}$ can be constructed from a $\g$-valued 2-form $B$ on $Y$ transforming under the adjoint representation of $G$. 

This construction is given as follows. Let $\gamma\in PY$ be a path and let $W_\gamma$ denote the parallel transport operator defined by a $G$-connection $A\in\Omega^1(Y)\otimes\g$ on $Y$ from $y_0$ to a point on the path $\gamma$. The Lie algebra valued 1-form
\begin{equation}
    \mathcal{A} =\int_0^1 \rd t\,\, (\operatorname{Ad}_{W_{\gamma(t)}}^{-1}B_{ij}(\gamma(t))) \dot{\gamma}^i\delta\gamma^j(t)  
\label{pathconn}
\end{equation}
has the correct transition laws, and hence serves as a local $G$-connection on $PY$. Here, $\delta \gamma^j(t)$ forms a basis for the cotangent space $T^*_\gamma PY$ at $\gamma\in PY$ and can be understood as the lift of the usual basis $\rd x^i$ on $T^*Y$ along the given path $\gamma$. The curvature of this $G$-connection $\mathcal{A}$ is thus computed as $\mathcal{F} = \delta\mathcal{A} + \frac{1}{2}[\mathcal{A},\mathcal{A}]$, which in terms of the 2-form $B$ reads  \cite{Alvarez:1997ma}
\begin{align}
    \mathcal{F} &=\int_0^1 \rd t\, \operatorname{Ad}_{W_{\gamma(t)}}^{-1}(\rd_AB)_{ijk} \dot{\gamma}^i\delta\gamma^j\wedge \delta\gamma^k \nonumber \\
    &\qquad -~\int_0^1\rd t\int_0^1\rd\tau  \big[\operatorname{Ad}_{W_{\gamma(t)}}^{-1}F_{ik}(\gamma(t)), \operatorname{Ad}_{W_{\gamma(\tau)}}^{-1} B_{jm}(\gamma(\tau))\big] \frac{\rd \gamma^i}{\rd t} \frac{\rd\gamma^j}{\rd\tau}\delta\gamma^k\wedge\delta\gamma^m \nonumber\\
    &\qquad +~\int_0^1\rd t\int_0^1\rd\tau  \big[\operatorname{Ad}_{W_{\gamma(t)}}^{-1}B_{ik}(\gamma(t)), \operatorname{Ad}_{W_{\gamma(\tau)}}^{-1}B_{jm}(\gamma(\tau))\big] \frac{\rd\gamma^i}{\rd t} \frac{\rd\gamma^j}{\rd\tau}\delta\gamma^k\wedge\delta\gamma^m. \label{pathcurv}
\end{align}
Due to this form of the curvature, we see that the $G$-connection $\mathcal{A}$ is in general not flat. 

We can then pull $\mathcal{A}$ back along the inclusion $\Omega Y\hookrightarrow PY$ to obtain a $G$-connection on the space $\Omega Y$ of loops based at $y_0$. Closed 2-manifolds (ie. those without boundary) smoothly submersed in $Y$ can be equivalently understood as a loop $\tilde \gamma$ in $\Omega Y$ based at the constant path $\gamma_0:t\mapsto y_0$. Let this loop $\tilde\gamma $ be obtained from lifting a loop $\gamma'$ on $Y$, then the argument for the flatness $F=0$ of the $G$-connection $A$ on $Y$ runs as follows: the values of $\mathcal{A}$ at the beginning and end of the loop are related by an adjoint action,
\begin{equation}\label{looptransport}
    \mathcal{A}[\tilde\gamma(1)] = \operatorname{Ad}_{V_{\gamma'}}^{-1}\mathcal{A}[\tilde\gamma(0)]\,,\qquad V_{\gamma'} =P\exp\left(-\int_{\gamma'}\mathcal{A}\right)\,.
\end{equation}
But given $\tilde\gamma$ is a loop, $\tilde\gamma(1) = \tilde\gamma(0)$ hence $\mathcal{A}$ is independent of the parameterization of the surface only when $W_{\gamma'} = 1$, namely $A$ is a flat $G$-connection.

This flatness condition on $A$ is very restrictive, hence we seek a generalization of this formalism of \cite{Alvarez:1997ma} that relaxes it. We shall do this in the following through the theory of \textit{higher groupoids}, and use it to prove several structural theorems about our improved 2-holonomies.



\subsubsection{2-holonomies as a map of 2-groupoids}\label{holonomies}
We now describe the construction of the higher holonomies \eqref{paratrans1}, \eqref{paratrans2} from the perspective of 2-groupoids. We point the reader to the literature \cite{Whitehead:1941,Brown,Baez:2004,Martins:2007,Kapustin:2013uxa,Ang2018,Bullivant:2017qrv,Wagemann+2021,Bochniak_2021} on 2-groups and their relevance to geometry and homotopy theory.

Let $\mathfrak{G}=\mathfrak{h}\rightarrow\mathfrak{g}$ denote a Lie 2-algebra and let $\mathbb{G}=\mathsf{H}\rightarrow G$ denote its corresponding Lie 2-group. Suppose $(L,H)$ is a 2-$\mathbb{G}$-connection. Solutions $(V,W)$ to the parallel transport equations \eqref{paratrans1}, \eqref{paratrans1} define a map $\mathsf{2Hol}$ of 2-groupoids \cite{Kim:2019owc}:
\[\begin{tikzcd}
P^2Y\rightrightarrows PY\rightrightarrows Y  \arrow[d]\\
\mathsf{H}\rtimes G\rightrightarrows G\rightrightarrows \ast
\end{tikzcd}.\]
This amounts to the following. Fix $y_0\in Y$ and we let this denote the initial condition $\gamma(0),\gamma'(0)=y_0$. Let $\Omega Y$ be the space of loops on $Y$ based at $y_0$, and $P\Omega Y$ denote the space of maps $\Sigma:[0,1]\times S^1\rightarrow Y$, equipped with source $\Sigma\mapsto \Sigma(0,-)$ and target $\Sigma\to \Sigma(1,-)$ maps, such that $\Sigma(-,0) = \gamma_0$ is the group unit (ie. the constant loop at $y_0\in Y$). We take $\mathsf{2Hol}(\ast) = y_0$, and consider the pointed 2-subgroupoid $P^2_{y_0} Y$ of the double path groupoid $P^2Y\rightrightarrows PY\rightrightarrows Y$ given by $$P\Omega Y\rightrightarrows \Omega Y\rightrightarrows \{y_0\}.$$ The statement is then that $\mathsf{2Hol}_{y_0} = (V,W): P^2_{y_0}Y \rightarrow \mathbb{G}$ is a 2-group homomorphism.

\begin{definition}
    Let $\Sigma_1,\Sigma_2:[0,1]^2\rightarrow Y$ denote two (smooth) surfaces in $Y$ with the same boundary $\partial\Sigma_1 = \partial\Sigma_2 =\gamma: S^1\rightarrow Y$. A {\bf homotopy relative boundary} is a (smooth) map $F:[0,1]\times [0,1]^2\rightarrow Y$ such that
    \begin{equation*}
        F(0,-) = \Sigma_1,\qquad F(1,-) = \Sigma_2,\qquad \partial F(t,-) \cong \gamma,
    \end{equation*}
    for all $t\in[0,1]$. See fig. \ref{bdyhomotopy}. We denote homotopy classes of surfaces relative boundary by $[[0,1]^2,Y]_\text{/bdy}$.
\end{definition}

\begin{figure}[ht]
    \centering
    \includegraphics[width=0.75\columnwidth]{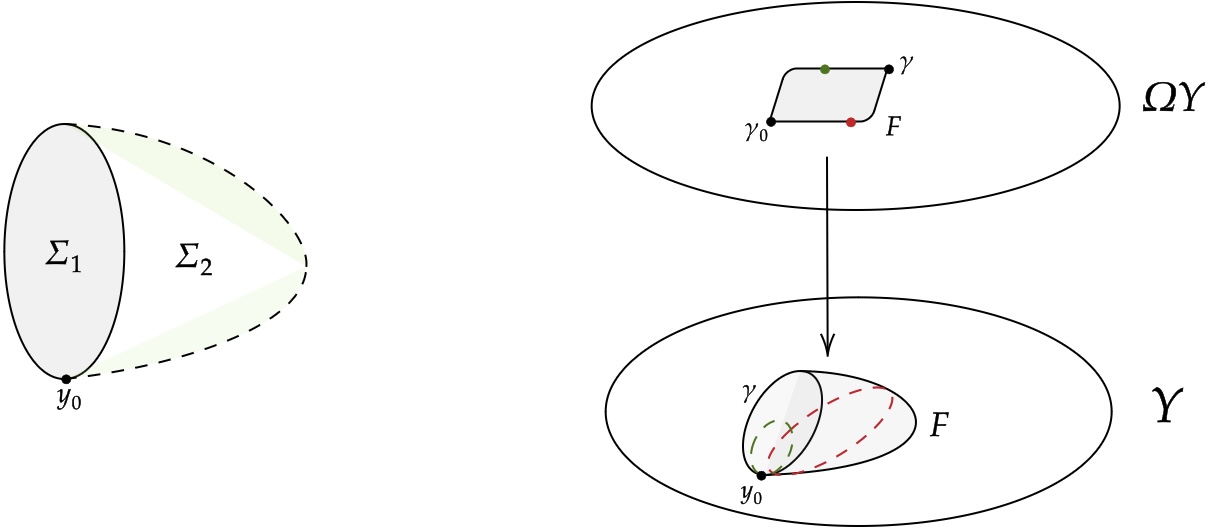}
    \caption{A homotopy relative boundary $F$ between surfaces $\Sigma_1,\Sigma_2\subset Y$ is a 3-manifold with boundary $\Sigma_1\cup_\gamma \overline{\Sigma}_2$. This boundary is the result of $\Sigma_2$ being "capped off" by $\Sigma_1$ at the loop $\gamma=\partial\Sigma_1=\partial\Sigma_2$. It can also be understood as a path homotopy on loop space $\Omega Y$. Recall $\gamma_0$ is the constant loop at $y_0$.}
    \label{bdyhomotopy}
\end{figure}

\begin{theorem}\label{relbdy}
    The fake and 2-flatness conditions
    \begin{equation*}
        \rd L + \frac{1}{2}[L, L] - \mu_1(H) = 0,\qquad \rd H + \mu_2(L,H) = 0
    \end{equation*}
    imply the 2-holonomy $V$ descends to a map $[[0,1]^2,Y]_\text{/bdy}\rightarrow \mathsf{H}$ on homotopy classes relative boundary.
\end{theorem}
\begin{proof}
    We consider a submersed 2-manifold $\Sigma\hookrightarrow Y$ with boundary as a 2-mormphsim $\Sigma\in P\Omega Y$ in the 2-group $P^2_{y_0}Y$. WLOG, we can take $\Sigma$ to have source $\gamma_0$ the constant loop on $y_0$ and target $\gamma=\partial\Sigma$ its boundary loop. We shall consider $\Sigma$ equivalently as a "path of loops" $\gamma'(\tau) = \Sigma(\tau,-) \in \Omega Y$ for each $\tau\in [0,1]$. 
    
    The goal now is to prove that fake- and 2-flatness conditions implies that $V_{\Sigma_1} = V_{\Sigma_2}$ whenever $\Sigma_1\simeq\Sigma_2$. For this, we leverage the construction of the $\mathsf{H}$-connection $\mathcal{A}$ on $PY$. This is defined analogous to \eqref{pathconn}, but instead leveraging the group action $\rhd: G\rightarrow\operatorname{Aut}\mathsf{H}$ such that
    \begin{equation*}
        \mathcal{A} =\int_0^1\mathrm{d}t\, \big(W_{\gamma(t)}^{-1}\rhd H_{ij}(\gamma(t))\big)\,\dot \gamma^i\,\delta \gamma^j\,,\qquad W_\gamma = P\exp\left(-\int_\gamma L\right)\,.
    \end{equation*}
    The holonomy is then constructed as usual $V_{\Sigma_1}= P\exp\left(-\int_{\gamma'_1}\mathcal{A}\right)$ where $\gamma'_1:[0,1]\rightarrow \Omega Y$ is the path on loop space corresponding to $\Sigma_1$. A homotopy $F:\Sigma_1\Rightarrow\Sigma_2$ relative boundary can then be understood as a homotopy $[0,1]^2\rightarrow\Omega Y$ between the paths $\gamma_1',\gamma_2'$ such that $\partial F(s,-) \cong \gamma=\partial\Sigma_1=\partial\Sigma_2$ for each $s\in [0,1]$; see fig. \ref{bdyhomotopy}. 
    
    This homotopy $F$ defines a contractible surface $\Gamma\subset\Omega Y$, or equivalently a closed 3-submanifold $S \subset Y$ whose boundary is the gluing $\overline{\Sigma}_2\cup_\gamma \Sigma_1= \overline{\Sigma}_2\#\Sigma_1$ of the two surfaces $\Sigma_1,\Sigma_2$ along $\gamma$, where $\overline{\Sigma}$ denotes the orientation reversal of a surface $\Sigma$. By usual computations, we have (cf. \S 2 of \cite{Alvarez:1997ma})
    \begin{equation}
    \label{Voflgue}
        V_{\overline{\Sigma}_2\cup_\gamma\Sigma_1} = P\exp\left(-\int_{\gamma_2'^{-1}\ast\gamma'_1}\mathcal{A}\right) = P\exp\left(-\int_\Gamma W^{-1}\rhd \mathcal{F}\right),
    \end{equation}
    where $\mathcal{F}=\delta\mathcal{A}+ \frac{1}{2}[\mathcal{A},\mathcal{A}]$ is the curvature, which reads analogously to \eqref{pathcurv} but now instead has the form
    \begin{align}\label{fakecurve}
        \mathcal{F}  &= \int_0^1 \rd t\, \operatorname{Ad}_{W_{\gamma(t)}}^{-1}(\rd_LH(t))_{ijk}\, \frac{\rd\gamma}{\rd t}^i\delta\gamma^j(t)\wedge \delta\gamma^k(t) \\
        & -~\int_0^1\rd t\int_{0}^1 \rd t' \, \left[\operatorname{Ad}_{W_\gamma(t)}^{-1}(F(L(t))_{ik} - \mu_1(H_{ik}(t)))\right] \rhd (W_{\gamma(t')}^{-1}\rhd H_{jl}(t'))\frac{\rd \gamma^i}{\rd t}\frac{\rd \gamma^j}{\rd t'}\delta \gamma^k(t)\wedge\delta\gamma^{l}(t')
    \end{align}
    where we have used the Peiffer identity. This quantity vanishes precisely when $(L,H)$ satisfies fake- and 2-flatness conditions, whence $V_{\overline{\Sigma}_1\cup_\gamma\Sigma_2} = 1$. By definition, $V_{\overline{\Sigma}_1\#\Sigma_2} = V^{-1}_{\Sigma_1} V_{\Sigma_2}=1$ and hence $V_{\Sigma_1}=V_{\Sigma_2}$ for homotopically equivalent surfaces $\Sigma_1\simeq\Sigma_2$ relative boundary, as desired.

    

\end{proof}

Note homotopies relative boundary between surfaces without boundary reduces to homotopies in the usual sense.

\begin{corollary}
    The 2-holonomy $V$ is consistent with the Eckmann-Hilton argument.
\end{corollary}
\begin{proof}
    The Eckmann-Hilton argument states that charge operators attached to closed submanifolds of codimension larger than one must have commutative fusion rules \cite{Gaiotto:2014kfa}. This comes from the fact that the concatenation of maps $\Sigma_1,\Sigma_2:S^2\rightarrow Y$, given by connected summation $\Sigma_1\#\Sigma_2$, is commutative up to homotopy: there exists a closed contractible 3-submanifold $S\subset Y$ whose boundary is $\partial S = \overline{\Sigma_2\# \Sigma_1} \cup \Sigma_1\#\Sigma_2$. 
    
    The 2-holonomies $V$ are precisely such 2-codimensional operators charged under $\mathsf{H}$, hence we must prove that
    \begin{equation*}
        V_{\Sigma_1}V_{\Sigma_2} = V_{\Sigma_2}V_{\Sigma_1}
    \end{equation*}
    for closed surfaces $\Sigma_1,\Sigma_2\in P\Omega Y$. Toward this, we consider closed submersed 2-submanifolds as loops $\gamma':S^1\rightarrow \Omega Y$ on loop space, where $\gamma'(0) = \gamma'(1)=\gamma_0$ is the constant loop at $y_0$. They are naturally encoded as automorphisms of $\gamma_0$ in the 2-group $P^2_{y_0} Y$, where the constant loop $\gamma_0\in \Omega Y$ is the identity under path concatenation --- in other words, closed 2-submanifolds are precisely the kernel of the boundary map $\partial:P\Omega Y\rightarrow\Omega Y$. Now $\mathsf{2Hol}$ being a map of 2-groupoids means in particular that
    \begin{equation*}
        W\circ\partial = \bar\mu_1\circ V,
    \end{equation*}
    whence on closed 2-submanifolds $\Sigma \in \operatorname{ker}\partial$ the surface holonomym is Abelian,
    \begin{equation*}
        1=W_{\gamma_0} = \bar\mu_1(V_\Sigma)\implies V_\Sigma\in\operatorname{ker}\bar\mu_1\subset\mathsf{H},
    \end{equation*}
     as desired.
    
\end{proof}

We recognize that, in the approach where $\mathcal{A}$ was obtained from pulling-back the $G$-connection on $Y$, \S 3.2.2 of \cite{Alvarez:1997ma} has derived the condition $[B,B]=0$ as part of local integrability. This restricts one to the case where the structure group is a semi-direct product $R[1]\rtimes G$ with an Abelian $G$-module $R$ in degree-(-1), which corresponds precisely to the case $\mu_1=0$ in our setup.


\subsubsection{Whiskering of 2-holonomies}\label{sec:whisker}
Now let us examine the operation of {\bf whiskering} \cite{Baez:2004}. This is an operation in a 2-groupoid by which 1-morphisms are composed with 2-morphisms in order to change the boundary of the 2-morphisms. Geometrically in the double path 2-groupoid $P^2Y$, this can be understood as the gluing of a surface $\Sigma\in P^2 Y$ with a path $\gamma\in PY$, which we consider as an "infinitely thin surface" $1_\gamma\in P^2Y$. The boundary of the composed surface is then $$\partial(\gamma\ast\Sigma) = \gamma\ast\partial\Sigma\ast\gamma^{-1},$$ where $\ast$ is the concatenation of paths in $PY$. This serves to change the base point of $\Sigma$ from $y_0$ to $y_1$; see fig. \ref{whiskerdiag}.

\begin{figure}[ht]
    \centering
    \includegraphics[width=0.75\columnwidth]{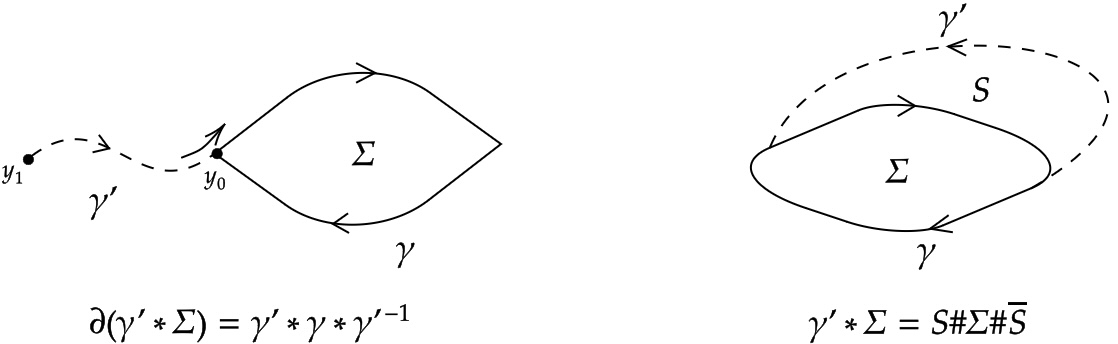}
    \caption{The pictorial representation of the operation of whiskering by $\gamma'$ (dashed). If $\gamma'$ is a path, then it serves to move the base point $y_0$ of the surface $\Sigma$. If $\gamma'$ is a loop bounding $S$, then it serves to change the shape of $\Sigma$. Recall whiskering involves a "walking back".}
    \label{whiskerdiag}
\end{figure}

In the 2-group $\mathbb{G}$, this operation is defined by a group action $\rhd: G\rightarrow\operatorname{Aut}\mathsf{H}$ satisfying the equivariance condition $\bar\mu_1(x\rhd y) = x\bar\mu_1(y)x^{-1}$ for each $x\in G,y\in\mathsf{H}$. The fact that the holonomies $(V,W)$ define a 2-groupoid map $\mathsf{2Hol}: P^2Y\rightarrow B\mathbb{G}$ means that whiskering is preserved: for each path $\gamma\in P Y$ based at the initial condition $y_0$, we have
\begin{equation}\label{whisker}
    W_{\gamma\ast \Sigma} = V_\gamma \rhd W_\Sigma,\qquad \forall ~\Sigma\in P\Omega Y.
\end{equation}
This is the statement that changing the base point of $\Sigma$ from $y_0$ to the endpoint $y_1$ of $\gamma$ amounts to a parallel transport of the holonomies $V$. This structure is called a {\bf balloon} in \cite{Yekuteli:2015}.

For (contractible) loops $\gamma\in\Omega Y$ with a non-empty intersection $\gamma\cap \partial\Sigma\neq \emptyset$, we can use the surface $S$ it encloses to deduce $$\gamma\ast \Sigma = S\#\Sigma\# \overline{S},$$ where we have glued copies of $S$ onto $\Sigma$ along the intersection $\gamma\cap \partial\Sigma$; see fig. \ref{whiskerdiag}. Even though $\gamma$ does not change the base point of $\Sigma$, it changes the shape of the boundary $\partial \Sigma$. Since the holonomy $V$ is only invariant under homotopies \textit{relative boundary}, this induces a non-trivial action
\begin{equation*}
    V_{\gamma\ast \Sigma} = W_{\partial S} \rhd V_\Sigma = \bar\mu_1(V_S)\rhd V_\Sigma = \operatorname{Ad}_{V_S}V_\Sigma,
\end{equation*}
where we have used the Peiffer identity. Further, if $S$ is a {\it closed} 2-submanifold (whose boundary $\partial S = \gamma_0$ just the base point $y_0$ given by the constant loop), then $V_S\in\operatorname{ker}\bar\mu_1$ whence $V_{\gamma\ast\Sigma} = V_\Sigma$. 

This means that, in our setup, the holonomy $V_\Sigma$ indeed {does} depend on the way by which we scan the surface $\Sigma$, in so far as $\Sigma$ has boundary. However, this is consistent with the geometry: a homotopy of $\Sigma$ is {\it not} in general a homotopy relative boundary, and the non-commutativity of the surface holonomies $V_\Sigma$ comes precisely from its boundary.


\subsection{Proof of conservation and bordism invariance of \texorpdfstring{$V$}{V}}\label{integrability}
With the above theory in hand, we are then ready to examine the conservation that arises from fake- and 2-flatness. In \S \ref{holonomies}, we proved that the surface holonomy \eqref{paratrans2} defined a group homomorphism
\begin{equation*}
     [[0,1]^2,Y]_\text{/bdy} \rightarrow \mathsf{H},
\end{equation*}
where $[[0,1]^2,Y]_\text{/bdy}$ denotes the double path 2-groupoid $P^2Y\rightrightarrows PY\rightrightarrows Y$ modulo homotopies relative boundary. This is intimately related to the conservation of surface charges based on $V$. To see this, we follow a geometric argument analogous to that given in \cite{Alvarez:1997ma}. 

Let us introduce a transverse (real or holomorphic) foliation $\Phi$, which determines locally a direction that denotes the "time" $\mathrm{d}u$ along the leaves. The goal is to compare the values of the surface charge $V_{\Sigma_0},V_{\Sigma_1}$ localized on surfaces, with diffeomorphic boundaries $\partial\Sigma_0\cong \partial\Sigma_1$, at two time-slices: $\Sigma_0$ at $u=0$ with base point $y_0 = (m_0,0)\in Y$ and $\Sigma_1$ at $u=1$ with base point $y_1 = (m_1,0)\in Y$. Take $\gamma_u:[0,1]\rightarrow Y$ to be a path along the $u$-direction, connecting the two base points $y_0,y_1$. 

Note here crucially that the fact that our surface holonomy defines invariants of homotopy relative boundary does {\it not} imply $V_{\Sigma_0} = V_{\Sigma_1}$, as $\Sigma_0,\Sigma_1$ have different base points, separated in the $u$-direction. As we have described in \S \ref{sec:whisker}, the way to transport the base point of $\Sigma_0$ to $u=1$ is by whiskering along the path $\gamma_u$,
\begin{equation*}
    \tilde V_{\Sigma_0} = W_{\gamma_u}\rhd V_{\Sigma_0}.
\end{equation*}
Now that both $\tilde V_{\Sigma_0}$ and $V_{\Sigma_1}$ share the same base point, homotopy invariance relative boundary then finally allows us to state
\begin{equation*}
    V_{\Sigma_1} = W_{\gamma_u}\rhd V_{\Sigma_0}.
\end{equation*}

This means that our surface holonomies separated in the $u$-direction differ by a whiskering along $u$. Conserved quantities can therefore be obtained from a certain "trace/character" $\mathcal{X}$ of $V_\Sigma$ that is invariant under whiskering,\footnote{Here we make a technical note that "invariance" is usually only up to higher homotopy when categories are involved, hence $\mathcal{X}(W_{\gamma_u}\rhd V_{\Sigma_0})$ is generally merely "isomorphic" to $ \mathcal{X}(V_{\Sigma_0})$. As far as the authors know, there are at least two such notions of homotopy invariant characters on 2-groups: one is the \textit{graded} character described in \cite{Chen:2023integrable} and another is the \textit{categorical} character \cite{Ganter:2006,Bartlett:2009PhD,Sean:private,Huang:2024}.} 
\begin{equation*}
    \mathcal{X}(V_{\Sigma_1}) = \mathcal{X}(W_{\gamma_u}\rhd V_{\Sigma_0}) = \mathcal{X}(V_{\Sigma_0}).
\end{equation*}
These time-independent \textbf{higher monodromy matrices} 
\begin{equation}
    \mathcal{M}_k([\Sigma]) = \mathcal{X}_k(\mathsf{2Hol}([\Sigma]))\,\label{conservedcharges}
\end{equation}
come labelled by invariant characters $\mathcal{X}_k$ of $\mathbb{G}$, as well as {\it free} homotopy classes (ie. orbits of based homotopy classes under base point-changing whiskering operations) $[\Sigma]$ of surfaces relative boundary in $Y$.

\begin{remark}
    Note that in a similar manner to the 2d WZW construction from CS$_4$ \cite{Costello:2019tri}, the 2-Lax connection we find is independent of the spectral parameter $z\in\bbC$. This prevents us from series expanding the monodromy matrices \eqref{conservedcharges} in powers series in $z$ in order to construct, in this way, an infinite number of independently conserved quantities. Notably, in the case of the 2d WZW model, a spectral parameter dependent Lax connection can be obtained from CS$_4$ theory by performing a suitable limiting procedure \cite{Costello:2019tri}.\footnote{JL thanks Benoit Vicedo for pointing this out.} Starting with the meromorphic $1$-form $\omega=z^{-2}(z-z_0)(z-z_1)\mathrm{d}z$ for fixed $z_0,z_1 \in \mathbb{CP}^1$, together with Dirichlet boundary conditions on the gauge field at the poles, one obtains the Principal Chiral Model with WZ term, with coefficients proportional to $z_0-z_1$ and $z_0+z_1$ respectively. In the $z_0\to 0$ limit, the 2d action becomes the WZW CFT, but remarkably, the Lax connection obtained in this limit is now spectral parameter dependent. We expect that a $\mathbb{C}P^1$-dependent 2-Lax connection can be constructed using a similar limiting procedure, but we leave this for a future work.
\end{remark}



\medskip

In any case, we will now prove in the following that the currents $(L,H)$ in the 3d theory \eqref{covariant3daction}, which are flat $\mathbb{G}$-connections by {\bf Proposition \ref{extended2flat}}, can in fact be made to satisfy a stronger form of "conservation" and topological invariance, called bordism invariance.

\begin{figure}[h]
    \centering
    \includegraphics[width=0.7\columnwidth]{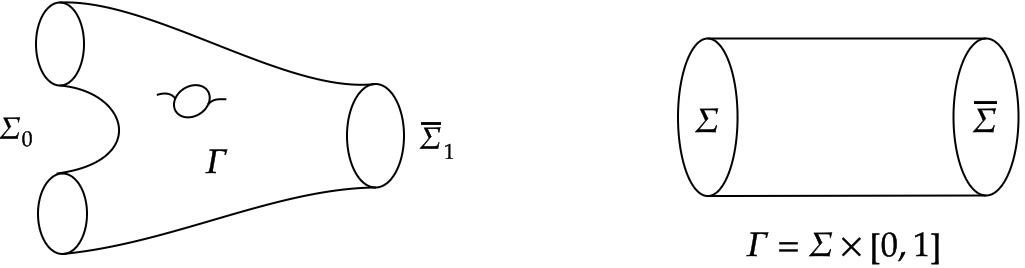}
    \caption{One can in principle define bordisms $\Gamma :\Sigma_0\rightarrow\overline{\Sigma}_1$ with non-trivial topology, or in between surfaces with different number of components, as seen in the left. The right side exhibits the trivial bordism $\Gamma=\Sigma\times [0,1]$ over a surface $\Sigma$.}
    \label{bord}
\end{figure}


\paragraph{Bordism invariance.} By bordism, we mean precisely the following.
\begin{definition}\label{bordbdry}
    A surface $\Sigma\subset Y$ is \textbf{framed} if it is equipped with a trivialization of its normal bundle $N\Sigma \cong TY/T\Sigma$. Let $\Sigma,\Sigma'$ denote two smooth framed 2-manifolds with boundary, a {\bf (smooth) framed open bordism} $\Gamma: \Sigma\rightarrow\Sigma'$ is a smooth 3-manifold $\Gamma$ with boundary (framed diffeomorphic to) $\Sigma\coprod \overline{\Sigma}'$, where $\overline{\Sigma}'$ is the orientation reversal of $\Sigma'$. 
\end{definition}
\noindent This is a stronger statement than the conservation of the holonomies. Indeed, bordism invariance is a global statement: conservation can be thought of as being invariant under "very small" bordisms, which is always trivial $\Gamma \cong \Sigma\times [0,1]$. See fig. \ref{bord}.

\begin{remark}
    Conversely, one can always "shrink the cylinder" in a trivial (open) bordism $\Gamma:\Sigma\rightarrow\Sigma'$ to give a homotopy equivalence $\Sigma\simeq\Sigma'$ (relative boundary). One can then ask if these are the only bordisms that do such a thing; for instance, are \textit{$h$-cobordisms} $\Gamma:\Sigma\rightarrow\Sigma'$, where the inclusions $\Sigma,\Sigma'\hookrightarrow\Gamma$ induce homotopy equivalences, homotopic to the trivial bordism, and hence also give rise to homotopy equivalences? This is the so-called \textbf{smooth $h$-cobordism theorem} \cite{Freedman:1990,Smale:1962,Milnor+1965}, and it turned out to be an extremely difficult problem. At dimension three, it is equivalent to the Poincar{\'e} conjecture.
\end{remark}

Let us first setup the geometry. Let $\Gamma: \Sigma\rightarrow\Sigma'$ denote a bordism between the two surfaces $\Sigma_0,\Sigma_1$, based respectively at $(m_0,0),(m_1,1)\in Y$. The statement of bordism invariance is then given by the condition $V_{\partial\Gamma}=1$. The boundary $\partial\Gamma$ has three pieces: the "bottom cap" $\Sigma_0=\Sigma \times \{0\}$, the "top cap" $\Sigma_1=\Sigma'\times\{1\}$ and the "cylinder tube" $C\cong \partial\Sigma \times\gamma$, hence we must have
\begin{equation}\label{stationarity}
    V_{\overline{\Sigma}_1} V_{C}V_{\Sigma_0} = V^{-1}_{\Sigma_1} V_C V_{\Sigma_0} =(W_\gamma^{-1}\rhd V_{{\Sigma}'}^{-1})V_C V_\Sigma=1,
\end{equation}
where $\overline{\Sigma}$ denotes the orientation reversal of $\Sigma$, and we have used a whiskering along a timelike curve $\gamma$ connecting the base points $(m_0,0),(m'_0,1)$ of the surfaces $\Sigma_0,\Sigma_1$, as in the argument for the conservation of the higher monodromy matrices \eqref{conservedcharges} above. See fig. \ref{fig:bord}.


\begin{figure}[ht]
    \centering
    \includegraphics[width=0.7\columnwidth]{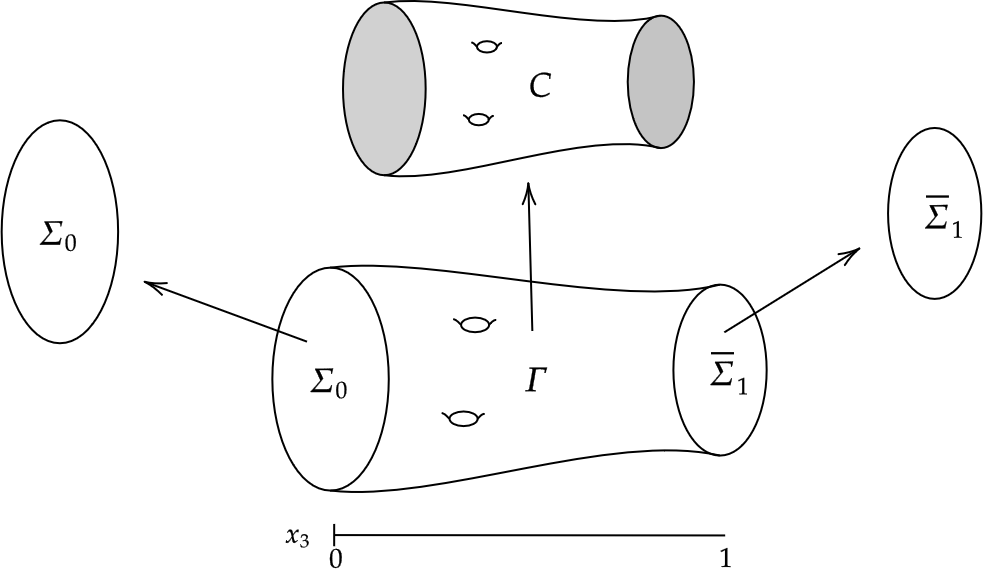}
    \caption{A diagram of a bordism $\Gamma:\Sigma_0\rightarrow \Sigma_1$ and its boundary components $\overline{\Sigma}_1,\Sigma_0,C$.}
    \label{fig:bord}
\end{figure}

Let $\Omega^\text{fr;o}_{2}(Y)_\text{/bdy}$ denote the set (Abelian group under disjoint union) of equivalence classes of surfaces in $Y$ with the diffeomorphic boundary, subject to the relation $\Sigma\sim\Sigma'$ iff there exists a framed open bordism $\Gamma$ between them as in {\bf Definition \ref{bordbdry}}. We now prove the following.
\begin{theorem}\label{bordinv}
    If the chirality $\mathrm{d}\ell$ aligns with the direction $\mathrm{d}u$ of the foliation $\Phi$, then the 2-holonomy $V$ arising from the 3d theory \eqref{covariant3daction} descends to a map $\Omega^\text{fr;o}_2(Y)_\text{/bdy}\rightarrow\mathsf{H}$ on framed open bordism classes.
\end{theorem}
\begin{proof}
    Given the covariance of \eqref{covariant3daction}, we can WLOG choose the chirality $\mathrm{d}\ell = \mathrm{d}x_3$ along the third coordinate. By hypothesis, this identifies $\mathrm{d}x_3$ with the time direction $\mathrm{d}u$ and $\partial_\ell = \partial_u$. 

    Now recall from the proof of {\bf Proposition \ref{extended2flat}} that the Lax 2-connection $(L,H)$ we obtain from $S_{3d}$ has $L_3=0, H_1=0, H_2=0$ when $\mathrm{d}\ell=\mathrm{d}x_3$. Since we have identified the third coordinate with $u\in\mathbb{R}$, this gives
    \begin{equation*}
        L_u = 0,\qquad \iota_{\partial_u}H=0.
    \end{equation*}
    The fact that $L_u=0$ means that the $W_\gamma=1$ is trivial along timelike curves $\gamma$, and the fact that $\iota_{\partial_u}H=0$ means that $V_C=1$ is trivial on surfaces $C$ whose normal $\hat n$ is spacelike --- such as the cyinder tube piece $C$ of the surface $\partial\Gamma$. This reduces the left-hand side of \eqref{stationarity} down to
    \begin{equation*}
        V_{\Sigma'}^{-1}V_\Sigma=V_{\overline{\Sigma}'}V_\Sigma = V_{\overline{\Sigma}'\# \Sigma},
    \end{equation*}
    where we have used the fact that $V$ respects the gluing of surfaces (cf. \S \ref{holonomies}). However, now that $\Sigma'$ and $\Sigma$ share the same time-slice $u=0$, the surface $\overline{\Sigma}'\#\Sigma$ bounds a 3-submanifold $S$ which is a homotopy relative boundary between $\Sigma'$ and $\Sigma$. We then have 
    $$V_{\overline{\Sigma}'\#\Sigma} = V_{\partial S}=1$$
    by {\bf Theorem \ref{relbdy}}. This proves \eqref{stationarity}.
\end{proof}

When the two directionalities misalign, however, such as in the case of the Chern-Simons/matter theory \eqref{CSmatter}, then the 2-holonomies are merely invariant under boundary-preserving deformations, but are not framed open bordism invariants. From the QFT literature \cite{Segal1985,Witten:2019,Freed:2014}, bordism invariants define non-perturbative anomalies; a similar sentiment is echoed in \cite{Baez:1995ph} (or {\it Remark \ref{homtrans}}) for the 4d Crane-Yetter-Broda TQFT. An observation {\it Remark \ref{alignedtrvi}} that we shall make in the next section, where we study the homotopy Lie algebra structure of the currents, may provide an explanation for this phenomenon.

\section{Topological-holomorphic 2+1d current algebra}\label{dgaffinecurrents}
Let us now investigate in more detail the algebra of currents in the boundary 3d action $S_{3d}$ given in \eqref{covariant3daction}. Recall that in 4d Chern-Simons theory with the disorder operator $\omega=z^{-1}\rd z$ and chiral boundary conditions, one obtains the two-dimensional Wess-Zumino-Witten CFT \cite{Costello:2019tri}.  In particular, the fields  of the theory form the affine Kac-Moody current algebra. Here, we seek to characterize the 3-dimensional homotopical analogue of these currents. 

\medskip

We begin by recalling the fields and currents in our boundary theory. In general, we have the fields $k = \mathrm{d} gg^{-1}\in\Omega^1(Y)\otimes \g,\Theta\in\Omega^1(Y)\otimes \h$ and the following currents
\begin{equation*}
    L = -k -\mu_1(\Theta^g),\qquad H = g\rhd (\rd\Theta - \Theta\wedge \Theta).
\end{equation*}
However, once we fix the Lagrangian $\mathcal{L}_\ell$ with the chirality vector $\vec\ell$, these currents are split up into two parts that are localized at the puncture $z=0$, with their dynamics governed by the boundary theory \eqref{covariant3daction}. To write down the currents, we first introduce a (not necessarily intergrable) subbundle $T^*_\ell Y\subset T^*Y$ of covectors on $Y$ tangent to $\mathrm{d}\ell$, defined by the projection map $\operatorname{proj}_\ell:T^*Y\rightarrow T^*Y_\ell$ that maps $a\mapsto \operatorname{proj}_\ell a = (\vec a\cdot \vec \ell)\vec \ell \cdot \mathrm{d}\ell$.  Since $Y$ is 3d, there is a linear isomorphism between 1-forms and 2-forms on $Y$, and we will use this fact to also define the projection map $\operatorname{proj}_\ell:\Omega^2(Y) \rightarrow \Omega^2_\ell(Y)$ on 2-forms. We can therefore express our currents as
\begin{equation}\label{currentsdefinition}
    L_\perp = -k_\perp - \mu_1(g\rhd \Theta_\perp)\,,\qquad   H_\ell = g\rhd (\operatorname{tor}\Theta_\perp - \frac{1}{2}[\Theta_\perp \times\Theta_\perp])\,,
\end{equation}
where $a_\perp = a - \operatorname{proj}_\ell a$ is the component of the 1- (or 2-)form $a$ normal to $\mathrm{d}\ell$, and $\operatorname{tor}a = \partial_\perp\times a$ denotes the scalar vorticity field of a 2d vector. The equations of motion they satisfy are given by
\begin{equation}\label{eom}
    \partial_{\ell} L_\perp =0\,,\qquad \partial_{\ell} H_\ell  = 0\,,
\end{equation}
as verified in \S \ref{3deoms}. We call the tuple $J = (L_\perp,H_\ell)$ a \textbf{2-current}. Since $\Theta_\ell$ does not at all appear in \eqref{covariant3daction}, we see that $\Theta = \Theta_\perp$, and hence we shall omit the subscript of "$\perp$" on $\Theta$ in the following.

\subsection{The dual 2-currents}
Now it is natural to expect the existence of a {\it dual} 2-current $\tilde J$ which satisfies an analogous set of fake- and 2-flatness equations of motion, and which can be constructed from $J$. This is in parallel with the case of the 2d Wess-Zumino-Witten theory, in which the right- and left-moving chiral currents satisfy dual equations of motion \cite{KNIZHNIK198483}. We will show that this is also the case for our theory \eqref{covariant3daction}.

We being by constructing the currents $\tilde J = (\tilde L_\ell,\tilde H_\perp)$, which are given by
\begin{equation}\label{currentsdefinition1}
    \tilde L_\ell = m_\ell,\qquad \tilde H_\perp = \partial_\ell \star_2\Theta + \mu_2(m_\ell,\star_2\Theta), 
\end{equation}
where $m = g^{-1}\mathrm{d} g$ is the {\it left} Maurer-Cartan form and $\star_2 \Theta$ is the 2d Hodge star operator on $\Theta\in\Omega^1_\ell(Y)$. Let us now prove the following.
\begin{proposition}
    If $J = (L_\perp,H_\ell)$ satisfies \eqref{eom}, then the dual currents $\tilde J = (\tilde L_\ell,\tilde H_\perp)$ satisfy
    \begin{equation}\label{eom1}
       \partial_\perp \tilde L_\ell = -\mu_1(\tilde H_\perp)\,,\qquad \partial_\perp\cdot \tilde H_\perp=0\,.
    \end{equation}
\end{proposition}
\noindent Notice these equations \eqref{eom1} are nothing but the fake and 2-flatness conditions for $\tilde J$.
\begin{proof}
    This is a direct computation.
    \begin{align*}
        \partial_\perp \tilde L_\ell &= \partial_\perp (g^{-1}\partial_\ell g) = \operatorname{Ad}_g^{-1}\partial_\ell (\partial_\perp gg^{-1}) \\ 
        &= -\operatorname{Ad}_g^{-1}\mu_1(\partial_\ell\Theta^g)) = -\mu_1(\mu_2(g^{-1}\partial_\ell g,\Theta) + \partial_\ell\Theta) = -\mu_1\tilde H_\perp,
    \end{align*}
    where we have used the equation $\partial_\ell L_\perp =0$ in \eqref{eom} in the second line. Next, we have
    \begin{align*}
        \partial_\perp\cdot \tilde H_\perp &= \partial_\perp \cdot (\partial_\ell \star_2 \Theta + \mu_2(g^{-1}\partial_\ell,\star_2\Theta)) \\
        &=\partial_\ell (\partial_\perp\cdot\star_2 \Theta) + \mu_2(\partial_\perp(g^{-1}\partial_\ell g)\cdot \star_2\Theta) + \mu_2(g^{-1}\partial_\ell g,\partial_\perp\cdot \star_2\Theta) \\
        &= \partial_\ell (\operatorname{tor} \Theta) - [H_\perp\cdot \star_2\Theta] + \mu_2(g^{-1}\partial_\ell g,\operatorname{tor}\Theta) \\
        &= g^{-1}\rhd (\partial_\ell g\rhd (\operatorname{tor}\Theta)) - [\partial_\ell \star_2\Theta + \mu_2(g^{-1}\partial_\ell,\star_2\Theta)\cdot \star_2\Theta]\,,
    \end{align*}
    where in the third line we have used the fact that $a_\perp\cdot\star_2 b_\perp = a_\perp\times b_\perp$ and the first equation of \eqref{eom1}, and in the fourth line we have used the definition of $\tilde H_\perp$. Consider now just the second term: recall $\star_2^2=-1$ on 1-forms in 2d, 
    \begin{align*}
        [\partial_\ell \star_2\Theta\cdot \star_2\Theta] &= \tfrac{1}{2}\partial_\ell [\star_2\Theta\cdot \star_2\Theta ] = \frac{1}{2}\partial_\ell [\Theta\times \Theta]\\
        [\mu_2(g^{-1}\partial_\ell g \,,\star_2\Theta)\cdot \star_2\Theta] &= \tfrac{1}{2}\mu_2(g^{-1}\partial_\ell g,[\Theta\times \Theta])\,,
    \end{align*}
    where the first line follows from the product rule and the second line follows from the Jacobi identity. The sum of these terms then gives
    \begin{equation*}
        \frac{1}{2}g^{-1} \rhd\partial_\ell (g\rhd [\Theta\times \Theta])\,,
    \end{equation*}
    which combines with $ g^{-1}\rhd (\partial_\ell g\rhd (\operatorname{tor}\Theta))$ to give precisely
    \begin{equation*}
         g^{-1}\rhd (\partial_\ell g\rhd (\operatorname{tor}\Theta - \frac{1}{2}[\Theta\times \Theta])) = g^{-1}\rhd \partial_\ell H_\ell = 0
    \end{equation*}
    by \eqref{eom}.
\end{proof}

We call these sets equations of motion \eqref{eom}, \eqref{eom1} the {\it 2-Lax equations} for the 2-currents $J,\tilde J$. The goal now is to study the differential graded algebraic structures of these topological-holomorphic currents. 

\subsection{Fields as derived 2-group elements in \texorpdfstring{$\mathsf{D}\L^\bullet$}{D}}\label{fieldcurrents}
Notice the 2-currents $J,\tilde J$, as far as their components are concerned, are just particular "splicing"/splitting of the full 1- and 2-forms $(L,H)$. Therefore, in order to characterize them, it suffices to characterize $(L,H)$ and then perform the slicing. We first introduce the covariant derivative $D = \rd - \mu_2(k,-)$ such that
\begin{equation*}
    g\rhd \rd \Theta = D(g\rhd\Theta)\,.
\end{equation*}
In this way, depending on how the components of the currents are sliced up for $J,\tilde J$, the 2-Lax equations \eqref{eom}, \eqref{eom1} can be written entirely in terms of the fields $(g,\Theta^g)$. This ubiquitous appearance of $\Theta^g = g\rhd\Theta$ instead of the field $\Theta$ itself, both in the actions and the currents, is not a coincidence.

This can be explained in the context of the derived superfield formulation \cite{Zucchini:2021bnn,Jurco:2018sby}. Recall the derived 2-group $\mathsf{D}\mathbb{G}$ as defined in \S \ref{sec:gsof2CS}; our fields $(g,\Theta^g)$ are degree-0 elements of the following graded algebra:
\begin{equation*}
    \mathsf{D}\mathscr{L}^\bullet=\Omega^\bullet(Y) \otimes \mathsf{D}\mathbb{G},\qquad \mathsf{D}\mathscr{L}^0 \cong (C^\infty(Y)\otimes G)\oplus (\Omega^1(Y)\otimes \h).
\end{equation*}
Due to the 2-group structure inherited from $\mathbb{G}$, this perspective allows us to use the whiskering operation (see \S \ref{sec:whisker}) to write
\[(g,g\rhd\Theta)= 
\begin{tikzcd}
	\ast &&& \ast
	\arrow[""{name=0, anchor=center, inner sep=0}, "g", curve={height=-30pt}, from=1-1, to=1-4]
	\arrow[""{name=1, anchor=center, inner sep=0}, "{\bar\mu_1\left(\exp\int g\rhd\Theta\right)g}"', curve={height=30pt}, from=1-1, to=1-4]
	\arrow["{\exp\int_\gamma g\rhd \Theta}"{description}, shorten <=8pt, shorten >=8pt, Rightarrow, from=0, to=1]
\end{tikzcd}\]
\[
 = \begin{tikzcd}
	\ast && \ast &&& \ast
	\arrow[""{name=0, anchor=center, inner sep=0}, "1", curve={height=-30pt}, from=1-3, to=1-6]
	\arrow[""{name=1, anchor=center, inner sep=0}, "{\bar\mu_1\left(\exp\int \Theta\right)}"', curve={height=30pt}, from=1-3, to=1-6]
	\arrow["g", from=1-1, to=1-3]
	\arrow["{\exp\int_\gamma \Theta}"{description}, shorten <=8pt, shorten >=8pt, Rightarrow, from=0, to=1]
\end{tikzcd} = g\rhd (1,\Theta).
\]
This explains why $\Theta^g=g\rhd\Theta$ keeps showing up in our theory \eqref{covariant3daction} --- it was encoding whiskerings by $g$! The subspace $\Omega^1\otimes\h\subset \mathsf{D}\mathscr{L}^0$, within which the field $\Theta$ itself lives, then naturally acquires the interpretation of face decorations in the derived 2-group.

\subsubsection{The differential in \texorpdfstring{$\mathsf{D}\mathscr{L}^\bullet$}{D} }
This graded algebra $\mathsf{D}\mathscr{L}^\bullet$ is equipped with the differential $\hat \rd:\mathsf{D}\mathscr{L}^\bullet\rightarrow \mathsf{D}\mathscr{L}^{\bullet+1}$ that satisfies the usual graded Leibniz rule, but also a compatibility with whiskering $\rhd$. Since $\G$ only has two terms, so does each degree of $\mathsf{D}\mathscr{L}^\bullet$. On elements $(g,g\rhd\Theta)$ of degree-0, we can explicitly write 
\begin{equation}
  \hat\rd(g,\Theta^g) = (\overrightarrow{\Delta} g - \mu_1\Theta^g,\rd\Theta^g - \Theta^g\wedge\Theta^g),\label{differential}  
\end{equation}
where $\overrightarrow{\Delta}g= -\rd gg^{-1}$ is the right Maurer-Cartan form. Note this produces the covariant derivative associated to a connection $\Theta^g$ on a $\mathsf{H}$-bundle over $X$ (recall $\mathsf{H}$ is a $G$-module).

For $g,h\in C^\infty(Y)\otimes G,$ and $ \Theta\in\Omega^1(Y)\otimes\h$, the compatibility with a whiskering by $g$ can then be written as
\begin{equation*}
    g\rhd \hat{\rd} (h,h\rhd \Theta) = \hat{D} (gh,(gh)\rhd \Theta)\,,\qquad \hat D =  (\rd_1 - [\partial gg^{-1},-]-\mu_1,\rd_2 - \mu_2(\partial gg^{-1},-)),
\end{equation*}
The nilpotency of the differential, $\hat \rd^2=0$, is equivalent to the fake- and 2-flatness of $(L,H)$. This then allows us to identify the (off-shell) currents $$(L,H) = g\rhd \hat{\rd}(1,\Theta) = \hat D(g,\Theta^g) \in \mathsf{D}\mathscr{L}^1$$ as a degree-1 element in this derived complex. 

\subsubsection{Characterizing the currents}
To characterize the on-shell 2-currents, we slice up the dg algebra $\mathsf{D}\mathscr{L}^\bullet$ into two according to the projection $\operatorname{proj}_\ell$ and its complement. This induces a fibration 
\begin{equation*}
    \mathsf{D}\mathscr{L}_\perp^\bullet \cong (\mathsf{D}\mathscr{L}^\bullet/\mathsf{D}\mathscr{L}_\ell^\bullet)\hookrightarrow \mathsf{D}\mathscr{L}^\bullet \rightarrow\mathsf{D}\mathscr{L}^\bullet_\ell\rightarrow 0,
\end{equation*}
where $\mathsf{D}\mathscr{L}_\ell^\bullet = \Omega_\ell^\bullet(Y)\otimes \mathsf{D}\mathbb{G}$. This splits the differential $\hat \rd = \hat \rd_\perp + \hat \rd_\ell$, which are individually nilpotent $\hat \rd_\perp^2,\hat \rd_\ell^2 =0$ on the fields $(g,\Theta^g)\in \mathsf{D}\mathscr{L}^0$, as a consequence of the 2-Lax equations \eqref{eom} and \eqref{eom1}. We can then characterize the current algebras as the following subspaces:
\begin{align}
    \D &= \{J \in g\rhd \hat{\rd}_\perp(\Omega^1\otimes\h)\mid \hat \rd_\ell J =0\} \label{homotopycurrents}\\
    \tilde\D &= \{\tilde J\in g\rhd \hat{\rd}_\ell(\Omega^1\otimes\h)\mid \hat \rd_\perp \tilde J  =0\}.\label{homotopycurrents1}
\end{align}
Note $\mu_1$ appears in both $\hat \rd_\ell,\hat\rd_\perp $ due to the definition \eqref{differential}. More explicit characterizations are possible given the data of a THF $\Phi$ on $Y$; we will use $\Phi$ more explicitly later.

\subsection{Graded Lie algebra structure of the currents}
In the following, we will extract the graded Lie algebra structure of the currents. We shall do this from the "dual" perspective of the conserved Noether charges associated to the semilocal symmetries of the theory studied in \S \ref{residuals}, and analyzing their transformation properties.

Taking inspiration from covariant field space approach \cite{Geiller:2020edh}, we first smear the currents with gauge fields $\alpha \in C^\infty(Y)\otimes \g$ and $\Gamma\in\Omega^1(Y)\otimes \h$, which gives us the local Noether charges (also called "charge aspects"),
\begin{equation*}
    Q_{(\alpha,\Gamma)} = \langle (\alpha,\Gamma),J\rangle = \langle \alpha, H_\ell\rangle + \langle \Gamma,L_\perp\rangle \,.
\end{equation*}
For the currents $J$ of the theory \eqref{covariant3daction}, the above quantities are 2-forms that live in the image of the projection $\operatorname{proj}_\ell$ --- that is, they have no legs along $\mathrm{d}\ell$. On-shell of the equations of motion \eqref{eom}, a simple computation gives
\begin{equation*}
    \rd Q_{(\alpha,\Gamma)} =\langle \partial_\ell\alpha,H_\ell\rangle + \langle \partial_\ell\Gamma,L_\perp\rangle
\end{equation*}
which vanishes provided we impose the boundary conditions  
\begin{equation}\label{bc1}
    \partial_\ell\alpha=0\,,\qquad \partial_\ell\Gamma =0\,.
\end{equation}
In other words, the currents $Q_{(\alpha,\Gamma)}$ generate the infinitesimal {\it left}-acting symmetries $\L^L_{3d}(\ell)$ of the theory \eqref{covariant3daction}. On the other hand, for the 2-current $\tilde J$, the local charge aspects 
\begin{equation*}
    \tilde Q_{(\alpha,\Gamma)} = \langle(\alpha,\Gamma),\tilde J\rangle = \langle \alpha,\tilde  H_\perp\rangle + \langle \Gamma,\tilde L_\ell\rangle 
\end{equation*}
are 2-forms living in the kernel of $\operatorname{proj}_\ell$ --- meaning they have only legs including $\mathrm{d}\ell$. On-shell of \eqref{eom1}, we find
\begin{align*}
    \rd\tilde Q &= \langle \partial_\perp \alpha, \tilde H_{\perp}\rangle + \langle \partial_\perp \cdot\star_2\Gamma,\tilde L_\ell\rangle - \langle \Gamma,\partial_\perp \tilde L_\ell\rangle  \\ 
    &= \langle \partial_\perp \alpha + \mu_1\Gamma,\tilde H_\perp\rangle + \langle \operatorname{tor}\Gamma, \tilde L_\ell\rangle\,,
\end{align*}
which vanishes provided we impose the boundary conditions\footnote{Recall we have without loss of generality (cf. \S \ref{residuals}) assumed that $\Gamma$ has no component along the chiral direction $\mathrm{d}\ell$, and hence $\Gamma_\perp=\Gamma$.}
\begin{equation}\label{bc2}
    \partial_\perp \alpha + \mu_1\Gamma=0\,,\qquad \operatorname{tor} \Gamma=0\,.
\end{equation}
In other words, the currents $\tilde Q_{(\alpha,\Gamma)}$ generate the (infinitesimal) {\it right}-acting symmetries $\mathring{\L}_{3d}^R(\ell)$ of the theory \eqref{covariant3daction}.

We now prove the following.
\begin{theorem}\label{gradedbrackets}
    The above conserved Noether charge aspects $Q,\tilde Q$ satisfy the following graded algebra structure:
    \begin{align}
            [Q_{(\alpha,\Gamma)},Q_{(\alpha',\Gamma')}] &= -Q_{([\alpha',\alpha],\mu_2(\alpha',\Gamma)-\mu_2(\alpha,\Gamma'))} \nonumber\\
    &\qquad -\langle \Gamma,\partial \alpha' + \mu_1\Gamma'\rangle + \langle \alpha,\operatorname{tor}\Gamma'\rangle \label{bracket} \\
        [\tilde Q_{(\alpha',\Gamma')},\tilde Q_{(\alpha,\Gamma)}] &= -\tilde Q_{([\alpha,\alpha'],\mu_2(\alpha,\Gamma')-\mu_2(\alpha',\Gamma))}  \nonumber\\
    &\qquad - \langle \Gamma',\partial_\ell\alpha\rangle +\langle \alpha',\partial_\ell \Gamma\rangle,.\label{bracket1}
    \end{align}
    Moreover, $[Q,\tilde Q]=0.$
\end{theorem}
\begin{proof}
    Following \cite{KNIZHNIK198483}, we proceed with the computation of the (graded) Lie algebra bracket between the charges using the general formula \cite{Julia:2002df}
\begin{equation}
    [Q_{(\alpha,\Gamma)}, Q_{(\alpha',\Gamma')}] = -\delta_{(\alpha',\Gamma')}Q_{(\alpha,\Gamma)}\,,
\end{equation}
where $\delta_{(\alpha,\Gamma)}$ denotes a gauge variation. Let us consider first the gauge variations of $Q\in\hat\D$, which is the Noether charges constructed from the 2-current $J = (L_\perp,H_\ell)$. Recall these charges generate the infinitesimal left variation,
\begin{equation}
\label{ec:infleftaction}
    \delta_{\alpha'} g = -{\alpha'}\cdot g,\qquad \delta_{(\alpha',\Gamma')} \Theta = -g^{-1}\rhd \Gamma',
\end{equation}
we find that
\begin{align*}
    \delta_{(\alpha',\Gamma')} L_\perp &= - [\alpha', L_\perp]+\partial_\perp \alpha' + \mu_1 \Gamma',\\
   \delta_{(\alpha',\Gamma')}  H_\ell &= - \mu_2(\alpha', H_\ell) - \mu_2( L_\perp, \Gamma') - \operatorname{tor}\Gamma',
\end{align*}
in which $(\alpha',\Gamma')$ parameterize the symmetries satisfying the boundary conditions \eqref{bc1}. This then leads to precisely \eqref{bracket}.

Now analogously, we compute the bracket on the dual Noether charges $\tilde Q\in \hat{\tilde \D}$ constructed from $\tilde J = (\tilde L_\ell,\tilde H_\perp)$. Recalling these generate the infinitesimal right-variation
\begin{equation}
\label{ec:infrightaction}
    \tilde\delta_{\alpha} g = g\cdot\alpha,\qquad \tilde\delta_{(\alpha,\Gamma)} \Theta = -\mu_2(\alpha,\Theta) + \Gamma,
\end{equation}
we can then compute
\begin{align*}
    \tilde\delta_\alpha \tilde L_\ell &= -[\alpha,\tilde L_\ell] + \partial_\ell \alpha, \\ 
    \tilde\delta_{(\alpha,\Gamma)} \tilde H_\perp &= -\mu_2(\alpha,\tilde H_\perp) -\mu_2(\tilde L_\ell,\Gamma) + \partial_\ell\Gamma\,,
\end{align*}
where $(\alpha,\Gamma)$ satisfy the boundary conditions \eqref{bc2}. These then give rise to the brackets \eqref{bracket1}. Finally, to show that $[Q,\tilde Q]=0$ we compute
\begin{equation}
\begin{aligned}
    \tilde \delta_{\alpha'} L_\perp &=-\mathrm{Ad}_g(\partial_\perp \alpha'+\mu_1(\Gamma')) \\
    \tilde \delta_{(\alpha',\Gamma')}H_\ell &= g\rhd \mathrm{tor} \, \Gamma'
\end{aligned}
\end{equation}
where in the variation of $H_\ell$ we have used the graded Jacobi identity and the first boundary condition in \eqref{bc2}. In particular, we note that both of the above vanish due to the boundary conditions \eqref{bc2}. An entirely analogous computation shows that $\delta_\alpha \tilde L_\ell =0$ and $\delta_{(\alpha,\Gamma)}\tilde H_\perp=0$. 
\end{proof}

Notice these two brackets \eqref{bracket}, \eqref{bracket1} take identical forms. In particular, the central terms in both of these expressions can be written collectively as
\begin{equation}
    \langle (\alpha,\Gamma),\hat \rd(\alpha',\Gamma')\rangle\,, \label{central}
\end{equation}
depending on the boundary conditions \eqref{bc1}, \eqref{bc2} that are satisfied by $(\alpha,\Gamma)$ (recall $\Gamma_\ell=0$). Note $\hat \rd$ is the differential on $\D,\tilde\D$ defined in \eqref{differential}, which includes the map $\mu_1$. 

These structures run in complete parallel to the story in the 2d WZW model \cite{KNIZHNIK198483}, where the dual charges satisfy identical algebras and they do not have mutual brackets. Moreover, we will show in the following that \eqref{central} can in fact be identified as a Lie 2-algebra 2-cocycle characterizing a central extension sequence.

We emphasize here that $\mathrm{d}\ell$ is the direction of the chirality, so our computations so far are a priori independent of the direction of a THF $\Phi$ of $Y$. However, we will use the data of $\Phi$ in the following to construct a differential for the charge algebra $\hat\D$ and $\hat{\tilde\D}$.

\subsection{Central extensions in the affine Lie 2-algebra}
In this section, we prove that the central term \eqref{central} in fact defines a Lie 2-algebra 2-cocycle. Given an arbitrary Lie 2-algebra $\G$ and an Abelian $\G$-module $\mathfrak{V}$ --- namely a 2-term complex $\mathfrak{V} = W\xrightarrow{\phi}V$  of vector spaces \cite{Baez:2003fs} with an action $\rho: \G\rightarrow \mathfrak{gl}(\mathfrak{V})$ by $\G$ --- the (twisted) cohomology classes $H^2_\rho(\G,\mathfrak{V})$ of such 2-cocycles have been shown \cite{Angulo:2018} to classify central extension sequences of Lie 2-algebras of the form
\begin{equation*}
    \mathfrak{V}\rightarrow \hat\G\rightarrow \G
\end{equation*}
over the ground field $\mathbb{K}$ of characteristic 0. 

For our purposes, it suffices to consider $\mathfrak{V}$ as a trivial $\G$-module, where $\rho=0$. We first recall the data of a representative 2-cocycle $(s_1,s_2,s_3)$ of an extension class.
\begin{definition}
    Suppose $\mathfrak{V}$ is a trivial $\G$-module, where $\G = \h\xrightarrow{\mu_1}\g$ is a Lie 2-algebra over the ground field $\mathbb{K}$. A {\bf Lie 2-algebra 2-cocycle} $s=(s_1,s_2,s_3): \G\otimes\G\rightarrow\mathfrak{V}$ with coefficients in $\mathfrak{V}=W\xrightarrow{\phi}V$ is the data of a triple
\begin{equation*}
    s_1: \g\otimes\g\rightarrow V,\qquad s_2: \g\otimes\h\rightarrow W,\qquad s_3: \h\rightarrow V,
\end{equation*}
such that
\begin{align}
    0 &= ~\circlearrowright ~ s_1([\sfx_1,\sfx_2],\sfx_3),\label{jac2} \\
    0 &= s_2(\mu_1(\sfy_1),\sfy_2) + s_2(\mu_1(\sfy_2),\sfy_1) \label{cond1}\\
    0&=~ \circlearrowright~ s_2(\mu_1([\sfy_1,\sfy_2]),\sfy_3),\label{jac3} \\
    0&= s_1(\sfx,\mu_1(\sfy)) - \phi s_2(\sfx,\sfy) + s_3(\sfx\rhd \sfy),\label{cond2} \\
    0&= s_2([\sfx_1,\sfx_2],\sfy) - s_2(\sfx_1,\sfx_2\rhd \sfy) + s_2(\sfx_2,\sfx_1\rhd \sfy),\label{jac}\\
    0&= s_2(\sfx,[\sfy_1,\sfy_2]) - s_2(\mu_1(\sfx\rhd \sfy_1),\sfy_2) + s_2(\mu_1(\sfx\rhd \sfy_2),\sfy_1)\label{jac1}
\end{align}
for all $\sfx,\sfx_1,\sfx_2,\sfx_3\in \g$ and $\sfy,\sfy_1,\sfy_2,\sfy_3\in\h$, where $\circlearrowright$ denotes a sum over cyclic permutations.
\end{definition}
\noindent The full list of conditions for when $\mathfrak{V}$ is not a trivial $\G$-module can be found in \cite{Angulo:2018}.

\subsubsection{Constructing the differential in \texorpdfstring{$\hat{\D},\hat{\tilde{\D}}$}{D}}
To begin, we first endow a grading where the $\Gamma$-charges have degree-0 and the $\alpha$-charges have degree-1. The idea is to perform a pushforward at the differential forms factor $\Omega^\bullet$ of $\hat{\tilde{\D}}$ (or $\hat\D$). This is accomplished by viewing the THF $\Phi$ as inducing a fibration \cite{Scrdua2017OnTH}
\begin{equation*}
    \R_{\mathrm{d}x_3} \hookrightarrow Y\rightarrow M
\end{equation*}
over a (compact) complex manifold $M$, where $\mathbb{R}_{\mathrm{d}x_3}$ is the 1-dimensional fibre along the covector direction $\mathrm{d}x_3$ of the THF (which may differ from the chirality $\mathrm{d}\ell$). For simplicity, we compactify the fibre from $\mathbb{R}_{\mathrm{d}x_3}$ to a circle $S^1_{\mathrm{d}x_3}$ such that $Y$ defines a {Seifert fibration} $p:Y\rightarrow M$ \cite{Aganagic:2017tvx}, then we can perform a fibrewise integration $\pi_*:\Omega^{\bullet}(Y)\rightarrow \Omega^{\bullet-1}(M)$ such that,\footnote{This gives rise to the well-known {\it Gysin homomorphism} $\pi^!:H^\bullet(Y)\rightarrow H^{\bullet-\operatorname{rank}E}(M)$ \cite{book-charclass} on cohomology, where $E$ is the compact fibre of $Y\rightarrow M$.} on 1-forms, 
\begin{equation}\label{integrate}
    (\pi_*\omega)_{(w,\bar w)} = \int_{S^1_{\mathrm{d}x_3} = p^{-1}(w,\bar w)}\omega,\qquad \omega\in\Omega^1(Y).
\end{equation}
We make use of this map to define the differential $\hat\mu_1:\hat\D_1\rightarrow \hat\D_0$ as a pullback by $\pi_*\otimes\mu_1$,
\begin{equation*}
    (\hat\mu_1\tilde Q)_\Gamma = \tilde Q_{(\pi_*\otimes \mu_1)\Gamma} = \langle \mu_1(\pi_*\Gamma),H_\perp\rangle,
\end{equation*}
assigning a $\Gamma$-charge to an $\alpha$-charge. We construct the same differential for the dual charges $\hat\D$.

\medskip

We are finally to prove the main theorem of this section.
\begin{theorem}\label{affine2alg}
    The tuple $(\hat{\tilde{\D}},\hat\mu_1,\hat\mu_2=[-,-])$ with $[-,-]$ given by \eqref{bracket1} is a Lie 2-algebra fitting into the central extension sequence
    \begin{equation*}
        \underline{\mathbb{K}} \rightarrow \hat{\tilde\D}\rightarrow \tilde\D,
    \end{equation*} 
    iff \eqref{central} defines a Lie 2-algebra 2-cocycle associated to this extension, where $\underline{\mathbb{K}} = \mathbb{K}\xrightarrow{0}\mathbb{K}$.
\end{theorem}
\begin{proof}
    Recall $\Gamma_\ell=0$. We will first perform some general computations before we specialize $(\alpha,\Gamma)$ to the space $\mathring{\L}_{3d}^R$ parameterizing the charges $\hat{\tilde{\D}}$. From the central term \eqref{central}, we shall define
    \begin{equation}
        s_1(\Gamma,\Gamma') = \langle \Gamma,\mu_1\Gamma'\rangle,\qquad s_2(\Gamma,\alpha) = \langle \operatorname{tor}\Gamma,\alpha\rangle,\qquad s_3=0\label{2cocycle}.
    \end{equation}
    It is useful to extend $s_2$ to a skew-symmetric map by an integration by parts,
    \begin{equation*}
       s_2(\alpha,\Gamma) = \langle \partial \alpha,\Gamma\rangle =- \langle \operatorname{tor}\Gamma,\alpha\rangle = -s_2(\Gamma,\alpha).
    \end{equation*}
    As $\phi=0,s_3=0$, and since $\pi_*$ acts trivially on the $\alpha$-charges, we see that \eqref{cond2}, \eqref{jac2} are trivially satisfied.
    
    Let us now check the graded Leibniz rule,
    \begin{equation*}
        \hat\mu_1\circ[-,-] = [-,-]\circ (\hat\mu_1\otimes 1 +(-1)^\text{deg} 1\otimes \hat\mu_1).
    \end{equation*}
    From \eqref{bracket}, the left-hand side gives 
    \begin{align*}
        \hat\mu_1[Q_{(\alpha,\Gamma)},Q_{(\alpha',\Gamma')}] &= -\hat\mu_1(Q_{([\alpha',\alpha],\mu_2(\alpha',\Gamma) - \mu_2(\alpha,\Gamma'))} + \langle (\alpha,\Gamma),\hat d(\alpha',\Gamma')\rangle )\\
        &= - Q_{[\alpha',(\pi_*\otimes \mu_1)\Gamma] - [\alpha,(\pi_*\otimes \mu_1)\Gamma']}\\
        &= -Q_{[\alpha',\mu_1(\pi_*\Gamma)]} - Q_{[\alpha,\mu_1(\pi_*\Gamma')]},
    \end{align*}
    while the right-hand side gives
    \begin{align*}
        [(\hat\mu_1Q)_\Gamma,Q_{(\alpha',\Gamma')}] - [Q_{(\alpha,\Gamma)},(\hat\mu_1Q)_{\Gamma'}] &= -Q_{([\alpha',(\pi_*\otimes\mu_1)\Gamma],-\mu_2((\pi_*\otimes\mu_1)\Gamma,\Gamma'))} \\
        &\qquad + Q_{[(\pi_*\otimes\mu_1)\Gamma',\alpha], \mu_2((\pi_*\otimes\mu_1)\Gamma',\Gamma)} \\
        &\qquad +\langle \mu_1(\pi_*\Gamma),\operatorname{tor}\Gamma'\rangle - \langle \Gamma,\partial \mu_1(\pi_*\Gamma')\rangle \\
        &= -Q_{[\alpha',\mu_1(\pi_*\Gamma)]} + Q_{[\mu_1(\pi_*\Gamma'),\alpha]}  \\
       &\qquad + \langle \mu_1(\pi_*\Gamma),\operatorname{tor}\Gamma'\rangle + \langle \operatorname{tor}\Gamma,\mu_1(\pi_*\Gamma')\rangle\\
       &\qquad - (Q_{\mu_2(\mu_1(\pi_*\Gamma'),\Gamma)}+ Q_{\mu_2(\mu_1(\pi_*\Gamma),\Gamma')}).
    \end{align*}
    Now given the parameters $\Gamma$ in $\mathring{\L}^R_{3d}$ must be valued in the maximal Abelian $\t\subset\h$, the final terms involving $\mu_2(\mu_1(\pi_*\Gamma'),\Gamma) = [\pi_*\Gamma',\Gamma] =0$ vanish. The two sides then coincide iff $s_2$ \eqref{2cocycle} satisfies \eqref{cond1}.

    We now check the graded Jacobi identities. Due to the form of the bracket \eqref{bracket}, the graded Jacobi identities for the non-central extension terms follow directly from those of $\G$. The remaining terms read
    \begin{align}
        \langle [\alpha,\alpha'],\operatorname{tor}\Gamma\rangle - \langle \alpha',\operatorname{tor}\mu_2(\alpha,\Gamma)\rangle + \langle \alpha,\operatorname{tor}\mu_2(\alpha',\Gamma)\rangle &=0 \label{proofcond},
        \\
        \langle \mu_1\Gamma,\mu_2(\alpha,\Gamma')\rangle + \langle \alpha,[\Gamma',\Gamma]\rangle - \langle \mu_1\Gamma',\mu_2(\alpha,\Gamma')\rangle &=0 \nonumber,
    \end{align}
    which are precisely the conditions \eqref{jac}, \eqref{jac1} for $s$ \eqref{2cocycle}. However, due to the lack of a $\langle\Gamma,\mu_1\Gamma'\rangle$ contribution in the central term of the bracket \eqref{bracket1} for $\hat{\tilde{\D}}$, \eqref{jac1} is trivially satisfied for $\hat{\tilde\D}$. And finally, since $\Gamma$ is Abelian, \eqref{jac3} is also trivially satisfied. This completes the proof.
Hence $\check\mu_1$ indeed preserves the bracket.
\end{proof}

From \eqref{bracket}, we see that the above general computations also applies to the charges $\hat\D$. Whence, under mild assumptions,\footnote{For $\hat\D$, we must keep quadratic terms $[\Gamma,\Gamma'],\langle\Gamma,\mu_1\Gamma'\rangle$ in our computations. Assuming the quantity $[\pi_*\Gamma,\Gamma'] = \llbracket \Gamma,\Gamma'\rrbracket$ defines a Lie bracket $Q_{\llbracket \Gamma,\Gamma'\rrbracket} = [Q_{\Gamma'},Q_\Gamma]$ on the charges, \eqref{jac3} follows from \eqref{proofcond} given $\hat\mu_1$ preserves this bracket. One can check that this is indeed the case:
\begin{equation*}
        (\pi_*\otimes\mu_1)(\llbracket\Gamma,\Gamma'\rrbracket) = \int_{S^1_{\mathrm{d}x_3}} \mu_1[(\pi_*\Gamma)_{(w,\bar w)},\Gamma'_{(x_3,w',\bar w')}] = [\mu_1(\pi_*\Gamma_{(w,\bar w)}),\mu_1(\pi_*\Gamma_{(w',\bar w')})],
    \end{equation*}
    where $(w,\bar w),(w',\bar w')\in M$ are local coordinates on the base space $M$ and $x_3\in S^1_\Phi$ is along the fibre.} one can prove that it also forms a centrally extended Lie 2-algebra. We call $\hat\D,\hat{\tilde{\D}}$ the \textbf{homotopy affine Lie algebras of planar currents}. 

\begin{remark}\label{alignedtrvi}
    Recall that $\Gamma_\ell=0$. In the affine Lie 2-algebras $\hat\D,\hat{\tilde\D}$, the pushforward $\pi_*=0$ \eqref{integrate} is trivial if the THF $\Phi$ aligns with $\mathrm{d}\ell$. This renders the charge algebra $\hat\D$ both strict and skeletal, and hence lacks a degree-3 {Postnikov class} $\kappa=0$ \cite{Baez:2003fs,chen:2022,Wagemann+2021}; see also \S \ref{classification}. As one can construct gauge-invariant local perturbative observables from such a class (strictly speaking one does this from the cohomology $H(\hat\D)$ \cite{Arvanitakis:2020rrk}, which contains the Postnikov class as a $k$-invariant), this suggests that $S_{3d}$ hosts no perturbative anomalies on 3-manifolds. This is consistent with our result found in \S \ref{bordinv} that, in the aligned case, the surface holonomies are bordism invariants.
\end{remark}

All in all, the conserved Noether charges in our theory $S_{3d}$ form centrally-extended Lie 2-algebras $\hat\D,\hat{\tilde\D}$. These algebras are parameterized by the symmetries $\L_\text{sym}\subset C^\infty(Y)\oplus \Omega^1(Y)$ of $S_{3d}$, which are function spaces of infinite-dimension over $\mathbb{K}$. The infinite dimensionality of the symmetry algebra of our three-dimensional theory thus leads to integrability. This is of course in direct analogy with the affine Lie algebra underlying the charges in the 2d WZW model \cite{KNIZHNIK198483}. 

We conclude our paper by briefly mentioning that one of the authors (HC) has developed a notion of {\it 2-Kac-Moody algebra} in \cite{Chen:2023integrable}, which governs the global symmetries of the model \eqref{covariant3daction}. This is in analogy to the role that the affine Kac-Moody algebra $\widehat{\Omega_k\g}$ plays in the 2d WZW model \cite{book-integrable,KNIZHNIK198483}, whose corresponding Kac-Moody group $\widehat{\Omega_k G}$ \cite{Baez:2005sn} is crucial for quantizing the WZW model. It is known that the extension class $k$ (i.e. the level) defining the Kac-Moody group $\widehat{\Omega_kG}$ comes from the $S^1$-transgression map\footnote{Geometrically, the transgression map is a way of associating bundle gerbes --- a higher notion of bundles --- on a space to a line bundle on its loop space; see \cite{Willerton:2008gyk}.} \cite{Carey_1997,Sharpe:2015mja,Waldorf:2012,Waldorf2015TransgressiveLG}
\begin{equation*}
    H^3(G,\bbC^\times)\rightarrow H^2(\Omega G,\bbC^\times)\,.
\end{equation*}
Therefore, one expects that our model \eqref{covariant3daction} may admit an analogous description in terms of a "surface transgression" from a certain Lie 2-group cohomology class (cf. \cite{Ginot:2009pil,Angulo:2018}). This shall be a crucial step in quantizing our 3d model, and we leave this for a future endeavour.

\section{Outlook}\label{outlook}
In this paper, we have constructed a 3d integrable field theory from a higher-gauge theory by following through with the Costello-Yamazaki localization procedure \cite{Costello:2019tri}. We also studied in detail the properties of this 3d field theory, and proved several very interesting facts about its properties. We believe that there are many more open questions to be answered in regards to this theory, and the connection between higher-gauge theories and higher-dimensional integrable field theories in general. Here, we list a few of them in this section.

\paragraph{Higher Monodromies.} Recall that the conserved 2-monodromy matrices \eqref{conservedcharges} in our 3d field theory come labelled by the categorical characters $\mathcal{X}_k$ of the Lie 2-group $\mathbb{G}$ and homotopy classes $[\Sigma]$ of surfaces in $Y$. We are lead to two very interesting problems to tackle when considering the properties of these higher monodromies.
\begin{enumerate}
    \item \textbf{The label $[\Sigma]$: differential gerbes and 2-tangle invariants.} In 3d TQFTs, it is known that the algebra of the Wilson lines define invariants of framed links \cite{Witten:1988hc}. The usual singular homology is not enough to capture such data: indeed, framed link invariants satisfy certain knot and satellite relations that are seen only by differential --- possibly Dolbeault, when the chirality and foliation are misaligned --- cohomology \cite{FreedMonopole} (eg. the \textit{Deligne-Beilinson} cohomology \cite{Guadagnini2008DeligneBeilinsonCA} in Abelian Chern-Simons theory). As such, the label $[\Sigma]$ should be treated as differential non-Abelian gerbes \cite{Nikolaus2011FOUREV} which classify invariants of framed 2-tangles \cite{Baez:1995xq,BAEZ1998145} in $Y$ (or its thickening $Y\times [0,1]$).
    
    \item \textbf{The label $k$: orthogonality of categorical characters.} In regular representation theory, one can prove that there are countably many orthogonal irreducible characters of a compact Lie group that span the representation ring \cite{Woronowicz1988}. Correspondingly, the label $k$ should run over all indecomposable $\mathbb{G}$-module categories, but categorical character theory is not yet powerful enough to prove this (see eg. \S 9.6 in \cite{Bartlett:2009PhD} and {\bf Proposition 1.2.19} in \cite{Douglas:2018}). Moreover, while an inner product is known to exist for categorical characters \cite{Huang:2024}, there is no satisfactory notion of "orthogonality" for them at the moment. One of the authors (HC) of this paper is currently investigating this problem.
\end{enumerate}
Resolving these two issues are of major interest in recent literature. For instance, higher-tangle invariants would help define 4d non-semisimple TQFTs (which have a chance of definition exotic 4-manifold invariants \cite{Reutter:2020bav}), and categorical character theory would have very explicit applications in the construction of 4d gapped phases in condensed matter theory \cite{Bullivant:2019tbp,Delcamp:2023kew}.

\paragraph{Quantization of $S_{3d}$: the 2-Kac-Moody group.} An obvious next step to take would be to try and quantize the theory $S_{3d}$. Towards this, one of the authors (HC) has devised in \cite{Chen:2023integrable} a notion of "2-Kac-Moody algebra" $\widehat{\Sigma_s\G}$ which underlies the algebraic properties of the global Noether {charges} of the theory, where $s$ is the extension 2-cocycle \eqref{2cocycle}. The higher monodromy matrices \eqref{conservedcharges} would then be controlled by the integrated 2-Kac-Moody group $\widehat{\Sigma_s\mathbb{G}}$. As $S_{3d}$ is a topological-holomorphic field theory, it would be interesting to understand the structure of this object $\widehat{\Sigma_s\mathbb{G}}$ and its relation to the raviolo vertex operator algebra \cite{Garner:2023zqn,Alfonsi:2024qdr}.

The quantum Hilbert space associated to $S_{3d}$ are then given by unitary (categorical) representations $\operatorname{UMod}_{\mathsf{2Vect}}(\widehat{\Sigma_s\mathbb{G}})$ of this 2-Kac-Moody group. In analogy with the modules of vertex operator algebras in general \cite{KazhdanLusztig:1994}, this 2-category should have equipped higher modular and braided tensor structures, which are of major interest in both pure categorical algebra \cite{Johnson_Freyd_2023} as well as the construction of higher-skein "lasagna modules" in 4d TQFTs \cite{Morrison2019InvariantsO4,Manolescu2022SkeinLM}.

\paragraph{Other choices of disorder operators} In this article we have considered a unique choice of disorder operator $\omega = z^{-1}\rd z$, which led to a rich family of three-dimensional theories satisfying properties which are higher homotopical analogues of those satisfied by the WZW model in two dimensions. Notably, in \cite{Schenkel:2024dcd} they considered a variety of disorder operators, except for the one hereby considered. Their construction led to different $3$d theories which have a Lax connection with spectral parameter. It would be very interesting to perform a deep analysis of the holonomies, focusing on the invariance under homotopies relative boundary for the Lax connections constructed there. Morever, it would be appealing to analyze the symmetries of the theories obtained in \cite{Schenkel:2024dcd}, and construct the corresponding current algebras.

\newpage

\appendix

\section{Classification of Lie 2-groups and Lie 2-algebras}\label{classification}
Let us begin with a brief overview of the notion of 2-groups, Lie 2-algebras and 2-gauge theory. Recall a 2-group $\mathbb{G}$ is, by definition, a monoidal groupoid $\Gamma\rightrightarrows G$ in which all objects are invertible. There are many different but equivalent characterizations of a 2-group, namely as a pointed 2-groupoid, as a group object internal to the category of groups, or as a 2-term crossed-complex of groups as in {\bf Definition \ref{2grpdef}}. 

The fundamental theorem of Lie 2-groups state that there is a one-to-one correspondence between connected, simply-connected Lie 2-groups and Lie 2-algebras. {\bf Definition \ref{lie2alg}} gives Lie 2-algebras as a differential crossed-complex \cite{Baez:2003fs}, which is equivalent to the $L_2$-algebra definition by simply identifying $\mu_1=t$ and $\mu_2 = (\rhd,[-,-])$. The homotopy Jacobi identities are equivalent to \eqref{pfeif1} and the 2-Jacobi identities. 

The goal of this section is to briefly recall some classification results for Lie 2-groups and Lie 2-algebras. If $\mathbb{G},\mathbb{G}'$ are two Lie 2-groups, a {\it Lie 2-group homomorphism} $(\Phi,\Psi): \mathbb{G}\rightarrow\mathbb{G}'$ is a tuple of Lie group maps $\Phi:G\rightarrow G',\Psi:\mathsf{H}\rightarrow\mathsf{H}'$ such that
\begin{enumerate}
    \item $\bar\mu_1'\circ \Phi = \Psi\circ\bar\mu_1$, and
    \item $\Phi(x\rhd y) = (\Psi x)\rhd' (\Phi y)$ for all $x\in G,y\in\mathsf{H}$.
\end{enumerate}
Similarly, If $\mathfrak{G},\mathfrak{G}'$ are two Lie 2-algebras, then a {Lie 2-algebra homomorphism} $(\phi,\psi):\mathfrak{G}\rightarrow\mathfrak{G}'$ is a tuple of Lie algebra maps $\phi: \h\rightarrow \h',\psi: \g\rightarrow \g'$ such that
\begin{enumerate}
    \item $\mu_1'\circ \phi = \psi \circ \mu_1$, and
    \item $\phi(\sfx\rhd \sfy) = (\psi \sfx)\rhd'(\phi \sfy)$ for each $\sfx\in \g,\sfy\in \h$.
\end{enumerate}
As in the Lie 1-algebra case, this notion transports to maps between principal 2-bundles $\mathcal{P}\rightarrow\mathcal{P}$ \cite{Baez:2004in,Chen:2022hct}.

In contrast to Lie groups, the notion of "isomorphism" for Lie groupoids is categorical equivalence \cite{maclane:71}. A bit more explicitly, we say $\mathbb{G}\simeq\mathbb{G}'$ are equivalent iff there exist Lie 2-group maps $(\Phi,\Psi):\mathbb{G}\rightarrow\mathbb{G}'$ and $(\Phi,\Psi)^{-1}:\mathbb{G}'\rightarrow\mathbb{G}$ such that there are invertible natural transformations 
\begin{equation*}
    \mathcal{H}:(\Phi,\Psi)^{-1}\circ (\Phi,\Psi)\Rightarrow\id_{\mathbb{G}}, \qquad \mathcal{E}: (\Phi,\Psi)\circ (\Phi,\Psi)^{-1}\Rightarrow \id_{\mathbb{G}'}.
\end{equation*}
The following is a classic result first proven by Ho{\`a}ng Xu{\^ a}n S{\'i}nh \cite{Nguyen2014CROSSEDMA,Ang2018,Kapustin:2013uxa,baez2023ho}, who is a Vietnamese mathematician taught by Alexander Grothendieck.
\begin{theorem}
    2-groups $\mathbb{G}$ are classified by the following data: a group $\Pi_1=\operatorname{coker}\bar\mu_1$, a $\Pi_1$-module $\Pi_2=\operatorname{ker}\bar\mu_1$, and a degree-3 group cohomology class $\tau\in H^3(\Pi_1,\Pi_0)$ called the \textbf{Postnikov class}.
\end{theorem}
\noindent Note $\Pi_2\subset\mathsf{H}$ must be an Abelian group by the Peiffer identity \eqref{pfeif1}. The tuple $(\Pi_1,\Pi_2,\tau)$ is often called the \textit{Ho{\`a}ng data} of $\mathbb{G}$.

An analogous result holds for Lie 2-algebras. Here, one can state equivalence in the context of chain complexes. We say $\G\simeq\G'$ are equivalent iff there exist Lie 2-algebra maps $(\phi,\psi):\G\rightarrow \G',(\phi,\psi)^{-1}:\G'\rightarrow \G$ and invertible chain homotopies
\begin{equation*}
    \eta: (\phi,\psi)^{-1}\circ(\phi,\psi)\Rightarrow \id_\G,\qquad \varepsilon:(\phi,\psi)\circ (\phi,\psi)^{-1}\Rightarrow\id_{\G'}.
\end{equation*}
The following is also a classic result, which was claimed to be first proven by Gerstenhaber by Wagemann in \cite{Wagemann+2021}.
\begin{theorem}
Lie 2-algebras are classified up to equivalence by the following data: a Lie algebra $\mathfrak{n}=\operatorname{coker}\mu_1$, a $\mathfrak{n}$-module $V=\operatorname{ker}\mu_1$ and a degree-3 Lie algebra cohomology class $\kappa\in H^3(\mathfrak{n},V)$, called the {\bf Postnikov class}.
\end{theorem}
\noindent One can show that, if two principal 2-bundles $\mathcal{P},\mathcal{P}'\rightarrow X$ have structure 2-groups $\mathbb{G},\mathbb{G}'$ that have distinct Postnikov classes, then they cannot be isomorphic 2-bundles over $X$.

\section{Homotopy transfer}\label{sec:homtransf}
The BV-BRST formalism is a way to organize the algebraic and geometric data in a field theory with gauge symmetry into cochain complexes, which are designed such that their cohomologies describes precisely the gauge-invariant on-shell quantities. The mathematical structure that arises out of this construction is an {\it $L_\infty$-algebra}, which can be described as a $\mathbb{Z}$-graded differential complex 
$$\mathfrak{L}\cong \dots \rightarrow \mathfrak{l}_{-2}\xrightarrow{\ell_1}\l_{-1}\xrightarrow{\mu_1}\l_0\xrightarrow{\ell_1}\l_1\xrightarrow{\ell_1}\l_2\rightarrow \dots$$ equipped with higher-brackets $\ell_n\in \operatorname{Hom}^{2-n}(\L^{n\otimes},\L)$ for $n\geq 2$ satisfying the Koszul-Jacobi identities. In gauge field theories with structure (2-)group $\mathbb{G}$, such as the homotopy Maurer-Cartan theories \cite{Jurco:2018sby}, for instance, this complex is given by $\L = \Omega^\bullet(X)\otimes \G$ where $\G=\operatorname{Lie}\mathbb{G}$.

Now from a purely mathematical perspective, it is known that $L_\infty$-algebras retracts onto its cohomology \cite{stasheff2018,Alfonsi_2023}.\footnote{$L_\infty$-algebras are objects in the derived category of Lie algebras. More generally in any derived Abelian category, equivalences are precisely quasi-isomorphisms, ie. derived chain maps that induce isomorphisms in cohomologies.} In other words, given a $L_\infty$-algebra $\L$, the inclusion $\iota:H(\L)\hookrightarrow \L$ treating $H^\bullet(\L)\cong\operatorname{ker}\ell_1/\operatorname{im}\ell_1$ as a subcomplex of $\L$ admits a homotopy inverse (ie. an adjunction) given by the projection $p:\L\rightarrow H(\L)$ --- there exists {chain homotopies} $\eta: \iota\circ p\Rightarrow \id_{\G}$, $\varepsilon: p\circ \iota\Rightarrow \id_{H(\G)}$ that "witness" this retraction. To be a chain homotopy, the map $\eta=(\eta_n)$ --- where $\eta_n: \L_n\rightarrow \L_{n-1}$ --- for instance must satisfy
\begin{equation*}
    \iota \circ p - 1 = \eta \circ \ell_1 + \ell_1\circ \eta 
\end{equation*}
as linear maps on $\L$.

In this context, it can be shown that a $L_\infty$-algebra structure, ie. the higher brackets $\tilde\ell=(\tilde\ell_n)_{n\geq 2}$, on $H(\L)$ can be induced from that $\ell=(\ell_n)_{n\geq 2}$ on $\L$ by using the projector $p$; notice, by construction, the cohomology $H(\L)$ has trivial differential. This is given by the formula
\begin{equation}\label{homotopytransfer}
    \tilde\ell = p\circ (\ell + \ell\circ \eta\circ \ell + \ell \circ \eta\circ \ell\circ \eta\circ \ell + \dots ) \circ \iota,
\end{equation}
and it can be proven that $\tilde\ell$ indeed satisfies the Koszul-Jacobi identities on $H(\L)$ \cite{stasheff2018,Arvanitakis:2020rrk}. This is the celebrated \textbf{homotopy transfer theorem} for $L_\infty$-algebras: the $L_\infty$-algebra structure on $\L$ is transferred to a $L_\infty$-algebra structure on its cohomology $H(\L)$ through the projection quasi-isomorphism $p$.

An interesting case occurs when $\L$ is only a dgla, such that $\ell_n=0$ for all $n>2$. The formula \eqref{homotopytransfer} implies that homotopy transfer in fact induces higher homotopy brackets on $H(\L)$. Indeed, at degree-3, we see that we have a generally non-trivial contribution given for each $\sfx_1,\sfx_2,\sfx_3\in\L$ by
\begin{equation*}
    \tilde\ell_3([\sfx_1],[\sfx_2],[\sfx_3]) = ~\circlearrowright  [\ell_2(\eta(\ell_2([\sfx_1],[\sfx_2])),\eta([\sfx_3]))],
\end{equation*}
where $\circlearrowright$ denotes a cyclic sum and we have denoted the cohomology class of an element $\sfx\in\L$ by $p(\sfx)=[\sfx]\in H(\L)$. This map $\tilde\ell_3$ is usually called the \textbf{homotopy map} in the literature \cite{Chen:2012gz}, and can be seen to measure the failure of the chain homotopy $\eta$ to respect the graded Jacobi identity. This is relevant for us, as given a strict Lie 2-algebra $\G$, the associated BV-BRST $L_\infty$-algebra $\L = \Omega^\bullet(X)\otimes \G$ is indeed a dgla with $\ell_1=d - \mu_1$. However, upon performing a homotopy transfer, we see that a degree-3 homotopy map is induced on the cohomology $H(\L)$. 

This idea can be used to immediately to prove Gerstenhaber's theorem in \S \ref{classification}. The cohomology $H(\G)$ of a Lie 2-algebra $\G$ consist of $M\xrightarrow{0}\mathfrak{n}$, and the homotopy map $\tilde\ell_3$ is a 3-cocycle representing the Postnikov class $\kappa\in H^3(\mathfrak{n},M)$.

\printbibliography

\end{document}